\newif\ifllncs

\documentclass[letterpaper,11pt]{article} 
\usepackage[letterpaper,hmargin=1in,vmargin=1in]{geometry} 
\usepackage{fullpage}
\usepackage[pdfstartview=FitH,colorlinks,linkcolor=blue,citecolor=blue]{hyperref}

\usepackage[T1]{fontenc}
\usepackage{color} 
\usepackage{times} 
\usepackage{microtype,pdfsync} 
\usepackage{amsmath,amsfonts,amssymb}
\usepackage{amsthm} 
\usepackage{url} 
\usepackage{xspace,paralist}
\usepackage{lmodern}
\usepackage{float}
\floatstyle{boxed}
\restylefloat{figure}
\usepackage{multirow}
\usepackage{makecell}
\usepackage{ stmaryrd }

\usepackage{marvosym}
\usepackage{graphicx}
\usepackage{graphbox}
\usepackage{subcaption}
\usepackage{ bbold }
\usepackage{breakcites}
\usepackage{mathtools}

\usepackage{xfrac}

\usepackage{pifont}
\usepackage[most]{tcolorbox} 

\theoremstyle{plain}

\newtheorem{theorem}{Theorem}
\newtheorem{lemma}[theorem]{Lemma} 
\newtheorem{corollary}[theorem]{Corollary}
\newtheorem{claim}[theorem]{Claim}

\theoremstyle{definition}

\newtheorem{definition}{Definition}
\newtheorem{remark}{Remark} 

\newtheorem{assumption}{Assumption}
\newtheorem{conjecture}[theorem]{Conjecture}
\newtheorem{algorithm}[theorem]{Algorithm}
\newtheorem{problem}[theorem]{Problem}

\newcommand{\F}{{\mathbb F}} 
 
\newcommand{\Z}{\mathbb{Z}} 
\newcommand{\N}{\mathbb{N}} 
\newcommand{\R}{\mathbb{R}}

\newcommand{\bv}{{\mathbf{b}}}
\newcommand{\cv}{{\mathbf{c}}}

\newcommand{\ev}{{\mathbf{e}}}

\newcommand{\pv}{{\mathbf{p}}}
\newcommand{\rv}{{\mathbf{r}}}
\newcommand{\sv}{{\mathbf{s}}}
\newcommand{\tv}{{\mathbf{t}}}
\newcommand{\uv}{{\mathbf{u}}}
\newcommand{\vv}{{\mathbf{v}}}
\newcommand{\wv}{{\mathbf{w}}}
\newcommand{\xv}{{\mathbf{x}}}
\newcommand{\yv}{{\mathbf{y}}}
\newcommand{\zv}{{\mathbf{z}}}

\newcommand{\onev}{{\boldsymbol{1}}}

\newcommand{\Id}{{\mathbf{I}}}
\newcommand{\Am}{{\mathbf{A}}}
\newcommand{\Bm}{{\mathbf{B}}}
\newcommand{\Cm}{{\mathbf{C}}}

\newcommand{\Gm}{{\mathbf{G}}}

\newcommand{\Mm}{{\mathbf{M}}}

\newcommand{\Rm}{{\mathbf{R}}}
\newcommand{\Sm}{{\mathbf{S}}}
\newcommand{\Tm}{{\mathbf{T}}}
\newcommand{\Um}{{\mathbf{U}}}
\newcommand{\Wm}{{\mathbf{W}}}

\newcommand{\Bs}{{\mathcal{B}}}

\newcommand{\Is}{{\mathcal{I}}}
\newcommand{\Js}{{\mathcal{J}}}

\newcommand{\Ls}{{\mathcal{L}}}

\newcommand{\Rs}{{\mathcal{R}}}
\newcommand{\Ss}{{\mathcal{S}}}
\newcommand{\Ts}{{\mathcal{T}}}
\newcommand{\Us}{{\mathcal{U}}}
\newcommand{\Vs}{{\mathcal{V}}}

\newcommand{\Xs}{{\mathcal{X}}}
\newcommand{\Ys}{{\mathcal{Y}}}

\newcommand{\negl}{{\sf negl}}

\newcommand{\adv}{{\sf A}}

\newcommand{\setupprotocol}[2]{%
	\expandafter\newcommand\csname protocol#1\endcsname{{\sf \Pi_{#2}}}%
	\expandafter\newcommand\csname gen#1\endcsname{{\sf Gen_{#2}}}%
	\expandafter\newcommand\csname enc#1\endcsname{{\sf Enc_{#2}}}%
	\expandafter\newcommand\csname dec#1\endcsname{{\sf Dec_{#2}}}%
	\expandafter\newcommand\csname sign#1\endcsname{{\sf Sign_{#2}}}%
	\expandafter\newcommand\csname ver#1\endcsname{{\sf Ver_{#2}}}%
	\expandafter\newcommand\csname prov#1\endcsname{{\sf Prov_{#2}}}%
	\expandafter\newcommand\csname obf#1\endcsname{{\sf Obf_{#2}}}%
	\expandafter\newcommand\csname genlossy#1\endcsname{{\sf GenLossy_{#2}}}%
	\expandafter\newcommand\csname eval#1\endcsname{{\sf Eval_{#2}}}%
	
	\expandafter\newcommand\csname msk#1\endcsname{{\sf msk_{#2}}}%
	\expandafter\newcommand\csname sk#1\endcsname{\sf sk_{#2}}%
	\expandafter\newcommand\csname pk#1\endcsname{\sf pk_{#2}}%
	\expandafter\newcommand\csname vk#1\endcsname{{\sf vk_{#2}}}%
}

\setupprotocol{}{}
\setupprotocol{pke}{{\sf PKE}}
\setupprotocol{ske}{{\sf SKE}}
\setupprotocol{cca}{{\sf CCA}}
\setupprotocol{ewh}{{\sf EwH}}
\setupprotocol{fo}{{\sf FO}}
\setupprotocol{fs}{{\sf FS}}
\setupprotocol{sig}{{\sf Sig}}
\setupprotocol{fhe}{{\sf fhe}}

\newcommand{\aux}{{\sf aux}}

\newcommand{\fe}{{\sf FE}}
\newcommand{\setup}{{\sf Setup}}
\newcommand{\pp}{{\sf pp}}

\newcommand{\keygen}{{\sf keygen}}
\newcommand{\puncture}{{\sf puncture}}
\newcommand{\ct}{{\sf ct}}
\newcommand{\adversary}{\mathcal{A}}

\newcommand{\rerand}{{\sf ReRand}}
\newcommand{\invertre}{{\sf ReRand^{-1}}}
\newcommand{\recoverrand}{{\sf RecoverR}}

\newcommand{\cO}{\mathcal{O}}

\newcommand{\st}{{\sf st}}

\def\E{\mathop{\mathbb E}}

\newcommand{\ignore}[1]{}

\definecolor{block-gray}{gray}{0.95}

\newtcolorbox{aquote}[2][]{%
	colback=block-gray,
	grow to right by=-10mm,
	grow to left by=-10mm, 
	boxrule=0pt,
	boxsep=0pt,
	arc=0pt,
	outer arc=0pt,
	breakable,
	enhanced jigsaw,
	borderline west={4pt}{0pt}{gray},
	title={#2\par},
	colbacktitle={block-gray},
	coltitle={black},
	fonttitle={\large\bfseries},
	attach title to upper={},
	#1,
}

\floatstyle{plain}
\restylefloat{figure}

\begin{document}

\title{%
Another Round of Breaking and Making Quantum Money: \\
\large How to Not Build It from Lattices, and More}
\author{
Jiahui Liu\footnote{ University of Texas at Austin. Email: jiahui@utexas.edu}
\and
Hart Montgomery\footnote{ Linux Foundation \& Fujitsu Research. Email: hart.montgomery@gmail.com}  \and
Mark Zhandry \footnote{NTT Research. Email: mzhandry@gmail.com}}
\date{}

\maketitle

\begin{abstract}
Public verification of quantum money has been one of the central objects in quantum cryptography ever since Wiesner's pioneering idea of using quantum mechanics to construct banknotes against counterfeiting. So far, we do not know any publicly-verifiable quantum money scheme that is provably secure from standard assumptions.

In this work, we provide both negative and positive results for publicly verifiable quantum money.
\begin{itemize}
	\item In the first part, we give a general theorem, showing that a certain natural class of quantum money schemes from lattices cannot be secure. We use this theorem to break the recent quantum money scheme of Khesin, Lu, and Shor.
	\item In the second part, we propose a framework for building quantum money and quantum lightning we call \emph{invariant money} which abstracts some of the ideas of quantum money from knots by Farhi et al.(ITCS'12). In addition to formalizing this framework, we provide concrete hard computational problems loosely inspired by classical knowledge-of-exponent assumptions, whose hardness would imply the security of \emph{quantum lightning}, a strengthening of quantum money where not even the bank can duplicate banknotes. 

    \item We discuss potential instantiations of our framework, including an oracle construction using cryptographic group actions and instantiations from rerandomizable functional encryption, isogenies over elliptic curves, and knots.
\end{itemize}
\end{abstract}


\section{Introduction}\label{sec:intro}

\subsection{Motivation}

Quantum information promises to revolutionize cryptography. In particular, the no cloning theorem of quantum mechanics opens the door to \emph{quantum cryptography}: cryptographic applications that are simply impossible classically. The progenitor of this field, due to Wiesner~\cite{Wiesner83}, is quantum money: quantum digital currency that cannot be counterfeited due to the laws of physics. Since Wiesner's proposal, many applications of quantum information to cryptography have been proposed, including quantum key distribution (QKD) ~\cite{BenBra84}, randomness expansion~\cite{Colbeck09,STOC:CouYue14,FOCS:BCMVV18}, quantum copy protection~\cite{CCC:Aaronson09,EC:AnaLap21,C:ALLZZ21,C:CLLZ21}, quantum one-time programs~\cite{C:BroGutSte13}, and much more.

Throughout the development of quantum cryptography, quantum money has remained a central object, at least implicitly. Indeed, the techniques used for quantum money are closely related to those used in other applications. For example, the first message in the BB84 quantum QKD protocol~\cite{BenBra84} is exactly a banknote in Wiesner's scheme. The techniques used by~\cite{FOCS:BCMVV18} to prove quantumness using classical communication have been used to construct quantum money with classical communication~\cite{RadSat19}. The subspace states used by~\cite{STOC:AarChr12} to construct quantum money were recently used to build quantum copy protection~\cite{C:ALLZZ21}.

\paragraph{The Public Verification Barrier.} Wiesner's scheme is only privately verifiable, meaning that the mint is needed to verify. This results in numerous weaknesses. Improper verification opens the scheme to active attacks~\cite{Lutomirski10}. Moreover, private verification is not scalable, as the mint would be required to participate in every single transaction.  Wiesner's scheme also requires essentially perfect quantum storage, since otherwise banknotes in Wiesner's scheme will quickly decohere and be lost. 

All these problems are readily solved with \emph{publicly verifiable} quantum money\footnote{Sometimes it is also referred to as public-key quantum money. We may use the two terms interchangeably.}, where anyone can verify, despite the mint being the sole entity that can mint notes. Public verification immediately eliminates active attacks, and solves the scaling problem since the transacting users can verify the money for themselves. Aaronson and Christiano~\cite{STOC:AarChr12} also explain that public verifiability allows for also correcting any decoherance, so users can keep their banknotes alive indefinitely. 

Unfortunately, constructing convincing publicly verifiable quantum money has become a notoriously hard open question.
Firstly, some natural modifications to Wiesner's quantum money scheme will not give security under public verification \cite{farhi2010quantum}.
Aaronson~\cite{CCC:Aaronson09}, and later Aaronson and Christiano~\cite{STOC:AarChr12} gave publicly verifiable quantum money relative to quantum and classical oracles, respectively. Such oracle constructions have the advantage of provable security, but it is often unclear how to instantiate them in the real world\footnote{Quantum oracles are quantum circuits accessible only as a black-box unitary. They are generally considered as strong relativizing tools when used in proofs. Classical oracles are black-box classical circuits, a much weaker tool.}: in both~\cite{CCC:Aaronson09} and~\cite{STOC:AarChr12}, ``candidate'' instantiations were proposed, but were later broken~\cite{ITCS:LAFGKH10,PDFHP19}. Another candidate by Zhandry~\cite{EC:Zhandry19b} was broken by Roberts~\cite{EC:Roberts21}. Other candidates have been proposed~\cite{ITCS:FGHLS12,Kane18,EPRINT:KanShaSil21}, but they all rely on new, untested assumptions that have received little cryptanalysis effort. The one exception, suggested by~\cite{BenSat16} and proved by~\cite{EC:Zhandry19b}, uses indistinguishability obfuscation (iO) to instantiate Aaronson and Christiano's scheme~\cite{STOC:AarChr12}. Unfortunately, the post-quantum security of iO remains poorly understood, with all known constructions of post-quantum iO~\cite{TCC:GenGorHal15,TCC:BGMZ18,EPRINT:BDGM20b,EC:WeeWic21} being best labeled as candidates, lacking justification under widely studied assumptions.

Thus, it remains a major open question to construct publicly verifiable quantum money from standard cryptographic tools. Two such post-quantum tools we will investigate in this work are the two most influential and well-studied: lattices and isogenies over elliptic curves.

This public verification barrier is inherited by many proposed applications of quantum cryptography. For example, quantum copy protection for any function whose outputs can be verified immediately implies a publicly verifiable quantum money scheme. As such, all such constructions in the standard model~\cite{C:ALLZZ21,C:CLLZ21} require at a minimum a computational assumption that implies quantum money.\footnote{This holds true even for certain weaker versions such as copy \emph{detection}, also known as infinite term secure software leasing.}

\paragraph{Quantum Money Decentralized: Quantum Lightning}

An even more ambitious goal is a publicly verifiable quantum money 
where the bank/mint itself should \emph{not} be capable of duplicating money states. To guarantee unclonability, the scheme should have a "collision-resistant" flavor: no one can (efficiently) generate two valid money states with the same serial number.
This notion of quantum money appeared as early in \cite{ITCS:LAFGKH10}; the name "quantum lightning"  was given in \cite{EC:Zhandry19b}. 

Quantum lightning has broader and more exciting applications: as discussed in \cite{EC:Zhandry19b, coladangelo2019smart, coladangelo2020quantum,amos2020one}, it can be leveraged as verifiable min-entropy, useful building blocks to enhance blockchain/smart contract protocols and moreover, it could lead to decentralized cryptocurrency without a blockchain. 

 Quantum money has a provably secure construction from iO, a strong cryptographic hammer but still a widely used assumption. On the other hand, quantum lightning from even \emph{relatively standard-looking} assumptions remains open. Some existing constructions \cite{Kane18, EPRINT:KanShaSil21} use strong oracles such as quantum oracles, with conjectured instantiations that did not go through too much cryptanalysis.
 \cite{ITCS:FGHLS12} is another candidate built from conjectures in knot theory. But a correctness proof and security reduction are not provided in their paper.
 
 \paragraph{Collapsing vs. Non-Collapsing}
 With a close relationship to quantum money, collapsing functions \cite{EC:Unruh16} are a central concept in quantum cryptography. 
 A collapsing function $f$ says that one should not be able to distinguish a superposition of pre-images $\frac{|x_1\rangle + |x_2\rangle \cdots |x_k\rangle}{\sqrt{k}}$, from a measured pre-image $|x_i\rangle, i\in [k]$ for some image $y = f(x_i)$, for all $i \in [k]$.
 
 While collapsing functions give rise to secure post-quantum cryptography like commitment schemes, its precise opposite is necessary for quantum money: if no verification can distinguish a money state in a superposition of many supports from its measured state, a simple forgery comes ahead.
Hence, investigating the collapsing/non-collapsing properties of hash functions from lattices and isogenies will provide a win-win insight into quantum money and post-quantum security of existing cryptographic primitives.

\section{Our Results}

In this work, we give both negative and positive results for publicly verifiable quantum money. 

\paragraph{Breaking Quantum Money.} Very recent work by Khesin, Lu, and Shor~\cite{KLS22} claims to construct publicly verifiable quantum money from the hardness of worst-case lattice problems, a standard assumption. Our first contribution is to identify a fatal flaw in their security proof, and moreover show how to exploit this flaw to forge unlimited money. After communicating this flaw and attack, the authors of~\cite{KLS22} have retracted their paper.\footnote{We thank the authors of~\cite{KLS22} for patiently answering our numerous questions about their work, which was instrumental in helping us identify the flaw.}

More importantly, we show that a general class of \emph{natural} money schemes based on lattices \emph{cannot} be both secure and publicly verifiable. We consider protocols where the public key is a short wide matrix $\Am^T$, and a banknote with serial number $\uv$ is a superposition of ``short'' vectors $\yv$ such that $\Am^T\cdot\yv=\uv\bmod q$. Our attack works whenever $\Am^T$ is uniformly random. We also generalize this to handle the case where $\Am^T$ is uniform conditioned on having a few public short vectors in its kernel. This generalization includes the Khesin-Lu-Shor scheme as a special case. Our result provides a significant barrier to constructing quantum money from lattices.

Along the way, we prove that the SIS hash function is \emph{collapsing}~\cite{EC:Unruh16} for all moduli, resolving an important open question in the security of post-quantum hash functions.\footnote{Previously,~\cite{C:LiuZha19} showed that SIS was collapsing for a super-polynomial modulus.}

\paragraph{Invariant Money/Lightning.} To complement our negative result, we propose a new framework for building quantum money, based on invariants. Our framework abstracts some of the ideas behind the candidate quantum money from knots in~\cite{ITCS:FGHLS12} and behind \cite{ITCS:LAFGKH10}. Our main contributions here are two-fold:
\begin{itemize}
	\item We propose a (classical) oracle construction that implements our framework assuming the existence of a quantum-secure cryptographic group action and a relatively modest assumption about \emph{generic} cryptographic group actions.  We then give proposals for instantiating our invariant framework on more concrete assumptions. The first is based on isogenies over elliptic curves\footnote{The recent attacks   \cite{castryck2022efficient,maino2022attack,robert2022breaking} on SIDH do not apply to the isogeny building blocks we need. We will elaborate in the \ref{sec:ellipticcurve} and \ref{prelim:groupaction}}; the second is based on rerandomizable functional encryption with certain properties; 
	finally, we also discuss the quantum money from knots construction in~\cite{ITCS:FGHLS12} with some modifications.
	
	\item In order to gain confidence in our proposals, we for the first time formalize abstract properties of the invariant money under which security can be proved. Concretely, we prove that a certain mixing condition is sufficient to characterize the states accepted by the verifier, and in particular prove correctness\footnote{\cite{ITCS:FGHLS12} did not analyze correctness of their knot-based proposal, nor analyze the states accepted by their verifier and formalize the property needed for a security proof. \cite{ITCS:LAFGKH10} had informal correctness analysis on their proposal, but also did not analyze the security property needed.}. 
	We also propose ``knowledge of path'' security properties for abstract invariant structures which would be sufficient to justify security. These knowledge of path assumptions are analogs of the ``knowledge of exponent'' assumption on groups proposed by Damg\r{a}rd~\cite{C:Damgaard91}. Under these assumptions, we are even able to show that the invariants give quantum \emph{lightning}~\cite{ITCS:LAFGKH10,EC:Zhandry19b}, the aforementioned strengthening of quantum money that is known to have additional applications.
	
	Note that the knowledge of exponent assumption in groups is quantumly broken on groups due to the discrete logarithm being easy. However, for many of our assumptions, which are at least conjectured to be quantum-secure, the analogous knowledge of path assumption appears plausible, though certainly more cryptanalysis is needed to gain confidence. The main advantage of our proposed knowledge of path assumption is that it provides a concrete cryptographic property that cryptographers can study and analyze with a well-studied classical analog.
\end{itemize}


\section{Technical Overview}

\subsection{How to Not Build Quantum Money from Lattices} \label{sec:overviewlattice}

We first describe a natural attempt to construct quantum money from lattices, which was folklore but first outlined by Zhandry~\cite{EC:Zhandry19b}. The public key will contain a random tall matrix $\Am\in\Z_q^{m \times n},m\gg n$. To mint a banknote, first generate a superposition $|\psi\rangle=\sum_{\yv}\alpha_\yv|\yv\rangle$ of short vectors $\yv\in\Z^m$, such that $|\yv|\ll q$. A natural $|\psi\rangle$ is the discrete-Gaussian-weighted state, where $\alpha_\yv\propto \sqrt{e^{-\pi |\yv|^2/\sigma^2}}$ for a width parameter $\sigma$. Then compute in superposition and measure the output of the map $\yv\mapsto \Am^T\cdot\yv\bmod q$, obtaining $\uv\in\Z_q^n$. The state collapses to:
\[|\psi_\uv\rangle\propto\sum_{\yv:\Am^T\cdot\yv=\uv}\alpha_\yv|\yv\rangle\enspace .\]
This will be the money state, and $\uv$ will be the serial number. This state can presumably not be copied: if one could construct two copies of $|\psi_\uv\rangle$, then one could measure both, obtaining two short vectors $\yv,\yv'$ with the same coset $\uv$. As $|\psi_\uv\rangle$ is a superposition of many vectors (since $m\gg n$), with high probability $\yv\neq \yv'$. Subtracting gives a short vector $\yv-\yv'$ such that $\Am^T\cdot(\yv-\yv')=0$,  solving the Short Integer Solution (SIS) problem. SIS is presumably hard, and this hardness can be justified based on the hardness of worst-case lattice problems such as the approximate Shortest Vector Problem (SVP).

The challenge is: how to verify $|\psi_\uv\rangle$? Certainly, one can verify that the support of a state is only short vectors $\yv$ such that $\Am^T\cdot\yv=\uv$. But this alone is not enough: one can fool such a verification by any \emph{classical} $\yv$ in the support of $|\psi_\uv\rangle$. To forge then, an adversary simply measures $|\psi_\uv\rangle$ to obtain $\yv$, and then copies $\yv$ as many times as it likes.

To get the scheme to work, then, one needs a verifier that can distinguish classical $\yv$ from superpositions. This is a typical challenge in designing publicly verifiable money schemes. A typical approach is to perform the quantum Fourier transform (QFT): the QFT of $\yv$ will result in a uniform string, whereas the QFT of $|\psi_\uv\rangle$ will presumably have structure. Indeed, if $|\psi_{\uv}\rangle$ is the Gaussian superposition, following ideas of Regev~\cite{STOC:Regev05}, the QFT of $|\psi_\uv\rangle$ will be statistically close to a superposition of samples $\Am\cdot \rv+\ev$, where $\rv$ is uniform in $\Z_q^n$, and $\ev\in\Z_q^m$ is another discrete Gaussian of width $q/\sigma$. The goal then is to distinguish such samples from uniform.

Unfortunately, such distinguishing is likely hard, as this task is the famous (decisional) Learning with Errors (LWE) problem. LWE is presumably hard, which can be justified based on the hardness of the same worst-case lattice problems as with SIS, namely SVP. So either LWE is hard, or the quantum money scheme is insecure in the first place.

Nevertheless, this leaves open a number of possible strategies for designing quantum money from lattices, including:
\begin{enumerate}
	\item\label{enum:nongauss} What if non-Gaussian $|\psi\rangle$ is chosen?
	\item\label{enum:otherdist} What if distinguishing is not done via the QFT but some other quantum process?
	\item\label{enum:easylwe} What if we somehow make LWE easy?
\end{enumerate}

The first significant barrier beyond the hardness of LWE is due to Liu and Zhandry~\cite{C:LiuZha19}. They show that, if the modulus $q$ is super-polynomial, then the map $\yv\mapsto\Am^T\cdot\yv$ for a random $\Am$ is \emph{collapsing}~\cite{EC:Unruh16}: that is, for \emph{any} starting state $|\psi_{\uv} \rangle$ of short vectors, distinguishing $|\psi_\uv\rangle$ from $\yv$ is infeasible for \emph{any} efficient verification process. Collapsing is the preferred notion of post-quantum security for hash functions, as it is known that collision resistance is often not sufficient for applications when quantum adversaries are considered.

The result of~\cite{C:LiuZha19} follows from the hardness of LWE (which is quantumly equivalent to SIS~\cite{STOC:Regev05}), albeit with a noise rate super-polynomially smaller than $q/\sigma$ which is a stronger assumption than the hardness with rate $q/\sigma$. Moreover, their result requires $q$ to be super-polynomially larger than $\sigma$. In practice, one usually wants $q$ to be polynomial, and the result of~\cite{C:LiuZha19} leaves open the possibility of building quantum money in such a setting.

What about making LWE easy (while SIS remains hard)? The usual approach in the lattice literature to making decisional LWE easy is to output a short vector $\sv$ in the kernel of $\Am^T$. If $|\sv|\ll (q/\sigma)$, this allows for distinguishing LWE samples from uniform, since $\sv\cdot (\Am\cdot \rv+\ev)=\sv\cdot\ev$, which will be small relative to $q$, while $\sv\cdot\xv$ for uniform $\xv$ will be uniform in $\Z_q$. Unfortunately, adding such short vectors breaks the security proof, since $\sv$ is a SIS solution, solving SIS is trivially easy by outputting $\sv$. To revive the security, one can try reducing to the 1-SIS problem, which is to find a short SIS solution that is linearly independent of $\sv$. 1-SIS can be proved hard based on the same worst-case lattice problems as SIS~\cite{PKC:BonFre11}. However,  in the scheme above, it is not clear if measuring two forgeries and taking the difference should result in a vector linearly independent of $\sv$.

\paragraph{The Recent Work of~\cite{KLS22}.} Very recently, Khesin, Lu, and Shor~\cite{KLS22} attempt to provide a quantum money scheme based on lattices. Their scheme has some similarities to the blueprint discussed above, taking advantage of each of the strategies~\ref{enum:nongauss},~\ref{enum:otherdist} and~\ref{enum:easylwe}. But there are other differences as well: the state $|\psi\rangle$ is created as a superposition over a lattice rather than the integers, and the measurement of $\uv$ is replaced with a move complex general positive operator-value measurement (POVM).~\cite{KLS22} claims to prove security under the hardness of finding a second short vector in a random lattice when already given a short vector. This problem is closely related to 1-SIS, and follows also from the hardness of worst-case lattice problems.

\paragraph{Our Results.} First, we show an alternative view of~\cite{KLS22} which shows that it does, indeed, fall in the above framework. That is, there is a way to view their scheme as starting from $|\psi\rangle$ that is a non-Gaussian superposition of short integer vectors $\yv$. The minting process in our alternate view then measures $\Am^T\cdot\yv$, where $\Am$ is part of the public key, and is chosen to be uniform except that it is orthogonal to 3 short vectors $\sv_0,\sv_1,\sv_2$. These vectors play a role in verification, as they make the QFT non-uniform. Using this alternative view, we also demonstrate a flaw in the security proof of~\cite{KLS22}, showing that forged money states actually do not yield new short vectors in the lattice. See Section~\ref{sec:flaw} for details. 

We then go on (Section~\ref{sec:attack}) to show an explicit attack against their money scheme. More generally, we show an attack on a wide class of instantiations of the above framework. Our attack works in two steps:
\begin{itemize}
	\item First, we extend the collapsing result of~\cite{C:LiuZha19} to also handle the case of polynomial modulus, and in particular, we only need LWE to be hard for noise rate that is slightly smaller than $q/\sigma$. This resolves an important open by showing that SIS is collapsing for all moduli.
	
	Our proof requires a novel reduction that exploits a more delicate analysis of the quantum states produced in the proof of~\cite{C:LiuZha19}. We also extend the result in a meaningful way to the case where several short kernel vectors $\sv_0,\sv_1,\dots$ are provided. We show that instead of just using $\yv$ as a forgery (which can be distinguished using the short vectors $\sv_i$), a particular superposition over vectors of the form $\yv+\sum_i c_i \sv_i$ can fool any efficient verification. Fooling verification requires the hardness a certain ``$k$-LWE'' problem, which we show follows from worst-case lattice problems in many settings (see Section~\ref{sec:klwe}). This requires us to extend the known results on $k$-LWE hardness, which may be of independent interest.
	\item Then we show how to construct such a superposition efficiently given only $\yv$ and the $\sv_i$, in many natural settings. Our settings include as a special case the setting of~\cite{KLS22}. Along the way, we explain how to construct Gaussian superpositions over lattices, when given a short basis. The algorithm is a coherent version of the classical discrete Gaussian sampling algorithm~\cite{STOC:GenPeiVai08}. In general, it is not possible to take a classical distribution and run it on a superposition of random coins to get a superposition with weights determined by the distribution. This is because the random coins themselves will be left behind and entangled with the resulting state. We show how to implement the classical algorithm coherently in a way that does not leave the random coins behind or any other entangled bits. Such an algorithm was previously folklore (e.g. it was claimed to exist without justification by~\cite{KLS22}), but we take care to actually write out the algorithm.
\end{itemize}
After communicating this flaw and attack to the authors of~\cite{KLS22}, they have retracted their paper.\footnote{We once again want to emphasize that the authors of~\cite{KLS22} were exceptionally helpful and we thank them for their time spent helping us understand their work.}

\subsection{Quantum Money from Walkable Invariants}

In the second part of the paper, we describe a general framework for instantiating publicly verifiable quantum money from invariants satisfying certain conditions. This framework abstracts the ideas behind the construction of quantum money from knots~\cite{ITCS:FGHLS12} and its precedent \cite{ITCS:LAFGKH10}. 

At a high level, we start from a set $X$, which is partitioned into many disjoint sets $O\subseteq X$. There is a collection of efficiently computable (and efficiently invertible) permutations on $X$, such that for every permutation in the collection and every $O$ in the partition, the permutation maps elements of $O$ to $O$. Such a set of permutations allows one to take an element $x\in O$, and perform a walk through $O$. We additionally assume an invariant $I:X\rightarrow Y$ on $X$, such that $I$ is constant on each element $O$ of the partition. In other words, $I$ is invariant under action by the collection of permutations.

In the case of~\cite{ITCS:FGHLS12}, $X$ is essentially the set of knot diagrams\footnote{Due to certain concerns about security,~\cite{ITCS:FGHLS12} actually sets $X$ to contain extra information beyond a knot diagram.}, the permutations are Reidemeister moves, and the invariant is the Alexander polynomial. 

An honest quantum money state will essentially be a uniform superposition over $O$\enspace\footnote{Technically, it is a uniform superposition over the pre-images of some $y$ in the image of $I$. If multiple $O$ have the same $y$, then the superposition will be over all such $O$.}. Such a state is constructed by first constructing the uniform superposition over $X$, and then measuring the invariant $I$. Applying a permutation from the collection will not affect such a state. Thus, verification attempts to test whether the state is preserved under action by permutations in the collection by performing an analog of a swap test, and only accepts if the test passes.

In~\cite{ITCS:FGHLS12}, it is explained why certain attack strategies are likely to be incapable of duplicating banknotes. However, no security proof is given under widely believed hard computational assumptions. To make matters worse,~\cite{ITCS:FGHLS12} do not analyze what types of states are accepted by the verifier. It could be, for example, that duplicating a banknote perfectly is computationally infeasible, but there are fake banknotes that pass verification that can be duplicated; this is exactly what happens in the lattice-based schemes analyzed above in Section~\ref{sec:overviewlattice}. Given the complexities of their scheme, there have been limited efforts to understand the security of the scheme. This is problematic, since there have been many candidates for public key quantum money that were later found to be insecure. 

Generally, a fundamental issue with public key quantum money schemes is that, while quantum money schemes rely on the no-cloning principle, the no-cloning theorem is information-theoretic, whereas publicly verifiable quantum money is always information-theoretically clonable. So unclonability crucially relies on the adversary being computationally efficient. Such computational unclonability is far less understood than traditional computational tasks. Indeed, while there have been a number of candidate post-quantum hard computational tasks, there are very few quantum money schemes still standing. The challenge is in understanding if and how quantum information combines with computational bounds to give computational unclonability.

To overcome this challenge, the security analysis should be broken into two parts: one part that relies on \emph{information-theoretic} no-cloning, and another part that relies on a computational hardness assumption. Of course, the security of the scheme itself could be such an assumption, so we want to make the assumption have nothing to do with cloning. One way to accomplish this is to have the assumption have classical inputs and outputs (which we will call ``classically meaningful''), so that it could in principle be falsified by a classical algorithm, which are obviously not subject to quantum unclonability. Separating out the quantum information from the computational aspects would hopefully give a clearer understanding of why the scheme should be unclonable, hopefully allow for higher confidence in security. Moreover, as essentially all widely studied assumptions are classically meaningful, any attempt to prove security under a widely studied assumptions would have to follow this blueprint, and indeed the proof of quantum money from obfuscation~\cite{EC:Zhandry19b} is of this form.

\paragraph{Our Results.} In this work, we make progress towards justifying invariant-based quantum money.
\begin{itemize}
	\item First, we prove that if a random walk induced by the collection of permutations mixes, then we can completely characterize the states accepted by verification. The states are exactly the uniform superpositions over $O$\enspace\footnote{Or more generally, if multiple $O$ have the same $y$, then accepting states are exactly those that place equal weight on elements of each $O$, but the weights may be different across different $O$}. Unfortunately, it is unclear if the knot construction actually mixes, and any formal proof of mixing seems likely to advance knot theory\footnote{Nevertheless we provide a discussion on the knot money instantiation in \ref{sec:knot_instantiation}.}.
	\item Second, we provide concrete security properties under which we can prove security. These properties, while still not well-studied, at least have no obvious connection to cloning, and are meaningful even classically. Under these assumptions, we can even prove that the schemes are in fact quantum lightning,  the aforementioned strengthening of quantum money where not even the mint can create two banknotes of the same serial number. 
\end{itemize}

 \paragraph{Our Hardness Assumptions}
 We rely on two hardness assumptions in our invariant money scheme for a provably secure: the \emph{path-finding} assumption and \emph{knowledge of path finding} assumption.  

Informally speaking, the path-finding assumption states that, given some adversarially sampled $x$ from a set of elements $X$ and given a set of "permutations" $\Sigma$, it is hard for any efficient adversary, given a random $z \in X$, where there exists some $\sigma \in \Sigma$ such that $\sigma \left(x \right) = z$, to find such a $\sigma$. One can observe that it is similar to a ``discrete logarithm'' style of problem. Even though we cannot use discrete logarithm due to its quantum insecurity, we have similar hard problems in certain isogenies over ellitic curves, abstracted as "group action discrete logarithm" problems \cite{AC:ADMP20}.

\paragraph{Our Knowledge of Path Assumptions.} 
The main novel assumption we use is a ``knowledge of path'' assumption. This roughly says that if an algorithm outputs two elements $x,z$ in the same $O$, then it must ``know'' a path between them: a list of permutations from the collection that, when composed, would take $x$ to $z$. While such a knowledge of path assumption is undoubtedly a strong assumption, it seems plausible in a number of relevant contexts (e.g. elliptic curve isogenies that have no known non-trivial attacks or ``generic'' group actions).

Formalizing the knowledge of path assumption is non-trivial. The obvious \emph{classical} way to define knowledge of path is to say that for any adversary, there is an extractor that can compute the path between $x$ and $z$. Importantly, the extractor must be given the same random coins as the adversary, so that it can compute $x$ and $z$ for itself and moreover know what random choices the adversary made that lead to $x,z$. Essentially, by also giving the random coins, we would be effectively making the adversary deterministic, which is crucial for the extractor's output to be related to the adversary's output. 

Unfortunately, quantumly the above argument does not make much sense, as quantum algorithms can have randomness without having explicit random coins. In fact, there are quantum procedures that are \emph{inherently probabilistic}, in the sense that the process is efficient, but there is no way to run the process twice and get the same outcome both times. This is actually crucial to our setting: we are targeting the stronger quantum lightning, which means that even the mint cannot create two banknotes with the same serial number. This means that the minting process is inherently probabilistic. The adversary could, for example, run the minting process, but with its own minting key. Such an adversary would then be inherently probabilistic and we absolutely would need a definition that can handle such adversaries.

Our solution is to exploit the fact that quantum algorithms can always be implemented \emph{reversibly}. We then observe that with a classical reversible adversary, an equivalent way to define knowledge assumptions would be to just feed the entire \emph{final} state of the adversary (including output) into the extractor. By reversibility, this is equivalent to giving the input, coins included, to the extractor. But this alternate extraction notion actually \emph{does} make sense quantumly. Thus our knowledge of path assumption is defined as giving the extractor the entire final (quantum) state of the reversible adversary, and asking that the extractor can find a path between $x$ and $z$. This assumption allows us to bypass the issue of inherently probabilistic algorithms, and is sufficient for us to prove security.

\paragraph{Instantiations of Invariant Quantum Money and Lightning}
After we provide the characterization of security needed for invariant money, we discuss four candidate instantiations\footnote{Throughout the sections on invariant quantum money framework and construction \ref{sec:invariant_main}, \ref{sec:ellipticcurve}, \ref{sec:fe_instantiation}, \ref{sec:groupaction_instantiation} and \ref{sec:knot_instantiation}, we will sometimes interchangeably use   "money" or "lightning". But in fact the proposed candidates are all candidates for quantum lightning.}: 

\begin{itemize}
\item We show a construction from structured oracles and generic cryptographic group actions.  Notably, while we do not know how to instantiate these oracles, we can prove that this construction is secure assuming the existence of a cryptographic group action and the assumption that the knowledge of path assumption holds over a generic cryptographic group action.\footnote{This seems like a very plausible assumption to us:  classically, the knowledge of exponent would almost trivially hold over generic groups.}
\item We explain how re-randomizable functional encryption, a type of functional encryption with special properties that seem reasonable, can be used to build another candidate quantum lightning. We don't currently have a provably secure construction from standard cryptographic assumptions for this special re-randomizable functional encryption, but we provide a candidate construction based on some relatively well-studied primitives.

\item Elliptic curve isogenies are our final new candidate instantiation.  We outline how, given some assumptions about sampling certain superpositions of elliptic curves, it may be possible to build quantum lightning from isogeny-based assumptions.
\item Finally, we analyze the construction of quantum money from knots in~\cite{ITCS:FGHLS12} in our framework.  
\end{itemize}

For all these three constructions, we show that their corresponding \emph{path-finding} problem between two elements $x,z$ in the same $O$ is relatively straightforward to study (reducible to reasonably well-founded assumptions). Nevertheless, we need the knowledge of path assumptions to show that we can extract these paths from a (unitary) adversary. We believe that one may show a knowledge-of-path property when replacing some plain model components in the above candidates with (quantum accessible) \emph{classical} oracles, thus giving the possibility for a first quantum lightning scheme relative to only classical oracles and widely studied assumptions.

\section{Related Work and Discussion}
\label{sec:related_work}

\paragraph{Other Related Work.}  In addition to the related work we have discussed in the introduction, we would like to mention some other relevant work.

The schemes \cite{amos2020one} and \cite{shmueli2022public} can be viewed as quantum lightning with an interactive minting procedure: a multi-round interactive protocol between the bank and the user is needed to create a valid money state. \cite{amos2020one} is based on classical oracles and \cite{shmueli2022public} is based on (subexponential) iO. One may argue that ``compressing'' their interaction might lead to a standard quantum lightning protocol. However, after investigating the constructions, we believe that showing the security for such ``compressed'' protocols leads us back to the land of highly non-standard assumptions.

Semi-quantum money is a notion put forward in
\cite{RadSat19}. More precisely, the authors construct an interactive private-key quantum money scheme verifiable through classical interactions, built from LWE, with ideas from classically-verifiable proof of quantumness  \cite{FOCS:BCMVV18}. Later, the idea was extended to a notion in-between public and private key quantum money, which is called \emph{two-tier quantum lightning}
in \cite{kitagawa2021secure}. Due to the essential structure of the proof-of-quantumness protocol on which these two schemes are based, they cannot be made publicly verifiable.

Regarding non-collapsing functions: it has been shown that some "natural" classes of hash functions are collapsing \cite{EC:Unruh16,C:LiuZha19,zhandry2022new,cao2022gap}.
\cite{FOCS:ARU14} constructed non-collapsing functions using quantum oracles.

For the cryptanalysis on existing quantum money proposals: \cite{cryptanalysisqm2022bilyk} shows a quantum reduction from the conjecture in \cite{EPRINT:KanShaSil21} to a linear algebra problem, which is the only cryptanalysis work we know of on a quantum money scheme that is not broken yet.

\paragraph{Discussion and Open Problems.}
In this paper, we aimed to investigate the feasibility of quantum money and, in particular, quantum lightning.  While we hope that readers believe our work helps shed light on the subject, we still believe that this is a wide open area for study with important applications once we reach a world where quantum computers proliferate.

In particular, fully settling the question of whether or not it is possible to build quantum lightning from lattice-based assumptions would be a very exciting result.  Our paper rules out (arguably) the most natural class of schemes, but that does not mean that a less natural lattice-based quantum lightning scheme could exist.


\section{Preliminaries}\label{sec:prelim}
In this section we explain some background material needed for our work.  

For quantum notations,  we denote $| \cdot \rangle$ as the notation for a pure state and $| \cdot \rangle \langle \cdot |$ for its density matrix. $\rho$ denotes a general mixed state.  

We will go over some fundamental lattice facts and then move to quantum money definitions.
 Due to the restriction of space, we leave some additional lattice basics, hardness theorems and necessary quantum background (in particular related to lattices) to Appendix \ref{sec:appendix_prelim}.  


\subsection{Lattice Basics}
We say a distribution $\mathcal{D}$ is $\left(B, \delta \right)$-bounded if the probability that $\mathcal{D}$ outputs a value larger than $B$ is less than $\delta$.  We extend this to distributions that output vectors in an entry-by-entry way.  Given a set of vectors $\Bm = \left\{ \bv_{1}, ... , \bv_{n} \right\}$, we define the norm of $\Bm$, denoted $|| \Bm ||$, as the length of the longest vector in $\Bm$, so $|| \Bm || = \max_{i} || \bv_{i} ||$.  For any lattice $\Lambda$, we define the minimum distance (or first successive minimum) $\lambda_{1} \left( \Lambda \right)$ as the length of the shortest nonzero lattice vector in $\Lambda$.

\ifllncs
\else
We next define a the \emph{Gram-Schmidt basis} and the \emph{Gram-Schmidt norm} based on the definitions of~\cite{STOC:GenPeiVai08}.

\begin{definition} \label{def:GS} \textbf{Gram-Schmidt Basis.} For any (ordered) set $\Sm = \left\{ \sv_{1}, ... , \sv_{n} \right\} \subset \R^{n}$ of linearly independent vectors, let $\tilde{S} = \left\{ \tilde{\sv}_{1}, ... , \tilde{\sv}_{n} \right\}$ denote its Gram-Schmidt orthogonalization, defined iteratively in the following way: $\tilde{\sv}_{1} = \sv_{1}$, and for each $i \in \left[2, n \right]$, $\tilde{\sv}_{i}$ is the component of $\sv_{i}$ orthogonal to $span \left(\sv_{1}, ... , \sv_{i - 1} \right)$.
\end{definition}

\fi

We next define discrete Gaussians formally.  Since we later use their lemmas, our definition is loosely based on that of~\cite{STOC:BLPRS13}.
\begin{definition}~\label{def:discG}
For any $\sigma > 0$, the $n$-dimensional Gaussian function $\rho_{\sigma} : \R^{n} \rightarrow \left[0, 1 \right]$ is defined as 
\[
\rho_{r} \left( \xv \right) = e^{ - \pi \frac{\xv^2}{\sigma^2}}
\]
We define the \emph{discrete Gaussian function} with parameter $\sigma$ at point $\pv \in \R^{n}$, which we usually denote $\mathcal{D}_{\mathbf{\Psi}_{\sigma}}$ or just $\mathbf{\Psi}_{\sigma}$ when the context is clear, as the function over all of the integers $\yv \in \Z^{n}$ such that the probability mass of any $\yv$ is proportional to
\[
e^{ - \pi \frac{\left( \pv - \yv \right)^2}{\sigma^2}} .
\]
We can also define more complicated discrete Gaussians over lattices.  In this case, let $\mathbf{\Sigma}$ be a matrix in $\R^{n \times n}$.  The discrete Gaussian over a lattice $\Lambda$ with center $\pv$ and ``skew'' parameter $\Sigma$ is the function over all \emph{lattice points} in $\Lambda$ such that the probability mass of  any $\yv$ is proportional to
\[
e^{ - \pi \left( \pv - \yv \right)^{T} \left( \mathbf{ \Sigma \Sigma^{T}} \right)^{-1} \left(\pv - \yv \right)} ,
\]
very similar to as before.  We usually denote this type of discrete Gaussian as $\mathbf{\Psi}_{\Lambda, \mathbf{\Sigma}, \pv}$ or $\mathcal{D}_{\mathbf{\Psi}_{\Lambda, \mathbf{\Sigma}, \pv}}$, where we sometimes substitute $\sigma$ for $\mathbf{\Sigma}$ when $\mathbf{\Sigma} = \sigma \cdot \mathbf{I}_{n}$, where $\mathbf{I}_{n}$ is the $n \times n$ identity matrix.  We also sometimes omit parameters when they are obvious (e.g. $0$) in context.

\end{definition}

\ifllncs
We will explain how to efficiently sample discrete Gaussians quantumly later in \ref{sec:lattice_quantum_fact}.
\else
\subsection{(Lattice-relevant) Quantum Facts}
\label{sec:lattice_quantum_fact}

All quantum notations used in this work are relatively basic and standard; we therefore omit extensive preliminaries on them. For further quantum basics, we refer the readers to \cite{nielsen2002quantum}.

\subsubsection{Gaussian Superposition Preparation}
We now show that it is possible to efficiently sample a discrete Gaussian quantumly over a lattice basis.

\paragraph{Generating Gaussian Superpositions.} Let $\Ls$ be a lattice. Given a vector $\cv$ (not necessarily in $\Ls$) and covariance matrix $\Sigma$, define \[|\mathbf{\Psi}_{\Ls,\Sigma,\cv}\rangle\propto\sum_{\xv\in\Ls}\sqrt{e^{-\pi (\xv-\cv)^T\cdot\Sigma^{-1}\cdot(\xv-\cv)}}|\xv\rangle\]
This is the discrete-Gaussian-weighted superposition over lattice vectors.

\begin{theorem}\label{thm:gaussiansuper} There is a QPT algorithm which, given a basis $\Bm=(\bv_1,\dots,\bv_n)$ for a lattice $\Ls$, a center $\cv$, and covariance matrix $\Sigma$ such that $\bv_i\cdot\Sigma^{-1}\cdot\bv_i\leq 1/\omega(\log\lambda)$, constructs a state negligibly close to $|\mathbf{\Psi}_{\Ls,\Sigma,\cv}\rangle$.
\end{theorem}

We prove Theorem~\ref{thm:gaussiansuper} by gradually building up from special cases.

\begin{lemma}\label{lem:integersuper}There is a QPT algorithm which, given $c\in\Z$ and $\sigma\geq\omega(\sqrt{\log\lambda})$, constructs a state negligibly close to $|\mathbf{\Psi}_{\Z,\sigma^2,c}\rangle$. 
\end{lemma}
\begin{proof}This is a straightforward adaptation of the classical algorithm for sampling from the discrete Gaussian over integers~\cite{STOC:GenPeiVai08}. Let $t\geq\omega(\sqrt{\log\lambda})$. The algorithm proceeds in the following steps:
	\begin{enumerate}
		\item Let $\Is = \Z\cap[c-t\sigma,c+t\sigma]$. Let $x_{min}$ be the minimal element of $\Is$, and $w=|\Is|$.
		\item\label{step:begining} Initialize a register to $|0\rangle$. Then using the QFT on $w$ elements, construct the state $\propto\sum_{i=0}^{w-1}|i\rangle$. 
		\item By adding $x_{min}$ in superposition, construct the state $\propto\sum_{i\in \Is}|i\rangle$
		\item Now apply in superposition the map $|i\rangle\mapsto|i\rangle\otimes \left(\sqrt{e^{-\pi (i-c)^2/\sigma^2}}|0\rangle+\sqrt{1-e^{-\pi (i-c)^2/\sigma^2}}|1\rangle\right)$.
		\item Measure the second register, obtaining a bit $b$; the first register collapses to a state $|\mathbf{\Psi}\rangle$. If $b=0$, output $|\mathbf{\Psi}\rangle$. Otherwise, discard $|\mathbf{\Psi}\rangle$ and restart from Step~\ref{step:begining}.
	\end{enumerate}
Note that the state outputted by the above algorithm is $|\mathbf{\Psi}\rangle\propto\sum_{i\in\Is} \sqrt{e^{-\pi (i-c)^2/\sigma^2}} |i\rangle$. This is identical to $|\mathbf{\Psi}_{\Z,\sigma^2,c}\rangle$, except that the support is truncated to the interval $\Is$ (and therefore the state is also re-scaled, but the proportions for $i\in\Is$ are identical).~\cite{STOC:GenPeiVai08} show that the analogous distributions over $i$ are negligibly close (\cite{STOC:GenPeiVai08}, Lemma 4.3). An almost identical argument shows that the states $|\mathbf{\Psi}_{\Z,\sigma^2,c}\rangle$ and $|\mathbf{\Psi}\rangle$ are negligibly close.
\end{proof}

\begin{lemma}\label{lem:latticesuper} There is a QPT algorithm which, given a basis $\Bm=(\bv_1,\dots,\bv_n)$ for a lattice $\Ls$, a center $\cv\in\Z^n$, and $\sigma\geq \|\Bm\|\omega(\sqrt{\log\lambda})$, constructs a state negligibly close to $|\mathbf{\Psi}_{\Ls,\sigma^2,\cv}\rangle$. 
\end{lemma}
\begin{proof}This is also a straightforward adaptation of the classical algorithm for sampling from discrete Gaussians over lattices. However, care is needed to ensure that there is no spurious information left behind; such spurious information would not affect classical sampling, but could be entangled with the resulting quantum state, thereby perturbing it. We show that the algorithm can be implemented without such perturbation.
	
The classical sampling algorithm (\cite{STOC:GenPeiVai08}, Section 4.2) works as follows:
\begin{enumerate}
	\item Initialize $\vv\gets 0$. Then for $i=1,\dots,m$ ($m$ being the dimension of $\Ls$), do:
	\begin{enumerate}
		\item Let $c_i' = \frac{(\cv-\vv)\cdot\tilde{\bv}_i}{|\tilde{\bv_1}_i|^2}$ and $\sigma_i'=\sigma/|\tilde{\bv_1}_i|$.
		\item Sample $z_i\gets D_{\Z,\sigma_i',c_i'}$
		\item Update $\vv\gets\vv+z_i\bv_i$. 
	\end{enumerate} 
\end{enumerate}
Then \cite{STOC:GenPeiVai08} proves that the output distribution is statistically close to $D_{\Ls,\sigma,\cv}$. We now explain how to run the sampling algorithm coherently to produce $|\mathbf{\Psi}_{\Ls,\sigma^2,\cv}\rangle$:
\begin{enumerate}
	\item Initialize a register $\Vs$ to $|0\rangle$. Then for $i=1,\dots,m$, do:
	\begin{enumerate}
		\item\label{step:ciprime} Apply the map $|\vv\rangle\mapsto|\vv\rangle\otimes|c_i'\rangle$ in superposition to $\Vs$, where $c_i' = \frac{(\cv-\vv)\cdot\tilde{\bv}_i}{|\tilde{\bv_1}_i|^2}$. Also compute $\sigma_i'=\sigma/|\tilde{\bv_1}_i|$.
		\item Apply the map $|\vv\rangle\otimes|c_i'\rangle\mapsto |\vv\rangle\otimes|c_i'\rangle\otimes |\mathbf{\Psi}_{\Z,(\sigma_i')^2,c_i'}\rangle$ in superposition. Here, we assume for simplicity the ideal $|\mathbf{\Psi}_{\Z,(\sigma_i')^2,c_i'}\rangle$, but we would actually use the algorithm from Lemma~\ref{lem:integersuper}, incurring a negligible error.
		\item Uncompute $c_i'$ by performing the operation in Step~\ref{step:ciprime} in reverse.
		\item Apply the map $|\vv,z_i\rangle\mapsto |\vv+z_i\bv_i,z_i\rangle$
		\item Uncompute $z_i$, which can be computed from $\vv+z_i\bv_i$ via linear algebra, since $z_i$ is just the coefficient of $\bv_i$ when representing $\vv$ in the basis $\{\bv_1,\dots,\bv_i\}$.
	\end{enumerate} 
\end{enumerate}
By an essentially identical analysis to~\cite{STOC:GenPeiVai08}, Theorem 4.1, the resulting state is statistically close to $|\mathbf{\Psi}_{\Ls,\sigma^2,\cv}\rangle$.
\end{proof}

We are now ready to prove Theorem~\ref{thm:gaussiansuper}, which follows as a simple corollary from Lemma~\ref{lem:latticesuper}.
\begin{proof} Write $\Sigma^{-1}$ as $\Sigma^{-1}=\Um^T\cdot\Um$. Let $\Bm'=\{\bv_1'=\Um\cdot\bv_1,\dots,\bv_n'=\Um\cdot\bv_n\}$, and let $\Ls'$ be the lattice generated by $\Bm'$. Let $\cv'=\Um\cdot\cv$.
	
Now invoke Lemma~\ref{lem:latticesuper} on $\Bm'$ to construct the state (statisticaly close to) $|\mathbf{\Psi}_{\Ls',1,\cv'}\rangle$. To see that $\Bm'$ satisfied the conditions of Lemma~\ref{lem:latticesuper}, observe that $\|\Bm'\|\leq \max_i |\bv_i'|$, and we have that $|\bv_i'|^2=\bv_i^T\cdot \Um^T\cdot\Um\bv_i=\bv_i\cdot\Sigma^{-1}\bv_i\leq 1/\omega(\log\lambda)$. 

Now given $|\mathbf{\Psi}_{\Ls',1,\cv'}\rangle$, we simply apply in superposition the map $\vv\mapsto \Um^{-1}\vv$, which gives us exactly. $|\mathbf{\Psi}_{\Ls,\Sigma,\cv}\rangle$.
\end{proof}
\fi
\medskip

\subsection{General LWE Definition}
\label{sec:lwe_def}
In this section we define basic LWE with an eye towards eventually defining $k$-LWE.  We note that, while equivalent to the standard definitions, our definitions here are presented a little bit differently than usual in lattice cryptography.  This is so that we can keep the notation more consistent with the typical quantum money and quantum algorithms presentation styles.  We first provide a properly parameterized definition of the LWE problem~\cite{STOC:Regev05}. 

\begin{definition}~\label{def:LWE} \textbf{Learning with Errors (LWE) Problem:}
Let $n$, $m$, and $q$ be integers, let $\mathcal{D}_{\Am}$ and $\mathcal{D}_{\rv}$ be distributions over $\Z_{q}^{n}$, and let $\mathcal{D}_{\mathbb{\Psi}}$ be a distribution over $\Z_{q}^{m}$.  Let $\Am \in \Z_{q}^{m \times n}$ be a matrix where each row is sampled from $\mathcal{D}_{\Am}$, let $\rv \in \Z_{q}^{n}$ be a vector sampled from $\mathcal{D}_{\rv}$, and let $\ev \in \Z_{q}^{m}$ be a vector sampled from $\mathcal{D}_{\mathbb{\Psi}}$.  Finally, let $\tv \in \Z_{q}^{m}$ be a uniformly random vector.

The $\left(n, m, q, \mathcal{D}_{\Am}, \mathcal{D}_{\rv}, \mathcal{D}_{\mathbb{\Psi}} \right)$-LWE problem is defined to be distinguishing between the following distributions: 
\[  \left( \Am,  \Am \cdot \rv + \ev \right) \text{ and } \left( \Am, \tv \right) . \] 
\end{definition}

\ifllncs
\else
\subsection{The $k$-LWE Problem}
With the LWE definition in place, we are ready to move to the actual $k$-LWE problem.  The $k$-LWE problem was first formally defined in~\cite{C:LPSS14} and used to build traitor-tracing schemes.  It extends the $k$-SIS assumption~\cite{PKC:BonFre11} which was used to build linearly homomorphic signatures.  Our definition below is essentially a parameterized version of the one in~\cite{C:LPSS14}.

\begin{definition}~\label{def:kLWE} \textbf{$k$-LWE Problem:}
Let $k$, $n$, $m$, and $p$ be integers, let $\mathcal{D}_{\Rm}$ be a distribution over $\Z_{q}^{n}$, and let $\mathcal{D}_{\mathbb{\Psi}}$ and $\mathcal{D}_{\Sm}$ be distributions over $\Z_{q}^{m}$.  Let $\Sm \in \Z_{q}^{k \times m}$ be a matrix where each \emph{row} is selected from $\mathcal{D}_{\Sm}$.  

Let $\Am \in \Z_{q}^{m \times n}$ be a matrix sampled uniformly from the set of matrices in $\Z_{q}^{m \times n}$ such that $\Sm \cdot \Am = 0 \mod q$, let $\rv \in \Z_{q}^{n}$ be a vector sampled from $\mathcal{D}_{\rv}$, and let $\ev \in \Z_{q}^{m}$ be a vector sampled from $\mathcal{D}_{\mathbb{\Psi}}$.  Let $\Cm \in \Z_{q}^{m \times \left(m - k \right)}$ be a basis for the set of vectors $\vv \in \Z_{q}^{m}$ such that $\Sm \cdot \vv = 0$, and let $\rv' \in \Z_{q}^{m - k}$ be a uniformly random vector.

The $\left(k, n, m, q, \mathcal{D}_{\Sm}, \mathcal{D}_{\rv}, \mathcal{D}_{\mathbb{\Psi}} \right)$-$k$-LWE problem is defined to be distinguishing between the following distributions: 
\[  \left( \Sm, \Am, \Cm, \Am \cdot \rv+ \ev \right) \text{ and } \left(\Sm, \Am,  \Cm, \Cm \cdot \rv'+ \ev \right) \] 
\end{definition}

We note that $k$-LWE is traditionally defined in a slightly different way:  usually the matrix $\Am$ is sampled before (or jointly with) $\Sm$ rather than after it.  We sample $\Sm$ first in our definition because we will need to handle very unusual (at least for cryptographic applications) distributions $\mathcal{D}_{\Sm}$.

\fi

\subsection{Quantum Money and Quantum Lightning}

Here, we define public key quantum money and quantum lightning. Following Aaronson and Christiano~\cite{STOC:AarChr12}, we will only consider so-called ``mini-schemes'', where there is only a single banknote.

Both quantum money and quantum lightning share the same syntax and correctness requirements. There are two quantum polynomial-time algorithms $\gen,\ver$ such that:
\begin{itemize}
	\item $\gen(1^\lambda)$ samples a classical serial number $\sigma$ and a quantum state $|\psi\rangle$.
	\item $\ver(\sigma,|\psi\rangle)$ outputs a bit 0 or 1.
\end{itemize}

\paragraph{Correctness.} We require that there exists a negligible function $\negl$ such that $\Pr[\ver(\gen(1^\lambda))]\geq 1-\negl(\lambda)$.

\paragraph{Security.} Where public key quantum money and quantum lightning differ is in security. The differences are analogous to the differences between one-way functions and collision resistance.

\begin{definition}[Quantum Money Unforgeability] $(\gen,\ver)$ is secure public key quantum money if, for all quantum polynomial-time $A$, there exists a negligible $\negl$ such that $A$ wins the following game with probability at most $\negl$:
	\begin{itemize}
		\item The challenger runs $(\sigma,|\psi\rangle)\gets\gen(1^\lambda)$, and gives $\sigma,|\psi\rangle$ to $A$.
		\item $A$ produces a potentially entangled joint state $\rho_{1,2}$ over two quantum registers. Let $\rho_1,\rho_2$ be the states of the two registers. $A$ sends $\rho_{1,2}$ to the challenger.
		\item The challenger runs $b_1\gets\ver(\sigma,\rho_1)$ and $b_2\gets\ver(\sigma,\rho_2)$. $A$ wins if $b_1=b_2=1$.
	\end{itemize}
\end{definition}

\begin{definition}[Quantum Lightning Unforgeability] $(\gen,\ver)$ is secure quantum lightning if, for all quantum polynomial-time $A$, there exists a negligible $\negl$ such that $A$ wins the following game with probability at most $\negl$:
	\begin{itemize}
		\item $A$, on input $1^\lambda$, produces and sends to the challenger $\sigma$ and $\rho_{1,2}$, where $\rho_{1,2}$ is a potentially entangled joint state over two quantum registers.
		\item The challenger runs $b_1\gets\ver(\sigma,\rho_1)$ and $b_2\gets\ver(\sigma,\rho_2)$. $A$ wins if $b_1=b_2=1$.
	\end{itemize}
\end{definition}

The difference between quantum lightning and quantum money is therefore that in quantum lightning, unclonability holds, even for adversarially constructed states.

Note that, as with classical collision resistance, quantum lightning does not exist against non-uniform adversaries. Like in the case of collision resistance, we can update the syntax and security definition to utilize a common reference string (crs), which which case non-uniform security can hold. For this paper, to keep the discussion simple, we will largely ignore the issue of non-uniform security.

\section{The Flaw in \cite{KLS22} Lattice-Based Quantum Money}\label{sec:flaw}

\subsection{Overview of \cite{KLS22}}

We do not describe the whole scheme here, but re-create enough of the scheme to describe our attack.

\paragraph{Setup.} To set up the scheme, first a few parameters are specified:
\begin{itemize}
\item An exponentially-large prime $P$
\item A Gaussian width $\sigma\gg \sqrt{d} P^{1/d}$
\item An integer $t$ (\cite{KLS22} suggest $t=3$).
\item An odd positive integer $\Delta\gtrsim\sigma$.
\item An integer $k\leq t\sigma\Delta$.
\end{itemize}

Next \cite{KLS22} create a vector $\vv=(v_1,\dots,v_d)\in\mathbb{Z}_P^d$ such that 
\[v_1=1\;\;\;,\;\;\; v_2=\frac{P+\Delta}{2}=\Delta/2\bmod P\]
All the remaining entries of $v$ are uniform in $\Z_P$. \cite{KLS22} then defines the linear subspaces $\Ls=\{\xv\in\Z_P^d:\xv.\vv=0\bmod P\}$ and $\Ls^\perp=\{m\vv:m\in\Z_P\}\subseteq\Z_P^d$. Note that the authors require that there is a $\wv\in\Ls^\perp$ such that $\wv\cdot\vv=1\bmod P$. This is equivalent to requiring that $\vv.\vv\neq 0\bmod P$, which holds with overwhelming probability since $\vv.\vv$ is statistically close to uniform over $\Z_P$. 

\paragraph{Minting Process.} After a few steps, the minting process has the following state (\cite{KLS22} page 3, Eq (8)) over registers $\Xs,\Us$:
\begin{equation}\label{eq:initialstate}|\phi\rangle\approx\frac{1}{C}\sum_{\uv\in\Ls^\perp}\sum_{\xv\in\Ls}e^{-(x-u)^2/4\sigma^2}|\xv\rangle|\uv\rangle\end{equation}
Where $C$ is a normalization constant.

\medskip

Next, the mint maps $\uv\in\Ls^\perp$ to $m\in\Z_P$ in the $\Us$ register, where $\uv=m\wv$, obtaining the state close to:
\[\frac{1}{C}\sum_{m\in\Z_P}\sum_{\xv\in\Ls}e^{-(x-m\wv)^2/4\sigma^2}|\xv\rangle|m\rangle\]

Next, the mint applies a certain POVM to the $\Us$ register. Let $\Gamma=1/2+t\sigma\Delta\bmod P$. The measurement is specified matrices $M_T$ where:
\[M_T=\frac{1}{4k+2}\left(\sum_{j=-k}^k|T+j\rangle \langle T+j|+|T+\Gamma+j\rangle\langle T+\Gamma+j|\right)\]
\cite{KLS22} does not explain how to realize the POVM. But the following process suffices:
\begin{itemize}
\item Initialize registers $\Js$ and $\Bs$ to $\frac{1}{\sqrt{2k+1}}\sum_{j=-k}^k |j\rangle$ and $\frac{1}{\sqrt{2}}\left(|0\rangle+|1\rangle\right)$, respectively. It also initializes a register $\Ts$ to $|0\rangle$
\item Apply the following operation to the $\Us,\Js,\Bs,\Ts$ registers:
\[|m,j,b,0\rangle\mapsto|m,j,b,m-j-b\Gamma\rangle\]
\item Now measure the $\Ts$ register to obtain a serial number $T$. Discard the $\Ts$ register.
\item Now, given $T=m-j-b\Gamma$, and $\xv$, this uniquely determines $m,j,b$: since $\xv-m\wv$ must be a short vector, there is only a single value of $m$ for any $\xv$. As explained in~\cite{KLS22}, $m\wv$ and hence $m$ can be computed from $\xv$ alone in the given parameter settings. Moreover, since $\Gamma\approx P/2 \gg k$, there is at most a single $(j,b)$ for any $m,T$ with $T=m-j-b\Gamma$.
Moreover, $(j,b)$ can be computed efficiently by simply trying both $b=0$ and $b=1$, and seeing which one gives $j=T+b\Gamma\in[-k,k]$. So we use this to uncompute $\Us,\Js,\Bs$ registers, which we then discard.
\end{itemize}

After obtaining $T$, the banknote is whatever state $|\phi_T\rangle$ is left in the $\Xs$ register, and the serial number is $T$.

\paragraph{Verification.} We don't replicate the verification procedure here. Ultimately, we will show that verification is irrelevant: there is a procedure which produces two faked quantum money states, but which nevertheless fools the verifier.

\begin{remark}In~\cite{KLS22}, the banknote actually consists of a tuple of states as above. The overall verification will verify each element of the tuple separately, and then accept if a certain threshold of the elements accept. This structure is needed because their verification algorithm will reject honest banknotes some of the time. As we will see, this fact will not affect our attack, since our forged money will pass the verification of individual elements with the same probability as an honest banknote (up to potentially negligible error). For this reason, we just focus on a single element as described above.
\end{remark}

\subsection{An Alternate View of \cite{KLS22}}

Here, we present an alternative view of~\cite{KLS22} that we believe is easier to reason about.

First, we notice that the requirement that $\vv\cdot\vv\neq 0\bmod P$ implies that the linear spaces $\Ls$ and $\Ls^\perp$ together span all of $\Z_P^d$. This means there is a one-to-one correspondence between  $\Z_P^d$ and $\Ls\times\Ls^\perp$. Using that $\Ls^\perp$ is just the linear space of multiples of $\vv$ and that  $\wv\cdot\vv=1$ by the definition of $\wv$, we also have the correspondence:
\begin{align}\nonumber\yv\in\Z_P^d&\leftrightarrow (\xv,m)\in\Ls\times\Z_P\\
\label{eqn:alternate}\begin{array}{rcl}\yv&=&\xv-m\wv\end{array}&\leftrightarrow
\begin{array}{rcl}m&=&-\vv\cdot\yv\\
\xv&=&\yv+m\wv=\yv-(\vv\cdot\yv)\wv\end{array}\end{align}
Using this correspondence, the state in Equation~\ref{eq:initialstate} can be equivalently written as a state over register $\Ys$ containing a Gaussian-weighted superposition over all integers:
\[|\phi_0\rangle=\frac{1}{C}\sum_{\yv \in \Z_P^d}e^{-|\yv|^2/4\sigma^2}|\yv\rangle\]
Going between $|\phi\rangle$ and $|\phi_0\rangle$ is just a simple matter of linear algebra to go between $\yv$ and $\xv,m$.

\begin{remark}Note that our alternative view gives a much simpler way to construct the state in Equation~\ref{eq:initialstate} than what was described in~\cite{KLS22}. Indeed, they construct the state by preparing first a large Guassian superposition of \emph{lattice points} in $\Ls$. Then it performs a \emph{lattice} quantum Fourier transform, which gives a superposition over lattice points very close to dual lattice vectors. Finally, the lattice points in the superposition are close enough to the dual vectors that bounded distance decoding algorithms can be run to recover the dual lattice vectors, which gives the state in Equation~\ref{eq:initialstate}. Each of these steps requires non-trivial algorithms that must all be performed in superposition. In contrast, our alternative view shows that one could instead simply create a Guassian-weighted superposition of \emph{integers} and then perform some basic linear algebra.
\end{remark}

\paragraph{The Measurement.} Now consider the POVM described above. The action of the POVM on $|\phi\rangle$ translates straightforwardly to acting on $|\phi_0\rangle$, where instead of acting on the $\Us$ register containing $m$ directly, we instead compute $m$ from $\yv$ as $m=-\vv\cdot\yv$, and apply the POVM to this $m$. The $|\phi_T\rangle$ of the original minting process can therefore be obtained from $|\phi_0\rangle$ as follows:
\begin{itemize}
\item Initialize registers $\Js$ and $\Bs$ to $\frac{1}{\sqrt{2k+1}}\sum_{j=-k}^k |j\rangle$ and $\frac{1}{\sqrt{2}}\left(|0\rangle+|1\rangle\right)$, respectively. It also initializes a register $\Ts$ to $|0\rangle$
\item Apply the following operation to the $\Ys,\Js,\Bs,\Ts$ registers:
\[|\yv,j,b,0\rangle\mapsto|\yv,j,b,-\vv\cdot\yv-j-b\Gamma\rangle\]
\item Now measure the $\Ts$ register to obtain serial number $T$. Discard the $\Ts$ register. Call the state remaining in the $\Ys,\Js,\Bs$ registers $|\phi_T'\rangle$.
\item Given $\yv,T=-\vv\cdot\yv-j-b\Gamma$, we can compute $m,\xv$ from $\yv$ via the correspondence in Equation~\ref{eqn:alternate}. From $m,\xv$ we can un-compute $\yv$, again using Equation~\ref{eqn:alternate}. Then using $\xv,T$, we can uniquely determine $(m,j,b)$. So we use this to uncompute $\Us,\Js,\Bs$ registers, which we then discard. 
\end{itemize}

\paragraph{Our equivalent formulation.} Instead of arriving at $|\phi_T\rangle$ using our alternate formulation, we simply stop at the second-to-last step, outputting $|\phi_T'\rangle$. Since the final step of the minting process is reversible, it is equivalent to give out $|\phi_T\rangle$ and $|\phi_T'\rangle$.

\subsection{Summary of Alternate Minting Process}\label{sec:alternate}

We can now concisely describe out alternate minting process. This process will be equivalent to the original one, in that, we can map banknotes from our process to banknotes from~\cite{KLS22} and vice versa. This means that our alternate minting process is secure if and only if~\cite{KLS22} is secure.

\paragraph{Setup.} Let $P,\sigma,\Delta,t,k$ as before. Let $\vv$ as before, and define $\vv'=(-\Gamma,-1,\vv)\in\Z_P^{d'}$ where $d'=d+2$.

\paragraph{Minting.} Let $\Ys'=\Bs\times\Js\times\Ys$. Initialize the following state over register $\Ys'$:
\[|\phi'\rangle\propto\sum_{\yv\in\Z_P^d,j\in[-k,k],b\in\{0,1\}}e^{-|\yv|^2/4\sigma^2}|b,j,\yv\rangle\]

Note that this is a superposition over short vectors $\yv'\in\Z_P^{d'}$. It is also a product state, with each coordinate being produced separately.

Now we apply the following map in superposition:
\[|\yv'\rangle\mapsto|\yv',\vv'\cdot\yv'\rangle\]
Now measure the $\Ts$ register (the second register), obtaining $T$. The result of the $\Ys'$ register is exactly the state $|\phi_T'\rangle$, a uniform superposition over short vectors $\yv'$ such that $\vv'\cdot\yv'=T$. Here, by short, we mean that each entry of $\yv'$ is in $[-W,W]$, except with negligible probability, where $W=\max(2,k,\sigma\times\omega(\sqrt{\log(\lambda))} )$.

\subsection{The Flaw}

Here, we describe the flaw in~\cite{KLS22}. The flaw is most easily seen using our alternate view of their scheme.

Note that the vector $\vv'$ has 3 known orthogonal short vectors, namely \[\sv_0=\left(\begin{array}{c}0\\0\\\Delta\\-2\\0\\\vdots\\0\end{array}\right)\;\;\;\sv_1=\left(\begin{array}{c}0\\1\\1\\0\\0\\\vdots\\0\end{array}\right)\;\;\;\sv_2=\left(\begin{array}{c}-2\\2\Gamma=2t\sigma\Delta+1\\0\\0\\0\\\vdots\\0\end{array}\right)\]

In~\cite{KLS22}, the authors define three vectors $s^{(0)},s^{(1)},s^{(2)}$; $\sv_0,\sv_1,\sv_2$ play the roles of $s^{(0)},s^{(1)},s^{(2)}$ in the alternate view, respectively.

In the security proof (originally starting from Page 4 of~\cite{KLS22}, here translated to our alternate view), the authors imagine an adversary an adversary that outputs two possibly entangled states that pass verification. The authors argue that each state $|\phi\rangle$ must have significant weight on $\yv'$ where the first digit $b$ is 0, and also significant weight where $b$ is 1. Recall that a valid money state must have the first digit being 0 or 1. If true, this would allow them to finish the proof: by measuring both states, the authors can then conclude that with non-negligible probability the resulting vectors $\yv'_0,\yv'_1$ will have different first bits. Then the difference between the vectors $\yv'_0-\yv'_1$ is a short vector orthogonal to $\vv'$. It moreover cannot be in the span of $\sv_0,\sv_1,\sv_2$ because if if did, the coefficient of $\sv_2$ must be $1/2$, and since the second coordinate of $\sv_2$ is odd, the resulting vector would have a large second coordinate. Thus, the authors obtain a short vector linearly independent of the provided vectors, which is presumably hard.

The flaw in their argument is the claim that each state the adversary constructs must have significant weight on $b=0$ and also $b=1$. Indeed, the authors toward contradiction consider an adversary which measures the provided state to obtain $\yv'$, and then constructs a state of the form:
\[\sum_j\alpha_j|\yv'+j\sv_1\rangle\]
for small $j$. They argue that their verification procedure would reject such a state. Indeed, their test essentially looks at the 4th coordinate in the Fourier domain (which is the 2nd coordinate in their view). Since the 4th coordinate in the primal is not affected by adding multiples of $\sv_1$, the 4th coordinate is exactly the 4th coordinate of $\yv'$. As such, in the Fourier domain the 4th coordinate is uniform. On the other hand, the authors show that the 4th coordinate of their money state is non-uniform, which is leveraged in verification.

However, the authors failed to consider more general states of the form 

\[\sum_{j_0,j_1}\alpha_{j_0,j_1}|\yv'+j_0\sv_0+j_1\sv_1\rangle\]
for small $j_0,j_1$. Since $\sv_0$ is non-zero in the 4th coordinate, by adding multiples of $\sv_0$, the Fourier transform of the 4th coordinate is no longer guaranteed to be uniform.

The authors reject such states as being a possibility, since the 3rd coordinate of $\sv_0$ is $\Delta$, which is larger than the allowed width of the 3rd coordinate of $\yv'$ (which is limited to $\sigma\ll \Delta$). Thus, adding $\sv_0$ to a vector in the allowed support of $\yv'$ would move do outside the allowed support. However, the vector $\sv_0-\Delta\times\sv_1$ has a small non-zero 4rd coordinate, and the 3rd coordinate is 0. Now, the 2nd coordinate of $\sv_0-\Delta\times \sv_1$ is of size $\Delta$, but the 2nd coordinate of $\yv'$ is allowed to be as large as $k\gg \Delta$. So while one cannot add $\sv_0$ itself to a state and remain in the required domain, one can add multiples of $\sv_0-\Delta\times\sv_1$. In other words, even though adding $\sv_0$ results in an invalid $\yv'$, we can add an appropriate multiple of $\sv_1$ to move back to a valid $\yv'$. The introduction of $\sv_0$ completely invalidates their security proof.

\section{Our General Attack on a Class of Quantum Money}\label{sec:attack}

\ifllncs
Due to limitation of space, we leave a detailed discussion of the \cite{KLS22} money scheme and its flaw in \ref{sec:flaw}.
\fi

Now, we show that a natural class of schemes, including  the equivalent view on \cite{KLS22} demonstrated in Section~\ref{sec:alternate}, cannot possibly give secure quantum money schemes, regardless of how the verifier works.

\subsection{The General Scheme}

Here, we describe a general scheme which captures the alternate view above. Here, we use somewhat more standard notation from the lattice literature. Here we give a table describing how the symbols from section \ref{sec:flaw} map to this section:
\begin{center}
	\begin{tabular}{|c|c|}\hline
		{\bf This Section}&{\bf Section~\ref{sec:alternate}}\\\hline
		$q$&$P$\\\hline
		$n$&1\\\hline
		$m$&$d'=d+2$\\\hline
		$\Am$&$\vv'$ as a column vector\\\hline
		$|\psi\rangle$&$|\phi'\rangle$\\\hline
		$\uv$&$T$\\\hline
		$W$&$k+\sqrt{m}\times\sigma\times\omega(\sqrt{\log(\lambda))} )$\\\hline
	\end{tabular}
\end{center}

\paragraph{Setup.} Let $q$ be a super-polynomial, which may or may not be prime. Sample from some distribution several short vectors $\sv_1,\dots,\sv_\ell\in\Z_q^{m}$ for a constant $\ell$, and assemble them as a matrix $\Sm\in\Z_q^{m\times\ell}$. Then generate a random matrix $\Am\in\Z_q^{m\times n}$ such that $\Am^T\cdot\Sm=0\bmod q$.

\paragraph{Minting.} Create some superposition $|\psi\rangle$ of vectors in $\yv\in\Z_q^{m}$ such that an all but negligible fraction of the support of $|\psi\rangle$ are on vectors with norm $W$. Let $\alpha_{\yv}$ be the amplitude of $\yv$ in $|\psi\rangle$.

Then apply the following map to $|\psi\rangle$:
\[|\yv\rangle\rightarrow|\yv,\Am^T\cdot\yv\bmod q\rangle\]
Finally, measure the second register to obtain $\uv\in\Z_q^n$. This is the serial number, and the note is $|\psi_\uv\rangle$, whatever remains of the first register, which is a superposition over short vectors $\yv$ such that $\Am^T\cdot\yv=\uv$.

\paragraph{Verification.} We do not specify verification. Indeed, in the following we will show that the money scheme is insecure, for \emph{any} efficient verification scheme.

\subsection{Attacking the General Scheme}

We now show how to attack the general scheme. Let $\Cm$ be a matrix whose columns span the space orthogonal to the columns of $\Ss$. Let $|\psi_\uv'\rangle$ be the state sampled from $|\psi_\uv\rangle$ by measuring $\yv\mapsto\Cm^T\cdot\yv$, and letting $|\psi_\uv'\rangle$ be whatever is left over.

Our attack will consist of two parts:
\begin{itemize}
	\item Showing that $|\psi_\uv'\rangle$ is indistinguishable from $|\psi_\uv\rangle$, for \emph{any} efficient verification procedure. We show (Section~\ref{sec:indistinguishable}) that this follows from a certain ``$k$-LWE'' assumption, which depends on the parameters of the scheme ($\SS,k,n$, etc). In Section~\ref{sec:klwe}, we justify the assumption in certain general cases, based on the assumed hardness of worst-case lattice problems. Note that these lattice problems are essentially (up to small differences in parameters) the same assumptions we would expect are needed to show security for the money scheme in the first place. As such, if $k$-LWE does not hold for these special cases, most likely the quantum money scheme is insecure anyway. Our cases include the case of~\cite{KLS22}.
	\item Showing that $|\psi_\uv'\rangle$ can be cloned. Our attack first measures $|\psi_\uv'\rangle$ to obtain a single vector $\yv$ in it's support. To complete the attack, it remains to construct $|\psi_\uv'\rangle$ from $\yv$; by repeating such a process many times on the same $\yv$, we successfully clone. We show (Section~\ref{sec:constructingpsi}) that in certain general cases how to perform such a construction. Our cases include the case of~\cite{KLS22}.
\end{itemize}

Taken together, our attack shows that not only is~\cite{KLS22} insecure, but that it quite unlikely that any tweak to the scheme will fix it.

\subsection{Indistinguishability of $|\psi_\uv'\rangle$}\label{sec:indistinguishable}

Here, we show that our fake quantum money state $|\psi_\uv'\rangle$ passes verification, despite being a very different state that $|\psi_\uv\rangle$. We claim that, from the perspective of any efficient verification algorithm, $|\psi_\uv'\rangle$ and $|\psi_\uv\rangle$ are indistinguishable. This would mean our attack succeeds.

Toward this end, let $\Cm\in\Z_q^{m\times (m-\ell)}$ be a matrix whose rows span the space orthogonal to $\Sm$: $\Cm^T\cdot\Sm=0$. Notice that the state $|\psi'_\uv\rangle$ can be equivalently constructed by applying the partial measurement of $\Cm^T\cdot\yv$ to $|\psi_\uv\rangle$

Consider the following problem, which is closely related to ``$k$-LWE''(\ref{def:kLWE}):

\begin{problem}\label{conj:klwe} Let $n,m,q,\Sigma$ be functions of the security parameter, and $D$ a distribution over $\Sm$. The $(n,m,q,\Sigma,\ell,D)$-LWE problem is to efficiently distinguish the following two distributions:
	\[(\Am,\Am\cdot\rv+\ev)\;\;\;\text{and}\;\;\; (\Am,\Cm\cdot\rv'+\ev)\enspace ,\]
	Where $\rv$ is uniform in $\Z_q^n$, $\rv'$ is uniform in $\Z_q^{m-\ell}$, and $\ev$ is Gaussian of width $\Sigma$. We say the problem is \emph{hard} if, for all polynomial time quantum algorithms, the distinguishing advantage is negligible.
\end{problem}
In Section~\ref{sec:klwe}, we explain that in many parameter settings, including importantly the setting of~\cite{KLS22}, that the hardness of Problem~\ref{conj:klwe} is true (assuming standard lattice assumptions).

With the hardness of Problem~\ref{conj:klwe}, we can show the following, which is a generalization of a result of~\cite{C:LiuZha19} that showed that the SIS hash function is collapsing for super-polynomial modulus:

\begin{theorem}\label{thm:main} Consider sampling $\Am,\Sm$ as above, and consider any efficient algorithm that, given $\Am,\Sm$, samples a $\uv$ and a state $|\phi_\uv\rangle$ with the guarantee that all the support of $|\phi_\uv\rangle$ is on vectors $\yv$ such that (1) $\Am^T\cdot\yv=\uv\bmod q$ and (2) $|\yv|_2\leq W$. 
	
Now suppose $|\phi_\uv\rangle$ is sampled according to this process, and then either (A) $|\phi_\uv\rangle$ is produced, or (B) $|\phi_\uv'\rangle$ is produced, where $|\phi_\uv'\rangle$ is the result of applying the partial measurement of $\Cm^T\cdot\yv$ to the state $|\phi_\uv\rangle$.

Suppose there exists $\Sigma$ such that $q/W\Sigma=\omega(\sqrt{\log\lambda})$ such that $(n,m,q,\Sigma,\ell,D)$-LWE is hard. Then cases (A) and (B) are computationally indistinguishable.	
\end{theorem}
Note that an interesting consequence of Theorem~\ref{thm:main} in the case $\ell=0$ is that it shows that the SIS hash function is collapsing for any modulus, under an appropriate (plain) LWE distribution. This improves upon~\cite{C:LiuZha19}, who showed the same but only for super-polynomial modulus. We now give the proof of Theorem~\ref{thm:main}:

\begin{proof}For an integer $t$, let $\lfloor\cdot\rceil_t$ denote the function that maps a point $x\in\Z_q$ to the $z\in\{0,\lfloor q/t\rfloor,\lfloor 2q/t\rfloor,\cdot,\lfloor (t-1)q/t\rfloor\}$ that minimizes $|z-x|$. Here, $|z-x|$ is the smallest $a$ such that $z=x\pm a\bmod q$. In other words, $\lfloor\cdot\rceil_t$ is a course rounding function that rounds an $x\in\Z_q$ to one of $t$ points that are evenly spread out in $\Z_q$.
	
Let $\rho$ be a mixed quantum state, whose support is guaranteed to be on $\yv$ such that (1) $\Am^T\cdot\yv=\uv\bmod q$ and (2) $|\yv|_2\leq W$. For a quantum process $M$ acting on $\rho$, let $M(\rho)$ be the mixed state produced by applying $M_i$ to $\rho$. We will consider a few types of procedures applied to on quantum states.
	
\paragraph{$M_0$:} Given $\Am$, $M_0$ is just the partial measurement of $\yv\mapsto\Cm^T\cdot\yv$.
	
\paragraph{$M_1^t$:} Given $\Am$, to apply this measurement, first sample an LWE sample $\bv=\Am\cdot\rv+\ev$. Then apply the measurement $\yv\mapsto \lfloor \bv\cdot\yv \rceil_t$. Discard the measurement outcome, and output the remaining state.

\begin{lemma}\label{lem:m1}For any constants $t,d$, $M_1^{t}(\rho)$ is statistically close to $\frac{1}{d}M_1^{t\times d}(\rho)+\left(1-\frac{1}{d}\right)\rho$\end{lemma}
Note that Lemma~\ref{lem:m1} means that $M_1^{t}$ can be realized by the mixture of two measurements: $M_1^{t\times d}$ with probability $1/t^2$, and the identity with probability $\left(1-\frac{1}{d}\right)$. We now give the proof.
\begin{proof}Consider the action of $M_1^t$ on $|\yv\rangle\langle\yv'|$, for a constant $t$. First, an LWE sample $\bv=\Am\cdot\rv+\ev$ is chosen. Then conditioned on this sample, if $\lfloor\bv\cdot\yv\rceil_t=\lfloor\bv\cdot\yv'\rceil_t$, the output is $|\yv\rangle\langle\yv'|$. Otherwise the output is 0. Averaging over all $\bv$, we have that 
	\[M_1^t(|\yv\rangle\langle\yv'|)=\Pr_\bv[\lfloor\bv\cdot\yv\rceil_t=\lfloor\bv\cdot\yv'\rceil_t]\]
where the probability is over $\bv$ sampled as $\bv=\Am\cdot\rv+\ev$. Recalling that $\uv=\Am^T\cdot\yv=\Am^T\cdot\yv'$, we have that:
\begin{align*}
\bv\cdot\yv&=\rv\cdot\uv+\ev\cdot\yv\\
\bv\cdot\yv'&=\rv\cdot\uv+\ev\cdot\yv'
\end{align*}

Now, by our choice of $\Sigma$, $|\ev\cdot(\yv-\yv')|<q/t$ for any constant $t$, except with negligible probability. We will therefore assume this is the case, incurring only a negligible error. 

Note that $\zv:=\rv\cdot\uv$ is uniform in $\Z_q$ and independent of $\ev\cdot\yv,\ev\cdot\yv'$. So measuring $\lfloor\bv\cdot\yv\rceil_t$ is identical to measuring the result of rounding $\ev\cdot\yv$, except that the rounding boundaries are rotated by a random $\zv\in\Z_q$. Since the rounding boundaries are $q/t$ apart, at most a single rounding boundary can be between $\ev\cdot\yv$ and $\ev\cdot\yv'$, where ``between'' means lying in the shorter of the two intervals (of length $|\ev\cdot(\yv-\yv')|$) resulting by cutting the circle $\Z_q$ at the points $\ev\cdot\yv$ and $\ev\cdot\yv'$. $\lfloor\bv\cdot\yv\rceil_t=\lfloor\bv\cdot\yv'\rceil_t$ if and only if no rounding boundary is between them.

Since the cyclic shift $\zv$ is uniform each rounding boundary is uniform. Since there are $t$ rounding boundaries and no two of them can between $\ev\cdot\yv$ and $\ev\cdot\yv'$, we have that, conditioned on $\ev$, the probability $\lfloor\bv\cdot\yv\rceil_t\neq\lfloor\bv\cdot\yv'\rceil_t$ is therefore $\frac{t}{q}|\ev\cdot(\yv-\yv')|$. Averaging over all $\ev$, we have that, up to negligible error:
	\[M_1^t(|\yv\rangle\langle\yv'|)=\left(1-\frac{t}{q}\E_\ev[|\ev\cdot(\yv-\yv')|]\right)|\yv\rangle\langle\yv'|\]
	
Notice then that $M_1^{t}(|\yv\rangle\langle\yv'|)=\frac{1}{d}M_1^{t\times d}(|\yv\rangle\langle\yv'|)+\left(1-\frac{1}{d}\right)|\yv\rangle\langle\yv'|$. By linearity, we therefore prove Lemma~\ref{lem:m1}.
\end{proof}

Note that the proof of Lemma~\ref{lem:m1} also demonstrates that $M_0$ and $M_1^t$ commute, since their action on density matrices is just component-wise multiplication by a fixed matrix.

\paragraph{$M_2^t$:} Given $\Am$, to apply this measurement, first sample an LWE sample $\bv=\Cm\cdot\rv'+\ev$. Then apply the measurement $\yv\mapsto \lfloor \bv\cdot\yv \rceil_t$. Let $p_t$ be the probability that $\lfloor x\rceil_t=\lfloor y\rceil_t$ for uniformly random $x,y\in\Z_q$. Note that for any constant $t$, $p_t\leq t^{-1}+O(q^{-1})$.

\begin{lemma}\label{lem:m2}For any constant $t$, $M_2^t(\rho)$ is statistically close to $M_0(M_1^t(\rho))+p_t(\rho-M_0(\rho))$.\end{lemma}
Note that unlike Lemma~\ref{lem:m1}, the expression in Lemma~\ref{lem:m2} does not correspond to a mixture of measurements applied to $\rho$. However, we will later see how to combine Lemma~\ref{lem:m2} with Lemma~\ref{lem:m1} to obtain such a mixture.
\begin{proof}The proof proceeds similarly to Lemma~\ref{lem:m1}. We consider the action of $M_2^t$ on $|\yv\rangle\langle\yv'|$, and conclude that
	\[M_2^t(|\yv\rangle\langle\yv'|)=\Pr_\bv[\lfloor\bv\cdot\yv\rceil_t=\lfloor\bv\cdot\yv'\rceil_t]\]
where the probability is over $\bv=\Cm\cdot\rv'+\ev$. But now we have that
\begin{align*}
\bv\cdot\yv&=\rv'^T\cdot\Cm^T\yv+\ev\cdot\yv\\
\bv\cdot\yv'&=\rv'^T\cdot\Cm^T\yv+\ev\cdot\yv'
\end{align*}

We consider two cases:
\begin{itemize}
	\item $\Cm^T\cdot\yv=\Cm^T\cdot\yv'$. This case is essentially identical to the proof of Lemma~\ref{lem:m1}, and we conclude that $\Pr_\bv[\lfloor\bv\cdot\yv\rceil_t=\lfloor\bv\cdot\yv'\rceil_t]=1-\frac{t}{q}\E_\ev[|\ev\cdot(\yv-\yv')|]$. Note that for such $\yv,\yv'$, we also have \[M_0(M_1^t(|\yv\rangle\langle\yv'|))+p_t(|\yv\rangle\langle\yv'|-M_0(|\yv\rangle\langle\yv'|))=M_1^t(M_0(|\yv\rangle\langle\yv'|))+p_t\times 0=1-\frac{t}{q}\E_\ev[|\ev\cdot(\yv-\yv')|]\enspace ,\] 
	since $M_0$ is the identity on such $|\yv\rangle\langle\yv'|$. Thus, we have the desired equality for $\rho=|\yv\rangle\langle\yv'|$.
	\item $\Cm^T\cdot\yv\neq\Cm_T\cdot\yv'$. In this case, $\bv\cdot\yv$ and $\bv\cdot\yv'$ are independent and uniform over $\Z_p$. Therefore, $\Pr_\bv[\lfloor\bv\cdot\yv\rceil_t=\lfloor\bv\cdot\yv'\rceil_t] =p_t$. Note that for such $\yv,\yv'$, we also have
	\[M_0(M_1^t(|\yv\rangle\langle\yv'|))+p_t(|\yv\rangle\langle\yv'|-M_0(|\yv\rangle\langle\yv'|))=0+p_t |\yv\rangle\langle\yv'|\enspace ,\]
	since $M_0(|\yv\rangle\langle\yv'|)=0$ in this case.
\end{itemize}
Thus for each $|\yv\rangle\langle\yv'|$, we have the desired equality. By linearity, this thus extends to all $\rho$.
\end{proof}
Combining Lemmas~\ref{lem:m1} and~\ref{lem:m2}, we obtain:
\begin{corollary}\label{cor:m}For any constants $t,d$, $M_2^{t}(\rho)$ is statistically close to $\frac{1}{d}M_0(M_1^{t\times d}(\rho))+\left(1-\frac{1}{d}-p_t\right)M_0(\rho)+p_t \rho$. 
\end{corollary}
For $d$ such that $1-\frac{1}{d}-p_t\geq 0$, this represents a mixture of measurements $M_0\circ M_1^{t\times d}, M_0$, and the identity.

We are now ready to prove Theorem~\ref{thm:main}. Suppose there is an algorithm $\adv$ that constructs a mixed state $\rho$, and then can distinguish $\rho$ from $M_0(\rho)$ with (signed) advantage $\epsilon$. Let $d$ be a positive integer, to be chosen later. Let $\rho_0=\rho$, and $\rho_i=M_1^{t\times d}(\rho_{i-1})$. Note that for any polynomial $i$, $\rho_i$ can be efficiently constructed. Let $\epsilon_0=\epsilon$, and $\epsilon_i$ be the (signed) distinguishing advantage of $\adv$ when given $\rho_i$ vs $M_0(\rho_i)$. 

Let $\delta_i$ be the (signed) distinguishing advantage of $\adv$ for $M_2^t(\rho_i)$ and $M_1^t(\rho_i)$. Write $g=1-\frac{1}{d}-p_t$. Invoking Lemma~\ref{lem:m1} and Corollary~\ref{cor:m} with $d$, we have that 
\[\delta_i=\frac{1}{d}\epsilon_{i+1}+g\epsilon_i\]

Now, we note that $\delta_i$ must be negligible, by the assumed hardness of $(n,m,q,\Sigma,\ell,D)$-LWE. Solving the recursion gives:
\[\epsilon_i(-dg)^{-i}=\epsilon-\frac{1}{d}\sum_{j=0}^{i-1}(-dg)^{-j} \delta_{j+1}\]

Next, assume $d$ is chosen so that $dg$ is a constant greater than 1. Define $T=\sum_{j=0}^{\lambda-1} (dg)^{-j}=\frac{dg}{dg-1}- 2^{-O(\lambda)}$. Consider the adversary $\adv'$ for $(n,m,q,\Sigma,\ell,D)$-LWE, which does the following:
\begin{itemize}
	\item On input $\Am,\Sm,\bv$, where $\bv=\Am\cdot\rv+\ev$ or $\bv=\Cm\cdot\rv'+\ev$, it chooses $j\in[0,\lambda-1]$ with probability $(dg)^{-j}/T$
	\item Then it constructs $\rho$ according to $\adv$. 
	\item Next, $\adv'$ computes $\rho_j$ by applying $M_1^{t\times d}$ to $\rho$ for $j$ times.
	\item Now $\adv'$ applies the measurement $\yv\mapsto \lfloor \bv\cdot\yv \rceil_t$ to $\rho_j$, obtaining $\rho_j'$.
	\item $\adv'$ runs the distinguisher for $\adv$, obtaining a bit $b$
	\item $\adv'$ outputs $b$ if $j$ is even, $1-b$ if $j$ is odd.
\end{itemize}
Note that if $\bv$ is $\Am\cdot\rv+\ev$, then $\rho'=M_1^t(\rho_i)$, and if $\bv$ is $\Cm\cdot\rv'+\ev$, then $\rho'=M_2^t(\rho_i)$. Therefore, the distinguishing advantage of $\adv'$ is:
\[\delta=\frac{1}{T}\sum_{j=0}^{\lambda-1} (-dg)^{-j}\delta_{j+1}\]

Thus, we have that \[\epsilon_\lambda(-dg)^{-\lambda}=\epsilon-\frac{T}{d}\delta\enspace ,\]
Noting that $\epsilon_\lambda$ must trivially be in $[-1/2,1/2]$, we have that:
 \[|\delta|\geq\frac{d}{T}\left(|\epsilon|-\frac{1}{2}(dg)^{-\lambda}\right)\geq d\left(1-\frac{1}{dg}\right)|\epsilon|-2^{-O(\lambda)}\]
Thus, if $\adv$ has non-negligible distinguishing advantage, so does $\adv'$, breaking the $(n,m,q,\Sigma,\ell,D)$-LWE assumption. This completes the proof of Theorem~\ref{thm:main}.\end{proof}

\subsection{Constructing $|\psi_\uv'\rangle$} 
\label{sec:constructingpsi}

Here, we explain how to construct $|\psi_\uv'\rangle$, given just the vector $\yv$ that resulted from measuring it. We first observe that, since $|\psi_\uv'\rangle$ has support only on vectors that differ from $\yv$ by multiples of the columns of $\Sm$, we can write:
	\[|\psi_\uv'\rangle\propto\sum_{\tv}\alpha_{\yv+\Sm\cdot\tv}|\yv+\Sm\cdot\tv\rangle\]
	
Where $\alpha_{\yv}$ is the amplitude of $\yv$ in $|\psi\rangle$. This gives a hint as to how to construct $|\psi_\uv'\rangle$: create a superposition over short linear combinations of $\Sm$, and then use linear algebra to transition to a superposition over $\yv+\Sm\cdot\tv$, weighted according to $\alpha$. The problem of course is that $\alpha$ may be arbitrary except for having support only on short vectors. Therefore, we do not expect to be able to construct $|\psi_\uv'\rangle$ in full generality, and instead focus on special (but natural) cases, which suffice for our use.

\paragraph{Wide Gaussian Distributed.} Suppose the initial state $|\psi\rangle$ is the discrete Gaussian over the integers: $|\psi\rangle=|\mathbf{\Psi}_{\Z^m,\Sigma,\cv}\rangle$ for some center $\cv$ and covariance matrix $\Sigma$. Then $|\psi_\uv'\rangle$ is simply \[|\mathbf{\Psi}_{\Ls+\yv,\Sigma,\cv}\rangle\]
Here, $\Ls$ is the integer lattice generated by the columns of $\Sm$, and $\Ls+\yv$ is the lattice $\Ls$ shifted by $\yv$. We can construct the state $|\mathbf{\Psi}_{\Ls+\yv,\Sigma,\cv}\rangle$ by first constructing $|\mathbf{\Psi}_{\Ls,\Sigma,\cv-\yv}\rangle$, and then adding $\yv$ to the superposition. Thus, as long as $\sv_i^T\cdot\Sigma^{-1}\cdot\sv_i\leq 1/\omega(\sqrt{\log\lambda})$ for all $i$, we can construct the necessary state.

\paragraph{Constant Dimension, Hyper-ellipsoid Bounded.} Here, we restrict $\Ls$ to having a constant number of columns, but greatly generalize the distributions that can be handled.

A hyper-ellipsoid is specified by a positive definite matrix $\Sigma$, which defines the set $E_{\Sigma,\cv}=\{\yv:(\yv-\cv)^T\cdot \Mm\cdot(\yv-\cv)\leq 1\}$.

\begin{definition}[Good Hyper-ellipsoid] A \emph{good hyper-ellipsoid} for $|\psi\rangle$ is an $E_{\Sigma,\cv}$ such that there exists a function $\eta(\lambda)$ and polynomials $p(\lambda),q(\lambda)$ such that, if $|\psi\rangle$ is measured to get a vector $\yv$, then each of the following are true except with negligible probability:
	\begin{itemize}
		\item $\yv\in E_{\Sigma,\cv}$. In other words, $E_{\Sigma,\cv}$ contains essentially all the mass of $|\psi\rangle$.
		\item $|\alpha_\xv|^2\leq\eta(\lambda)$. In other words, $\eta$ is an approximate upper bound on $\alpha_{\xv}$.
		\item If a random vector $\xv$ is chosen from $E_{\Sigma,\cv}\cap \{\yv+\Sm\cdot\tv:\tv\in\Z^\ell\}$, then with probability at least $1/p(\lambda)$, $|\alpha_{\xv}|^2\geq\eta/q(\lambda)$. In other words, $E_{\Sigma,\cv}$ doesn't contain too many points with mass too much lower than $\eta$.
	\end{itemize}
\end{definition}
Taken together, a good hyper-ellipsoid is one that fits reasonably well around the $|\psi\rangle$. It must contain essentially all the support of $|\psi\rangle$, but can over-approximate it by a polynomial factor.

\begin{lemma}\label{lem:hyper} Suppose there is a good hyper-ellipsoid for $|\psi\rangle$, and that $\alpha_{\yv}$ can be efficiently computed given any vector $\yv$. Then there is a polynomial-time algorithm which constructs $|\psi_\uv'\rangle$ from $\yv$\end{lemma}
\begin{proof}Let $E_{\Sigma,\cv}$ be the good hyper-ellipsoid. Let $\Ls$ be the lattice generated by the columns of $\Sm$. By assumption, with overwhelming probability if we measure $|\psi\rangle$ to get $\yv$, we have $\yv\in E_{\Sigma,\cv}$. Let $E_{\Sigma',\cv'}$ be the ellipsoid that is the intersection of $E_{\Sigma,\cv}$ and the affine space $\{\yv+\Sm\cdot\tv:\tv\in\R^\ell\}$.
	
\begin{claim}There is PPT algorithm which, given $\Sm,\Sigma'$, computes $\Tm=\{\rv_1,\cdots,\rv_{\ell'}\}$ such that:
	\begin{itemize}
		\item $\rv_i^T\cdot(\Sigma')^{-1}\cdot\rv_i\leq 2$ for all $i\in[\ell']$, and
		\item $E_{\Sigma',\cv'}\cap\{\yv+\Tm\cdot\tv:\tv\in\Z^{\ell'}\}=E_{\Sigma',\cv'}\cap\{\yv+\Sm\cdot\tv:\tv\in\Z^\ell\}$.
	\end{itemize}
\end{claim}
\begin{proof}Write $(\Sigma')^{-1}$ as $(\Sigma')^{-1}=\Um^T\cdot\Um$. Let $\Sm'=\{\sv_1'=\Um\cdot\sv_1,\dots,\sv_\ell'=\Um\cdot\sv_n\}$, and let $\Ls'$ be the lattice generated by $\Sm'$. Since $\ell$ is constant, we can find shortest vectors in $\Ls'$ in polynomial time. Therefore, compute $\rv_1',\dots,\rv_\ell'$ such that $\rv_i'$ is the shortest vector in $\Ls'$ that is linearly independent from $\{\rv_1',\dots,\rv_{i-1}'\}$. Then let $\ell'$ be such that $|\rv'_{\ell'}|^2\leq 2$, but $|\rv'_{\ell'+1}|^2>2$, or $\ell'=\ell$ if no such $\ell'$ exists.
	
Finally, let $\rv_i=\Um^{-1}\cdot\rv'_i$. Clearly, we have that $\rv_i^T\cdot(\Sigma')^{-1}\cdot\rv_i\leq 2$. It remains to show that $E_{\Sigma',\cv'}\cap\{\yv+\Tm\cdot\tv:\tv\in\Z^{\ell'}\}=E_{\Sigma',\cv'}\cap\{\yv+\Sm\cdot\tv:\tv\in\Z^\ell\}$. First, we notice that the lattice $\Ls(\Tm)$ spanned by $\Tm$ is a sub-lattice of $\Ls(\Sm)$ spanned by $\Sm$. So one containment is trivial. Now assume toward contradiction that there is a $\xv\in E_{\Sigma',\cv'}\cap\{\yv+\Sm\cdot\tv:\tv\in\Z^\ell\}$ that is not in $E_{\Sigma',\cv'}\cap\{\yv+\Tm\cdot\tv:\tv\in\Z^{\ell'}\}$. This means $\xv-\yv$ is in  $\Ls(\Sm)$. We also have that $(\yv-\cv')^T\cdot(\Sigma')^{-1}\cdot(\yv-\cv')\leq 1$ (since and $(\xv-\cv')^T\cdot(\Sigma')^{-1}\cdot(\xv-\cv')\leq 1$. By the triangle inequality, we have therefore that $(\xv-\yv)^T\cdot(\Sigma')^{-1}\cdot(\xv-\yv)\leq 2$. 

But then we have that $\Um\cdot(\xv-\yv)$ has norm at most $2$, lies in $\Ls'$, and is linearly independent of $\{\rv_1',\dots,\rv_{\ell'}'\}$. This contradicts that $\rv'_{\ell'+1}$ (which has norm squared strictly greater than 2) is a shortest vector linearly independent of $\{\rv_1',\dots,\rv_{\ell'}'\}$. This completes the proof of the claim.\end{proof}

We now return to proving Lemma~\ref{lem:hyper}. Let $\beta=\omega(\log\lambda)$. We construct $|\psi_\uv'\rangle$ in three steps:
\begin{itemize}
	\item We first construct a state negligibly close to $|\mathbf{\Psi}_{\Ls+\yv,\beta\Sigma',\cv'}\rangle$, as we did in the Gaussian-distributed case above.
	\item We then construct the state $|E\rangle$, defined as the uniform superposition over the intersection of $\Ls+\yv$ and $E_{\Sigma',\cv'}$. $|E\rangle$ will be obtained from $|\mathbf{\Psi}_{\Ls+\yv,\beta\Sigma',\cv'}\rangle$ via a measurement.
	\item Construct $|\psi_\uv'\rangle$ from $|E\rangle$. This also will be obtained via a measurement.
\end{itemize}

We now describe the two measurements. We start from the second. Let $\eta,p,q$ be the values guaranteed by the goodness of $E_{\Sigma,\cv}$. Define $\eta_\xv=1/\eta$ if $|\alpha_{\xv}|^2\leq \eta$, and otherwise $\eta_\xv=1/|\alpha_{\xv}|^2$. To obtain $|\psi_\uv'\rangle$ from $|E\rangle$, we apply the following map in superposition and measure the second register:
\[|\xv\rangle\mapsto|\xv\rangle\left(\sqrt{\eta_\xv}\alpha_{\xv}|0\rangle+\sqrt{1-|\eta_\xv\alpha_{\xv}|^2}|1\rangle\right)\]

Suppose for the moment that $\eta_\xv=1/\eta$ for all $\xv$. Then conditioned on the measurement outcome being 0, the resulting state is exactly $|\psi_\uv'\rangle$. By the guarantee that $E_{\Sigma,\cv}$ is good, we have that except with negligible probability over the choice of $\yv$, all but a negligible fraction of the support of $|\psi_\uv'\rangle$ satisfies $\eta_\xv=1/\eta$. Therefore, we will assume (with negligible error) this is the case. The probability the measurement is 0 (over the choice of $\yv$ as well) is $\E_{\xv\gets E_{\Sigma',\cv'}}[\alpha_{\xv}^2/\eta]$, which, with probability at least $1/p$ over the choice of $\yv$, is at least $1/q$. Thus, the overall probability of outputting 0 is inverse polynomial, and in this case we produce a state negligibly close to $|\psi_\uv'\rangle$.

It remains to construct $|E\rangle$ from $|\mathbf{\Psi}_{\Ls+\yv,\beta\Sigma',\cv'}\rangle$. This follows a very similar rejection-sampling argument. Let \[\gamma_\xv=\begin{cases}e^{-\pi/\beta}\times\sqrt{e^{\pi (\xv-\cv')^T\cdot(\beta\Sigma')^{-1}\cdot(\xv-\cv')}}&\text{ if }(\xv-\cv')^T\cdot(\Sigma')^{-1}\cdot(\xv-\cv')\leq 1\\0&\text{ otherwise}\end{cases}\]

Note that $0\leq\gamma_\xv\leq 1$. Now apply to $|\mathbf{\Psi}_{\Ls+\yv,\beta\Sigma',\cv'}\rangle$ the map $|\xv\rangle\mapsto|\xv\rangle(\gamma_\xv|0\rangle+\sqrt{1-\gamma_\xv^2}|1\rangle)$, and measure the second coordinate. If the measurement outcome is 0, then the resulting state is exactly $|E\rangle$. For $\xv\in E_{\Sigma',\cv'}$, we have $\gamma_\xv\geq e^{-\pi/\beta}\geq 1-o(1)$. Therefore, the probability the measurement outputs 0 is at least $1-o(1)$ times the probability measuring $\mathbf{\Psi}_{\Ls+\yv,\beta\Sigma',\cv'}$ produces an $\xv\in E_{\Sigma',\cv'}$. This latter probability is $O_\ell(\beta^{-\ell/2})$, where the constant hidden by the big $O$ depends on $\ell$. Since $\ell$ is constant and $\beta$ is polynomial (in fact, sub-polynomial), the overall probability is polynomial. This completes the construction of $|\psi_\uv'\rangle$ and the proof of Lemma~\ref{lem:hyper}.\end{proof}

\subsection{Applying to \cite{KLS22}}
\label{sec:attack_KLS_main}

We can then adapt the quantum money scheme in \cite{KLS22} into the above general framework.

\ifllncs
To avoid confusion, we first refer the readers to our alternate view on \cite{KLS22} scheme in \ref{sec:alternate}. Then we will see how to apply our attack onto their scheme in \ref{sec:attack_KLS}.
\fi


\ifllncs
\subsection{Applying Our General Attack to \cite{KLS22}}
\label{sec:attack_KLS}
We now show that we can adapt the quantum money scheme in \cite{KLS22} into the general framework presented in Section \ref{sec:attack}.

\else

To apply to \cite{KLS22}, we apply the above with $n=1$.

First, recall that the first few entries of $\vv'$ are $-\Gamma=(2t\sigma\Delta+1)/2,-1,1,\Delta/2$, and the rest are random. We then note that we can always arbitrarily re-scale the vector $\vv'$, without affecting the scheme, since it only permutes the serial numbers but does not change the scheme. Then we see that $\vv'$ is a random vector, subject to being orthogonal to the following three vectors:
\[\sv_0=\left(\begin{array}{c}0\\0\\\Delta\\-2\\0\\\vdots\\0\end{array}\right)\;\;\;\sv_1=\left(\begin{array}{c}0\\1\\1\\0\\0\\\vdots\\0\end{array}\right)\;\;\;\sv_2=\left(\begin{array}{c}-2\\2\Gamma=2t\sigma\Delta+1\\0\\0\\0\\\vdots\\0\end{array}\right)\]

A randomly re-scaled $\vv'$ is just a random vector orthogonal to each of these three vectors, which are short and linearly independent.

\medskip

So our alternative view of~\cite{KLS22} is an instance of the above the general scheme with $n=1$ and short vectors $\Sm=(\sv_1,\sv_2,\sv_3)$. We also note that we can easily construct a good ellipsoid for their $|\psi\rangle$, which has axis lengths $4,2k,\sigma\omega(\log\lambda),\cdots,\sigma\omega(\log{\lambda})$. By Lemma~\ref{lem:hyper}, we can therefore construct as many copies as we would like of the state $|\psi_\uv'\rangle$, which will fool verification under Conjecture~\ref{conj:klwe}. In the next section, we demonstrate that Conjecture~\ref{conj:klwe} holds for their particular parameter settings. Thus their scheme is insecure.

\begin{remark} Note that in our re-conceptualized version of~\cite{KLS22}, the superpositions always have support on vectors whose the first entry equal to 0 or 1. As such, there is no non-zero small multiple of $\sv_1$ such that adding this multiple results results in another vector whore first entry is 0 or 1. This means that, in our attack, the superposition $|\psi_\uv'\rangle$ will actually only have support on shifts by multiples of $\sv_0,\sv_1$. 
\end{remark}

\fi

\section{Invariant Money}
\label{sec:invariant_main}
From this section on, we discuss our positive results on quantum money/lightning.

We now describe our framework for instantiating quantum money using invariants, or more precisely what we call \emph{walkable} invariants.

 Let $X,Y$ be sets, and $I:X\rightarrow Y$ an efficiently computable function from $X$ to $Y$. $I$ will be called the ``invariant.'' We will additionally assume a collection of permutations $\sigma_i:X\rightarrow X$ indexed by $i\in [r]$ for some integer $r$, with the property that the permutations respect the invariant:
\[I(\;\sigma_i(x)\;)=I(x), \forall i\in[r]\]
In other words, action by each $\sigma_i$ preserves the value of the invariant. We require that $\sigma_i$ is efficiently computable given $i$. $r$ may be polynomial or may be exponential. To make the formalism below simpler, we will be implicitly assuming that there exists a perfect matching between the $\sigma_i$ such that for any matched $\sigma_i,\sigma_{i'}$, we have $\sigma_{i'}=\sigma_{i}^{-1}$. Moreover, $i'$ can be found given $i$. This can be relaxed somewhat to just requiring that $\sigma_i^{-1}$ can be efficiently computed given $i$, but requires a slightly more complicated set of definitions.

Given a point $x$, the orbit of $x$, denoted $O_x\subseteq X$, is the set of all $z$ such that there exists a non-negative integer $k$ and $i_1,\dots,i_k\in[r]$ such that $z=\sigma_{i_k}(\sigma_{i_{k-1}}(\cdots\sigma_{i_1}(x)\cdot ))$. In other words, $O_x$ is the set of all $z$ ``reachable'' from $x$ by applying some sequence of permutations. Note that $I(z)=y$ for any $z\in O_x$. We will therefore somewhat abuse notation, and define $I(O_x)=y$. We also let $P_y$ be the set of pre-images of $y$: $P_y=\{x\in X:I(x)=y\}$.

We will additionally require a couple properties, which will be necessary for the quantum money scheme to compile:
\begin{itemize}
	\item {\bf Efficient Generation of Superpositions:} It is possible to construct the uniform superposition over $X$: $|X\rangle:=\frac{1}{\sqrt{|X|}}\sum_{x\in X}|x\rangle$.
	\item {\bf Mixing Walks:} For an orbit $O$, with a slight abuse of notation let $\sigma_{O,i}$ be the (possibly exponentially large) permutation matrix associated with the action by $\sigma_i$ on $O$. Then let $M_O=\frac{1}{r}\sum_{i\in[r]} \sigma_{O,i}$ be the component-wise average of the matrices. Let $\lambda_1(O),\lambda_2(O)$ be the largest two eigenvalues by absolute value\footnote{They are real-valued, since $M_O$ is symmetric, owing to the fact that we assumed the $\sigma_i$ are perfectly matched into pairs that are inverses of each other.}, counting multiplicities. Note that $\lambda_1(O)=1$, with corresponding eigenvector the all-1's vector. We need that there is an inverse polynomial $\delta$ such that, for every orbit $O$, $\lambda_2(O)\leq 1-\delta$. This is basically just a way of saying that a random walk on the orbit using the $\sigma_i$ mixes in polynomial time. 
\end{itemize}
We call such a structure above a walkable invariant.

\subsection{Quantum Money from Walkable Invariants}
\label{sec:qm_scheme_invariant}

We now describe the basic quantum money scheme.

\paragraph{Minting.} To mint a note, first construct the uniform superposition $|X\rangle$ over $X$. Then apply the invariant $I$ in superposition and measure, obtaining a string $y$, and the state collapsing to:
\[|P_y\rangle:=\frac{1}{\sqrt{|P_y|}} \sum_{x\in P_y}|x\rangle\]
This is the quantum money state, with serial number $y$.

\paragraph{Verification.} To verify a supposed quantum money state $|\phi\rangle$ with serial number $y$, we do the following.
\begin{itemize}
	\item First check that the support of $|\phi\rangle$ is contained in $P_y$. This is done by simply applying the invariant $I$ in superposition, and measuring if the output is $y$. If the check fails immediately reject.
	\item Then apply the projective measurement given by the projection $\sum_{O\subseteq P_y}|O\rangle\langle O|$, where $O$ ranges over the orbits contained in $P_y$, and $|O\rangle:=\frac{1}{\sqrt{|O|}}\sum_{x\in O}|x\rangle$. In other words, project onto states where, for each orbit, the weights of $x$ in that orbit are all identical; weights between different orbits are allowed to be different.
	
	We cannot perform this measurement exactly, but we can perform it approximately using the fact that $\lambda_2(O)\leq 1-\delta$. This is described in Section~\ref{sec:verification_main} below. Outside of Section~\ref{sec:verification_main}, we will assume for simplicity that the measurement is provided exactly.
	
	If the projection rejects, reject the quantum money state. Otherwise accept.
\end{itemize}

It is hopefully clear that honestly-generated money states pass verification. Certainly their support will be contained in $P_y$, and they apply equal weight to each element in an orbit (and in fact, equal weight across orbits).

\subsection{Approximate Verification}
\label{sec:verification_main}

Here, we explain how to approximately perform the verification projection $V=\sum_{O\subseteq P_y}|O\rangle\langle O|$, using the fact that $\lambda_2(0)\leq 1-\delta$ for all $O$. The algorithm we provide is an abstraction of the verification procedure of~\cite{ITCS:FGHLS12}, except that that work presented the algorithm without any analysis. We prove that the algorithm is statistically close to the projection $V$, provided the mixing condition $\lambda_2(0)\leq 1-\delta$ is met.

\begin{theorem}\label{thm:verification}Assume $\lambda_2(0)\leq 1-\delta$ for all $O$, for some inverse-polynomial $\delta$. Then there is a QPT algorithm $\tilde{V}$ such that, for any state $|\psi\rangle$, if we let $|\psi'\rangle$ be the un-normalized post-measurement state from applying $\tilde{V}$ to $|\psi\rangle$ in the case $\tilde{V}$ accepts, then $|\psi'\rangle$ is negligibly close to $V|\psi\rangle$.\end{theorem}


\medskip

\ifllncs
We refer the readers to \ref{sec:verification} for the proof due to restriction on the space.

\else

\begin{proof}Let $r'=2r$. Let $\Rs$ be a register containing a superposition over $1,\dots,r'$, and $\Ss$ be the register containing the supposed quantum money state. Define the following:
\begin{itemize}
	\item The unitary $U$ acting on $\Rs\otimes\Ss$ defined as $U=\sum_{i=1}^r |i\rangle\langle i|\otimes \sigma_i + \sum_{i=r+1}^{r'} |i\rangle\langle i|\otimes \Id$. Here, we use a slight abuse of notation, using $\sigma_i$ to denote the unitary implementing the classical permutation $\sigma_i$. $U$ can be computed efficiently, by our assumption that we can efficiently invert $\sigma_i$.
	\item The state $|\onev\rangle=\frac{1}{\sqrt{r'}}\sum_{i=1}^{r'}|i\rangle$
	\item The projection $P=|\onev\rangle\langle\onev|\otimes\Id$. 
\end{itemize}

\noindent Our algorithm is the following:
\begin{itemize}
	\item Initialize the register $\Rs$ to $|\onev\rangle$.
	\item Repeat the following $t=\lambda/\delta$ times:
	\begin{itemize}
		\item Apply the unitary $U$ to $\Rs\otimes\Ss$.
		\item Apply the measurement corresponding to projection $P$ to $\Rs\otimes\Ss$. If the measurement outcome is reject, immediately abort and reject the state.
		\item Apply $U^{-1}$.
	\end{itemize}
	\item If all $t$ trials above accepted, then accept and output the contends of the $\Ss$ register.
\end{itemize}

We now analyze the algorithm above. $M=\frac{1}{r}\sum_{i\in[r]} \sigma_i$. This matrix is symmetric, since we assumed the $\sigma_i$ come in pairs that are inverses of each other. It is not necessarily positive. For example, consider $X$ that can be divided into two subsets $X_0,X_1$ such that each $\sigma_i$ maps $X_b$ to $X_{1-b}$. Then the vector that places equal weight on all $x\in X$, but makes the values in $X_0$ positive and $X_1$ negative will have eigenvalue $-1$. In general, the eigenvalues of $M$ must be in the real interval $[-1,1]$. The eigenvectors with eigenvalue 1 are exactly $|O\rangle$ and the space spanned by them, as $O$ ranges over all possible orbits. By our mixing assumption, all other eigenvalues are at most $1-\delta$.

Suppose $|\psi_\ell\rangle$ is an eigenvector of $M$, with eigenvalue $\ell$. We now explore how the algorithm above behaves on $|\psi_\ell\rangle$. Then we will extend our understanding to non-eigenvector states.

We have that 
\begin{align*}
PU|\onev\rangle|\psi_\ell\rangle&=\frac{1}{\sqrt{r'}}\left(|\onev\rangle\langle\onev|\otimes\Id\right)\cdot\left(\sum_{i=1}^r |i\rangle\otimes \sigma_i|\psi_\ell\rangle + \sum_{i=r+1}^{r'} |i\rangle  |\psi_\ell\rangle\right)\\
&=\frac{1}{r'}|\onev\rangle\otimes\left(\sum_{i=1}^r \sigma_i |\psi_\ell\rangle+r |\psi_\ell\rangle\right)\\
&=|\onev\rangle\otimes\left(\frac{1}{2}M|\psi_\ell\rangle+\frac{1}{2}|\psi_\ell\rangle\right)=\left(\frac{1+\ell}{2}\right)|\onev\rangle|\psi_\ell\rangle
\end{align*}

Now consider a general state $|\psi\rangle$, which we write in an eigenbasis for $M$ as $|\psi\rangle=\sum_{\ell}\alpha_\ell|\psi_\ell\rangle$. Then $PU|\onev\rangle|\psi\rangle=\sum_{\ell}\alpha_\ell\left(\frac{1+\ell}{2}\right)|\psi_\ell\rangle$. After $t$ trials, we have:
\[(PU)^t|\onev\rangle|\psi\rangle=|\onev\rangle\otimes\sum_{\ell}\alpha_\ell\left(\frac{1+\ell}{2}\right)^t|\psi_\ell\rangle\]
Let $|\psi'\rangle$ be the state after discarding $|\onev\rangle$. Now
\begin{align*}|\enspace V|\psi\rangle-|\psi'\rangle\enspace |^2&=\left|\enspace|\psi_\ell\rangle-\sum_{\ell}\alpha_\ell \left(\frac{1+\ell}{2}\right)^t |\psi_\ell\rangle\enspace\right|^2\\
&=\left| -\sum_{\ell\neq 1}\alpha_\ell \left(\frac{1+\ell}{2}\right)^t |\psi_\ell\rangle\enspace\right|^2\\
&=\sum_{\ell\neq 1}|\alpha_\ell|^2 \left(\frac{1+\ell}{2}\right)^{2t}\\
&\leq \sum_{\ell\neq 1}|\alpha_\ell|^2 \left(1-\delta/2\right)^{2t}\\
&=(1-|\alpha_1|^2)\left(1-\delta/2\right)^{2t}\leq e^{-\lambda}
\end{align*}
This bound is negligible. This completes the proof of Theorem~\ref{thm:verification}.
\end{proof}

\fi

\ifllncs
\subsection{Hardness Assumptions}
We rely on two hardness assumptions in our inviant money scheme:  the \emph{path-finding assumpion} and the \emph{knowledge of path} assumption. Due to space constraints, we refer the readers to \ref{sec:hardness_assumptions} for the presentation on our hardness assumptions needed.

Informally speaking, the path-finding assumption states that, given some adversarially sampled $x$ in a set $X$, it is hard for any efficient adversary, given a random $x' \in X$ such that there exists some $\sigma$ such that $\sigma \left(x \right) = x'$, to find such a $\sigma$.  

The knowledge of path assumption can be thought of as a quantum analogue to the (classical) knowledge of exponent assumption.  We define two different versions of the knowledge of path assumption to account for the fact that some of our invariants could be invertible. 

\subsection{Security}
\label{sec:security_from_kop_proof_main}
\begin{theorem}Assuming the Path-Finding assumption (Assumption~\ref{assump:pathfinding}) and the Knowledge of Path Assumption (Assumption~\ref{assump:kop1}), the scheme above is secure quantum lightning. If the invariant is invertible, then assuming the Path-Finding assumption (Assumption~\ref{assump:pathfinding}), the Knowledge of Path Assumption for Invertible Invariants (Assumption~\ref{assump:kop2}), and the Inversion Inverting assumption (Assumption~\ref{assump:invertinvert}), the scheme above is secure quantum lightning.
\end{theorem}

We refer the readers to \ref{sec:security_from_kop_proof} for the proof on the above theorem.

\else
\subsection{Hardness Assumptions}
\label{sec:hardness_assumptions}

To get an intuition for security, we define assumptions which, together, imply the quantum money (even quantum lightning) scheme is secure. We first introduce some notation. Let $p\in[r]^k$ for some $k$. Given a starting element $x\in X$, we will interpret $p$ as a path from $x$, leading to an element $z=\sigma_{p_k}(\sigma_{p_{k-1}}(\cdots\sigma_{p_1}(x)\cdot ))$. We will therefore call $p$ a path from $x$ to $z$. We will say that $p$ is a path \emph{between} $x$ and $z$ if $p$ is a path from $x$ to $z$ or a path from $z$ to $x$.

\paragraph{Hardness of Path-finding.} This is an analog of discrete log, but for our setting. When thinking of $X$ as elliptic curves and the $\sigma$ as isogenies, path-finding is just the problem of computing isogenies between elliptic curves, which is presumably hard.

In our setting, the assumption says: it should be hard for any quantum algorithm, given points in the same orbit, to find a path between them.

\begin{assumption}\label{assump:pathfinding} Consider an adversary $\adv$ playing the following game: 
	\begin{itemize}
		\item The adversary outputs an $x\in X$.
		\item The challenger then computes a random $z\in O_x$. 
		\item The adversary wins if it can output a path $p$ between $x$ to $z$.
	\end{itemize}
	The \emph{path-finding assumption} is that, for all quantum polynomial-time adversaries $\adv$, the probability $\adv$ wins in the above game is negligible.
\end{assumption}
\begin{remark}Note that technically the challenger in the path-finding game might be inefficient, since it is required to sample a random $z$ and no such explicit procedure is provided by an invariant. However, the game can be made efficient using the statistical property that $\lambda_2(O_x)\leq 1-\delta$: it can simply do a random walk on $O_x$ to compute $z$. Such $z$ will be statistically close to uniform.
\end{remark}

\paragraph{Knowledge of Path.} This is an analog of the ``knowledge of exponent'' assumption due to Damg\r{a}rd~\cite{C:Damgaard91}. The knowledge of exponent assumption (KEA) says that, given $(g,g^a)$ for unknown $a$, it is hard to find $(h,h^a)$ without \emph{also} finding a $b$ such that $h=g^b$. Formalizing KEA classically is a bit subtle, as an adversary can always construct $(h,h^a)$ using some exponent $b$, but then simply forget it. Or it could encrypt $(g,g^a)$ under an FHE scheme, choose an encryption of a random $b$, and then homomorphically compute $(h,h^a)$ where $h=g^b$. Then it decrypts the result to get $(h,h^a)$, without ever explicitly writing down $b$.

In both the cases above, it is nevertheless trivial to figure out what $b$ is by looking at the algorithm. In the first example, the execution transcript will contain $b$ before it is erased. In the FHE example, the secret key is needed to decrypt $(h,h^a)$. Using the same secret key also allows decrypting $b$.

KEA says that it must be possible to find $b$ \emph{always}, for any algorithm. This is formalized by means of an extractor: given any algorithm $A$, there exists an extractor $E$ that is given the same inputs as $A$---importantly, including any random coins of $A$---that can find the corresponding $b$ whenever $A$ outputs an $(h,h^a)$.

\medskip

We now explore what such a knowledge assumption looks like quantumly. Of course, quantumly the KEA assumption is not interesting since there is no hardness over groups. But we can nevertheless try to see how we might formalize it. The immediate problem is that quantum algorithms can be probabilistic without having explicit random coins as input. Instead, the quantum algorithm could create a superposition and measure it. This measurement is unpredictable, and un-repeatable. Trying to run the adversary from the same initial inputs will give different answers every time.

This says that a knowledge assumption in the quantum setting must be conditioned on the \emph{output} the adversary produces, rather than the input. Note that in the classical setting, we cared about the output as well (since different outputs would have different $b$ values), but we could connect the inputs to outputs by making the adversary deterministic by considering the random coins as input. In the quantum setting, this is no longer the case.

But if we are only looking at the final state of the algorithm, we run into a different problem. Namely, we are back to the setting where the adversary could have known an explicit path at one point, and then discarded the information. This is potentially even easier quantumly than classically: a quantum algorithm could measure the path in the Fourier basis, which would have the effect of erasing the path.

Our solution is natural: we consider only adversaries $\adv$ that are unitary. To get the final output, we must measure the registers containing the output. However, these are the \emph{only} measurements performed and there is no measurement on the adversaries internal state. The extractor $E$ is then given the final state of $\adv$ (as well as the output). $E$ then must be able to find $b$ given this state. Note that the restriction to unitary $\adv$ is without loss of generality, as any $\adv$ can be made reversible by Stinespring dilation (basically, instead of measuring, we XOR the registers into some newly created registers).

The classical analog is to restrict to \emph{reversible} classical adversaries, and giving $E$ the \emph{final} internal state of $\adv$ in addition to $\adv$'s output. Again, restriction to reversible $\adv$ is without loss of generality. This notion is actually equivalent to the usual classical KEA: given the output and internal state of $\adv$, reverse $\adv$ to find the input, including the random coins. Then apply the traditional KEA extractor using these random coins.

We now adapt this idea to the path-finding setting, giving the following notion of ``Knowledge of Path''. We will define two variants, based on whether the invariant is invertible.

\begin{assumption}\label{assump:kop1}Let $\adv$ be a quantum polynomial time adversary $\adv$ that is unitary (in the above sense where there are no measurements except the output registers). Let $E$ be a quantum polynomial time extractor that is given $\adv$'s final output as well as its final state. Let $(x,z)$ be the $\adv$'s output, and $p\in[r]^*$ be the output of $E$. 
	
Let $B$ be the event that (1) $I(x)=I(z)$, but (2) $p$ is not a path between $x$ and $z$. In other words, $B$ is the event that $\adv$ outputs two elements with the same invariant but $E$ fails to find a path between them.

The \emph{knowledge of path assumption} is that, for any quantum polynomial time unitary $\adv$, there exists a quantum polynomial time $E$ such that $\Pr[B]$ is negligible.
\end{assumption}

\begin{remark}Note that Assumption~\ref{assump:kop1} implies that it is infeasible to find $x,z$ such that $x$ and $z$ are in \emph{different} orbits, but $I(x)=I(z)$. This is because in such case, there does not exist a path from $x$ to $z$ and therefore $E$ must fail.\end{remark}

\paragraph{Invariant Inversion.} Sometimes, the invariant $I$ may be invertible. This is not required (or forbidden) for constructing quantum money, but it makes it likely that Assumption~\ref{assump:kop1} is false, since inverting $I(x)$ would give an element $z$ that most likely has no known path to $x$ (and a path may not even exist). Therefore, we will need to explicitly model such an invertible invariant, and modify our assumptions appropriately. So we introduce a classical \emph{randomized} algorithm $I^{-1}:Y\rightarrow X$, with the guarantee that $\Pr[I(I^{-1}(y))=y]=1$ for all $y$. We will typically consider the random coins of $I^{-1}$ as an explicit input, writing $I^{-1}(y;t)$.

Since our adversary can find multiple elements with the same invariant by inverting, and presumably elements obtained by inverting have no known path, we need to model this in our extractor. We therefore allow the extractor $E$ to do one of two things:
\begin{itemize}
	\item It can find a path from $x$ to $z$, or
	\item It can find random coins $t$, a path $p$, with the requirement that $p$ connects either $x$ or $z$ to $I^{-1}(y;t)$, where $y=I(x)=I(z)$.
\end{itemize}
In other words, the assumptions requires that the only way to find $x,z$ with the same invariant is to either know a path between them, or at least one of $x,z$ was the result of using the inversion algorithm on $y$ and then following some path. We now give the assumption.

\begin{assumption}\label{assump:kop2}Let $\adv$ be a quantum polynomial time adversary $\adv$ that is unitary (in the above sense where there are no measurements except the output registers). Let $E$ be a quantum polynomial time extractor that is given $\adv$'s final output as well as its final state. Let $(x,z)$ be the $\adv$'s output, and $p\in[r]^*,t$ be the output of $E$. Let $B$ be the event that (1) $I(x)=I(z)$, but (2) $p$ is not a path between any two of $\{x,z,I^{-1}(y;t)\}$.
	
The \emph{knowledge of path assumption for invertible invariants} is that, for any quantum polynomial time unitary $\adv$, there exists a quantum polynomial time $E$ such that $\Pr[B]$ is negligible.
\end{assumption}

We will also require that $I^{-1}$ is hard to invert. Namely, that, given $x$, it should be infeasible to come up with coins $t$ such that $I^{-1}(I(x);t)=x$. This is required for our updated knowledge of path assumption to be meaningful.

\begin{assumption}\label{assump:invertinvert} Consider an adversary $\adv$ playing the following game: 
	\begin{itemize}
		\item The adversary outputs an $x\in X$. Let $y=I(x)$.
		\item The challenger then computes a random $z\in O_x$.
		\item The adversary wins if it can output $t$ and a path $p$ between $I^{-1}(y;t)$ and $z$.
	\end{itemize}
	The \emph{Inversion Inverting assumption} is that, for all quantum polynomial-time adversaries $\adv$, the probability $\adv$ wins in the above game is negligible.
\end{assumption}
As with Path Finding Assumption (Assumption~\ref{assump:pathfinding}), the challenger in the Inversion Inverting assumption can be made efficient by choosing $z$ as a random walk starting from $x$. Note that in the experiment, $x$ is only used to specify an orbit $O_x$; the adversary's goal only depends on $z$.

\subsection{Security Proof for Invariant Money}
\label{sec:security_from_kop_proof}
\begin{theorem}Assuming the Path-Finding assumption (Assumption~\ref{assump:pathfinding}) and the Knowledge of Path Assumption (Assumption~\ref{assump:kop1}), the scheme above is secure quantum lightning. If the invariant is invertible, then assuming the Path-Finding assumption (Assumption~\ref{assump:pathfinding}), the Knowledge of Path Assumption for Invertible Invariants (Assumption~\ref{assump:kop2}), and the Inversion Inverting assumption (Assumption~\ref{assump:invertinvert}), the scheme above is secure quantum lightning.
\end{theorem}
\begin{proof}We prove the case for invertible invariants, the case of non-invertible invariants been very similar and somewhat simpler.
	
Toward contradiction, let $\adv$ be a quantum lightning adversary with non-negligible advantage $\epsilon$. By running $\adv$ for up to $\lambda/\epsilon$ times, stopping at the first success, we can gaurantee that $\adv$ wins with advantage negligibly close to 1. For simplicity in the following proof, we will assume the success probability is actually 1, incurring only a negligible error. 

We then assume without loss of generality that $\adv$ is unitary, so that the output is a pure state $\sum_{x,z,s}\alpha_{x,z,s}|x,z,s\rangle$, where the first two registers are the supposed quantum money states, and the last register is auxiliary state left over by running $\adv$.

Since $\adv$ passes verification with probability 1, we can instead write the output of $\adv$ as: 
\begin{align*}
|\psi\rangle&:=\sum_{s,O_1,O_2:I(O_1)=I(O_2)}\beta_{O_1,O_2,s}|O_1\rangle|O_2\rangle|s\rangle\\
&=\sum_{s,O_1,O_2:I(O_1)=I(O_2)}\frac{\beta_{O_1,O_2,s}}{\sqrt{|O_1||O_2|}}\sum_{x\in O_1,z\in O_2}|x,z,s\rangle\\
&=\sum_{s,x,z:I(x)=I(z)}\frac{\beta_{O_x,O_z,s}}{\sqrt{|O_x||O_z|}}|x,z,s\rangle
\end{align*}
where $O_1,O_2$ range over orbits with the same invariant.

Let $E$ be the extractor guaranteed by applying Assumption~\ref{assump:kop2} to $\adv$. Now consider measuring the registers containing $x,z$, leaving the auxiliary state as
\[|\psi_{O_x,O_z}\rangle\propto \sum_s \beta_{O_x,O_z,s}|s\rangle\]
Then since $I(x)=I(y)$, we have that  $E(x,z,|\psi_{O_x,O_z}\rangle)$ outputs $p,t$ such that, with overwhelming probability over $x,z,p,t$, $p$ is a path between two of $\{x,z,I^{-1}(I(x);t)\}$. $B_1$ as the event $p$ connects $x,z$, $B_2$ as the event $p$ connects $x,I^{-1}(I(x);t)$, and $B_3$ as the event $p$ connects $z,I^{-1}(I(x);t)$. Let $q_1,q_2,q_3$ be the probabilities of the events $B_1,B_2,B_3$. Then $q_1+q_2+q_3\geq 1-\negl$.

Now we notice that for any $x'\in O_x$ and and $z'\in O_z$, the probability of obtaining $x',z'$ is identical to the probability of obtaining $x,z$, and the state $|\psi_{O_x,O_z}\rangle=|\psi_{O_{x'},O_{z'}}\rangle=|\psi_{O_1,O_2}\rangle$. In particular, we have that $E(x',z',|\psi_{O_x,O_z}\rangle)$ outputs a path $p$ between $\{x',z',I^{-1}(I(x');t)\}$ with non-negligible probability, for uniform $x'\in O_x,y'\in O(z)$. Moreover, the quantities $q_1,q_2,q_3$ are unchanged by using $x',z'$. We now use this to construct adversaries for path-finding (Assumption~\ref{assump:pathfinding}) and inversion inverting (Assumption~\ref{assump:invertinvert}).

Let $\adv_P$ be the following path-finding adversary:
\begin{itemize}
	\item Run $\adv$ and measure $x,z$, obtaining the state $|\psi_{O_x,O_z}\rangle$.
	\item Send $z$ to the challenger, obtaining a random $x'\in O_z$ in response.
	\item Run $E(x',z,|\psi_{O_x,O_z}\rangle)$ to get $p,t$. Output $p$.
\end{itemize}
Note that $\adv_P$ simulates exactly $\adv,E$ in the case where $x,z$ are in the same orbit, which in particular is implied by event $B_1$. Therefore, $\adv_P$ outputs a path between $x',z$ with probability at least $q_1$. By the assumed hardness of path-finding (Assumption~\ref{assump:pathfinding}), $q_1$ must be negligible.

Let $\adv_I$ be the following inversion-inverting adversary:
\begin{itemize}
	\item Run $\adv$ and measure $x,z$, obtaining the state $|\psi_{O_x,O_z}\rangle$.
	\item Send $x$ to the challenger, obtaining a random $x'\in O_x$ in response.
	\item Run $E(x',z,|\psi_{O_x,O_z}\rangle)$ to get $p,t$. Output $p,t$.
\end{itemize}
Note that $\adv_P$ simulates exactly $\adv,E$, and so $p$ connects $x'$ to $I^{-1}(I(x);t)$ with probability $q_2$. By the assumed hardness of inversion inverting (Assumption~\ref{assump:invertinvert}), $q_2$ must be negligible. By an identical argument exchanging the roles of $x$ and $z$, we must also have $q_3$ is negligible. This contradicts $q_1+q_2+q_3\geq 1-\negl$. This completes the security proof.\end{proof}

\fi

\bibliographystyle{alpha}
\bibliography{abbrev0,crypto,bib}

\appendix

\section{Additional Preliminaries}
\label{sec:appendix_prelim}

\ifllncs
\subsection{Lattice Basics}

We define a the \emph{Gram-Schmidt basis} and the \emph{Gram-Schmidt norm} based on the definitions of~\cite{STOC:GenPeiVai08}.

\begin{definition} \label{def:GS} \textbf{Gram-Schmidt Basis.} For any (ordered) set $\Sm = \left\{ \sv_{1}, ... , \sv_{n} \right\} \subset \R^{n}$ of linearly independent vectors, let $\tilde{S} = \left\{ \tilde{\sv}_{1}, ... , \tilde{\sv}_{n} \right\}$ denote its Gram-Schmidt orthogonalization, defined iteratively in the following way: $\tilde{\sv}_{1} = \sv_{1}$, and for each $i \in \left[2, n \right]$, $\tilde{\sv}_{i}$ is the component of $\sv_{i}$ orthogonal to $span \left(\sv_{1}, ... , \sv_{i - 1} \right)$.
\end{definition}

We next define the GapSVP problem, which is the worst-case lattice problem upon which the hardness of LWE is based.

\begin{definition}~\label{def:GapSVP}
Let $n$ be an integer and $\gamma = \gamma \left(n \right) \geq q$ a real number.  The $\left(n, \gamma \right)$-GapSVP problem is the problem of deciding, given a basis $\Bm$ of an $n$-dimensional lattice $\Lambda$ and a number $d$, whether or not $\lambda_{1} \left( \Lambda \right) \leq d$ or $\lambda_{1} \left( \Lambda \right) > \gamma d$.
\end{definition}

We emphasize that GapSVP is a ``promise problem'' and that a (successful) adversary can output whatever it wants when $d < \lambda_{1} \left( \Lambda\ \right) \leq \gamma d$.

\paragraph{Worst-Case Lattice Problems and LWE}

We next define the GapSVP problem, which is the worst-case lattice problem upon which the hardness of LWE is based.

\begin{definition}~\label{def:GapSVP}
Let $n$ be an integer and $\gamma = \gamma \left(n \right) \geq q$ a real number.  The $\left(n, \gamma \right)$-GapSVP problem is the problem of deciding, given a basis $\Bm$ of an $n$-dimensional lattice $\Lambda$ and a number $d$, whether or not $\lambda_{1} \left( \Lambda \right) \leq d$ or $\lambda_{1} \left( \Lambda \right) > \gamma d$.
\end{definition}

We emphasize that GapSVP is a ``promise problem'' and that a (successful) adversary can output whatever it wants when $d < \lambda_{1} \left( \Lambda\ \right) \leq \gamma d$.

In his seminal work~\cite{STOC:Regev05}, Regev showed that LWE was hard worst-case lattice problems for uniformly random choices of $\mathcal{D}_{\Am}$ and $\mathcal{D}_{\Rm}$ and when $\mathcal{\mathbb{\Psi}}$ was defined to be choosing each coordinate as a (small) discrete Gaussian.  To capture this, we cite a theorem from~\cite{STOC:BLPRS13} which itself is derived from Theorem 3.1 of~\cite{STOC:Regev05} and Theorem 3.1 of~\cite{STOC:Peikert09}.

\begin{theorem}~\label{thm:LWEHardness} \textbf{(Theorem 2.16,~\cite{STOC:BLPRS13})} 
Let $n$, $m$, and $q$ be positive integers and let $\mathcal{D}_{\mathbf{\Psi}_{\sigma}}^{m}$ be a distribution over $\Z^{m}$ where each entry is selected according to a discrete Gaussian distribution with noise rate parameter $\sigma > 2 \sqrt{n}$.  Then there exists a quantum reduction from worst-case $\left(n, \tilde{O} \left( n q / \sigma  \right) \right)-GapSVP$ to $\left(n, m, q, \mathcal{U} \left( \Z_{q}^{n} \right), \mathcal{U} \left( Z_{q}^{n} \right), \mathcal{D}_{\mathbf{\Psi}_{\sigma}}^{m} \right)$-LWE.  In addition, if $q \geq 2^{n/2}$, then there is also a classical reduction between those problems. 
\end{theorem}

We note that our presentation of this theorem differs quite a bit from the presentation in~\cite{STOC:BLPRS13} because they present LWE as a problem over the cycle (the additive group of reals modulo 1) for ease of exposition about the noise parameters, but it makes more sense when working in quantum money setting to present things over the integers.

\subsection{The $k$-LWE Problem}
With the LWE definition in place, we are ready to move to the actual $k$-LWE problem.  The $k$-LWE problem was first formally defined in~\cite{C:LPSS14} and used to build traitor-tracing schemes.  It extends the $k$-SIS assumption~\cite{PKC:BonFre11} which was used to build linearly homomorphic signatures.  Our definition below is essentially a parameterized version of the one in~\cite{C:LPSS14}.

\begin{definition}~\label{def:kLWE} \textbf{$k$-LWE Problem:}
Let $k$, $n$, $m$, and $p$ be integers, let $\mathcal{D}_{\Rm}$ be a distribution over $\Z_{q}^{n}$, and let $\mathcal{D}_{\mathbb{\Psi}}$ and $\mathcal{D}_{\Sm}$ be distributions over $\Z_{q}^{m}$.  Let $\Sm \in \Z_{q}^{k \times m}$ be a matrix where each \emph{row} is selected from $\mathcal{D}_{\Sm}$.  

Let $\Am \in \Z_{q}^{m \times n}$ be a matrix sampled uniformly from the set of matrices in $\Z_{q}^{m \times n}$ such that $\Sm \cdot \Am = 0 \mod q$, let $\rv \in \Z_{q}^{n}$ be a vector sampled from $\mathcal{D}_{\rv}$, and let $\ev \in \Z_{q}^{m}$ be a vector sampled from $\mathcal{D}_{\mathbb{\Psi}}$.  Let $\Cm \in \Z_{q}^{m \times \left(m - k \right)}$ be a basis for the set of vectors $\vv \in \Z_{q}^{m}$ such that $\Sm \cdot \vv = 0$, and let $\rv' \in \Z_{q}^{m - k}$ be a uniformly random vector.

The $\left(k, n, m, q, \mathcal{D}_{\Sm}, \mathcal{D}_{\rv}, \mathcal{D}_{\mathbb{\Psi}} \right)$-$k$-LWE problem is defined to be distinguishing between the following distributions: 
\[  \left( \Sm, \Am, \Cm, \Am \cdot \rv+ \ev \right) \text{ and } \left(\Sm, \Am,  \Cm, \Cm \cdot \rv'+ \ev \right) \] 
\end{definition}

We note that $k$-LWE is traditionally defined in a slightly different way:  usually the matrix $\Am$ is sampled before (or jointly with) $\Sm$ rather than after it.  We sample $\Sm$ first in our definition because we will need to handle very unusual (at least for cryptographic applications) distributions $\mathcal{D}_{\Sm}$.

\paragraph{Previous Results on Hardness of $k$-LWE}
The authors of~\cite{C:LPSS14} prove the following theorem about $k$-LWE, which does not apply to our setting as we will specify later.

\begin{theorem}~\label{thm:oldkLWE} \textbf{(Theorem 3,~\cite{C:LPSS14})}:
Let $n$, $m$, and $p$ be integers, and $\sigma$ and $\sigma'$ be rational numbers such that $\sigma \geq \Omega \left(n \right)$, $\sigma' \geq \Omega \left(n^{3} \sigma^{2} / \log n \right)$, $p \geq \Omega \left( \sigma' \sqrt{\log m} \right)$ is prime, and $m \geq \Omega \left( n \log p \right)$.  For instance, we could set $\sigma = \Theta \left(n \right)$, $\sigma' = \Theta \left(n^{5} / \log n \right)$, $p = \Theta \left( n^{5} \right)$, and $m = \Theta \left( n \log n \right)$.  Let $\mathcal{D}_{alpha}^{d}$ and $\mathcal{D}_{\alpha'}^{d}$ denote discrete Gaussian distributions over $\Z^{d}$ with parameter $\alpha$ and $\alpha'$, respectively.

Define the matrix 
\[
\tilde{\Sm} = \left[ \begin{array}{cc}
\sigma \cdot \mathbf{I}_{m + n} & 0 \\
0 & \sigma' \cdot \mathbf{I}_{n}
\end{array} \right]
\]
and let $\mathcal{D}_{\tilde{\Sm}}$ denote the discrete Gaussian distibution over $\Z^{m + 2n}$ with skew $\tilde{\Sm}$.

Then, for some constant $c$ such that $k = n / \left(c \log n \right)$, there exists a PPT reduction from $\left(n, m+1, p, \mathcal{U} \left( \Z_{q}^{n} \right), \mathcal{U} \left( \Z_{q}^{n} \right), \mathcal{D}_{\alpha} \right)$-LWE to $\left(k, n, m + 2n, p, \mathcal{D}_{\tilde{\Sm}}, \mathcal{U} \left( \Z_{q}^{n} \right), \mathcal{D}_{\alpha'} \right)$-k-LWE and advantage $\epsilon' = \Omega \left( \left( \epsilon - 2^{\Omega \left(n / \log n \right)} \right)^{3} - O \left(2^{-n} \right) \right)$.
\end{theorem}

Informally speaking, this theorem shows that $k$-LWE is hard for a fixed lattice dimension $n$ with only a minor (polynomial) blowup in parameters for a short vector distribution of a certain type of (nicely) skewed Gaussian.  Unfortunately, this distribution of short vectors is not related to the distribution implied by the~\cite{KLS22} scheme, so we cannot directly apply this theorem.  However, we will prove a different result than can directly be applied to the~\cite{KLS22} scheme in the body of the paper.

\else

\fi

\section{On the Hardness of $k$-LWE}
\label{sec:klwe}

In this section we prove a number of hardness results on $k$-LWE, showing that the $k$-LWE instance implied by the~\cite{KLS22} scheme is hard, assuming standard lattice assumptions.  In addition, we provide additional evidence through proofs that more general instances of $k$-LWE are likely to be hard, which seemingly indicates that it would be difficult to ``tweak'' the~\cite{KLS22} construction by altering the distributions to gain security.

We start by presenting some useful LWE lemmas from previous work.

\subsection{Helpful LWE Lemmas}
In this section, we add in some useful lemmas that allow us to change the distribution of the key in the LWE problem ($\mathcal{D}_{\rv}$) without affecting the hardness of the underlying problem too much.  Looking ahead, we will need to use various versions of the powerful modulus switching lemma from~\cite{STOC:BLPRS13}, which requires LWE instances with ``low-norm'' keys.

We start with a simple folklore lemma, which informally states that LWE with a uniformly sampled random key is at least as hard as LWE with any other secret key distribution.

\begin{lemma}~\label{lem:unifkey}
Let $n$, $m$, and $q$ be integers, let $\mathcal{D}_{\Am}$ and $\mathcal{D}_{\rv}$ be distributions over $\Z_{q}^{n}$, and let $\mathcal{D}_{\mathbb{\Psi}}$ be a distribution over $\Z_{q}^{m}$.  Any adversary that can solve the $\left(n, m, q, \mathcal{D}_{\Am}, \mathcal{U} \left( \Z_{q}^{n} \right), \mathcal{D}_{\mathbb{\Psi}} \right)$-LWE problem with advantage $\epsilon$ can be used to solve the $\left(n, m, q, \mathcal{D}_{\Am}, \mathcal{D}_{\rv}, \mathcal{D}_{\mathbb{\Psi}} \right)$-LWE problem with advantage $\epsilon$.
\end{lemma}

\begin{proof}
We give an abbreviated proof because this result is simple and well-known.  Given an LWE challenge tuple $\left(\Am, \tv \right)$ where $\tv = \Am \cdot \rv + \ev$ and $\rv \leftarrow \mathcal{D}_{\rv}$ or $\tv$ is random, sample $\rv' \in \Z_{q}^{n}$ uniformly at random and add $\Am \cdot \rv'$ to $\tv$.  This gives the correct LWE challenge distribution if $\tv = \Am \cdot \rv + \ev$ and is still uniformly random if $\tv$ was uniformly random.
\end{proof}

We will also use a lemma from~\cite{STOC:BLPRS13} that says, informally speaking, that LWE with certain parameters where the key is drawn from the noise distribution is at least as hard as when the key is uniform (modulo some small parameter losses).  We state this below.

\begin{lemma}~\label{lem:noisekey} \textbf{(Lemma 2.12,~\cite{STOC:BLPRS13})}
Let $n$, $m$, and $q$ be positive integers with $q \geq 25$.  Let $m' = m - \left(16 n + 4 \ln \ln q \right)$.  Consider some parameter $s \geq \sqrt{\ln \left( 2n \left( 1 + 1/\epsilon \right) / \pi \right)}$.  Let $\epsilon < \frac{1}{2}$ and $\sigma, \sigma' > 0$ be real numbers such that $\sigma'  \geq \sqrt{ \sigma^{2} + s^{2}}$.  Finally, let $\mathcal{D}_{\mathbf{\Psi}_{\sigma}}$ be a discrete Gaussian distribution with parameter $\sigma$.

Any adversary that can solve the $\left(n, m', q, \mathcal{U} \left( \Z_{q}^{n} \right), \mathcal{D}_{\mathbf{\Psi}_{\sigma'}}, \mathcal{D}_{\mathbf{\Psi}_{\sigma}} \right)$-LWE problem with advantage $\epsilon'$ can be used to solve the $\left(n, m, q, \mathcal{U} \left( \Z_{q}^{n} \right), \mathcal{U} \left( \Z_{q}^{n} \right), \mathcal{D}_{\mathbf{\Psi}_{\sigma}} \right)$-LWE problem with advantage $\left(\epsilon' - 8 \epsilon \right) / 4$.  In particular, assuming $\sigma > s$, we can take $s = \sigma$ and set $\sigma' = \sqrt{2} \sigma$.
\end{lemma}

Note that the two versions of LWE in this reduction have a (slightly) different number of samples ($m$ and $m'$), but the dimension of the LWE problem ($n$) is the same.

\subsection{Modulus Switching}
In their seminal work~\cite{STOC:BLPRS13}, Brakerski \emph{et al.} use a technique called \emph{modulus switching} to improve known reductions from the GapSVP problem to the LWE problem.  Informally speaking, this modulus switching technique allows those authors to show that LWE in ``small modulus'' and ``high dimension'' is roughly equivalent in hardness to LWE in ``big modulus'' and ``low dimension''.  We state some of their results here since we would eventually like to use the fact that one-dimensional LWE with exponential modulus is hard, which is known from results in~\cite{STOC:BLPRS13}.

Rather than present a single instance of their main theorem, we go through two of their corollaries to make it easier to follow.  We emphasize that our presentation looks very different from theirs because they consider LWE over the unit cycle $\mathbb{T}$ and we work over the integers.

\begin{lemma}~\label{lem:modreduction} \textbf{(Corollary 3.2,~\cite{STOC:BLPRS13})}
Let $n$, $m$, $q$, and $q'$ be positive integers with $q' > q$, and consider some $\left(B, \delta \right)$-bounded distribution $\tilde{\mathcal{D}}$ over $\Z^{n}$.  Let $\epsilon \in \left(0, \frac{1}{2} \right)$ be a parameter and let $\sigma, \sigma > 0$ be real numbers.  Finally, let 
\[
\sigma' \geq \sqrt{ \left( \sigma \frac{q'}{q} \right)^{2} + \left(4 / \pi \right) \ln \left( 2n \left(1 + 1/ \epsilon \right) \right) \cdot B^{2} }
\]
Then there is an efficient reduction from $\left(n, m, q, \mathcal{U} \left( \Z_{q}^{n} \right), \tilde{\mathcal{D}}, \mathcal{D}_{\mathbf{\Psi}_{\sigma}} \right)$-LWE to $\left(n, m, q', \mathcal{U} \left( \Z_{q}^{n} \right), \tilde{\mathcal{D}}, \mathcal{D}_{\mathbf{\Psi}_{\sigma'}} \right)$-LWE that reduces the advantage by at most $\sigma + 14 \epsilon m$. 
\end{lemma}

Informally speaking, this lemma gives us a way to increase the modulus of an LWE instance while keeping the ``gap between the noise level and the modulus'' almost the same.  This is a rather counterintuitive result, and even moreso when you consider the fact that it works for essentially arbitrary distributions on LWE secrets.  Importantly, $q$ and $q'$ can be (essentially) arbitrary as long as $q' > q$, so we can use this lemma to help us prove LWE hardness for ``arbitrary'' choices of $q'$ in conjunction with other lemmas that require certain properties of $q'$.

We next present another modulus switching lemma that lets us go from ``high modulus, low dimension'' to ``normal modulus, normal dimension'' instances of LWE.  Once again, note that this reduction approximately preserves the gap between the noise rate and the modulus.

\begin{lemma}~\label{lem:modexp} \textbf{(Corollary 3.4,~\cite{STOC:BLPRS13})}
Consider any positive integers $n$, $m$, $q$, and $k$ such that $k$ divides $n$, real numbers $\sigma, \sigma' > 0$, a parameter $\epsilon \in \left(0, \frac{1}{2} \right)$, and some $\left(B, \delta \right)$-bounded distribution $\tilde{\mathcal{D}}$.  In addition, let
\[
\sigma' \geq \sqrt{ \left( \sigma q^{k - 1} \right)^{2} + \left(4 / \pi \right) \ln \left( 2n \left(1 + 1/ \epsilon \right) \right) \cdot \left(B q^{k - 1}\right)^{2} }
\]
and define $\Gm = \mathbf{I}_{n/k} \otimes \left(1, q, q^{2}, ... , q^{k - 1} \right)^{T}$.

Then there is an efficient reduction from $\left(n, m, q, \mathcal{U} \left( \Z_{q}^{n} \right), \tilde{\mathcal{D}}, \mathcal{D}_{\mathbf{\Psi}_{\sigma}} \right)$-LWE to $\left(n/k, m, q^{k}, \mathcal{U} \left( \Z_{q}^{n} \right), \mathbf{G} \tilde{\mathcal{D}}, \mathcal{D}_{\mathbf{\Psi}_{\sigma'}} \right)$-LWE that reduces the advantage by at most $\delta + 14 \epsilon m$.
\end{lemma}

Note that setting $k = n$ gives us hardness for single-dimension LWE (i.e. $n = 1$).  We will use this exact setting later in our proofs.

\subsection{$k$-LWE for a Constant Number of Vectors Is Hard}
In this section we prove that $k$-LWE with the appropriate parameters is as hard as regular LWE for \emph{any} distribution on the short vectors, up to a superpolynomial loss in the noise, \emph{assuming that the number of short vectors $k$ is constant}.  We state this below.  Our proof is inspired by the $k$-SIS proof of~\cite{PKC:BonFre11}.

\begin{lemma}~\label{lem:constkLWE}
Let $k$, $n$, $m$, and $q$ be positive integers, and let $\mathcal{D}_{\Rm}$ be a distribution over $\Z_{q}^{n}$.  Let $\mathcal{D}_{\mathbb{\Psi}_{\sigma}}$ be a discrete Gaussian distribution over $\Z_{q}^{m}$ with noise parameter (width) $\sigma$.  Furthermore, let $\mathcal{D}_{\Sm}$ be a distributions over $\Z_{q}^{m}$ with the additional requirement that $\mathcal{D}_{\Sm}$ is $B$-bounded.  Let $\Sm \in \Z_{q}^{k \times m}$ be a matrix where each \emph{row} is selected from $\mathcal{D}_{\Sm}$.  Let $f \left(n \right)$ be a function that is superpolynomial in $n$.

Any adversary that can solve the $\left(k, n, m + k, q, \mathcal{D}_{\Sm}, \mathcal{D}_{\rv}, \mathcal{D}_{\mathbb{\Psi}_{\sigma f \left(n \right) }} \right)$-$k$-LWE problem with advantage $\epsilon$ can be used to solve the $\left(n, m, q, \mathcal{U} \left( \Z_{q}^{n} \right), \mathcal{D}_{\rv}, \mathcal{D}_{\mathbb{\Psi}_{\sigma}} \right)$-LWE problem with advantage $\epsilon - \negl \left(n \right)$.
\end{lemma}

\begin{proof}
Suppose we are given an LWE challenge tuple in the form of $\left(\Am, \tv \right)$ where $\tv = \Am \cdot \rv + \ev$ (the ``real'' case) or $\tv = \rv$ for some uniformly sampled $\rv$.  We will show that we can use this LWE challenge to build a $k$-LWE challenge of the appropriate distribution in a way that succeeds with all but negligible probability.  Thus, given an appropriate $k$-LWE adversary, we can simply feed it our challenge tuple and then mimic that response in the LWE challenge game, winning with probability negligibly close to $\epsilon$.  Our reduction proceeds as follows (and borrows heavily from the techniques of~\cite{PKC:BonFre11}): 

First, suppose we sample $\Sm \leftarrow \mathcal{D}_{\Sm} \in \Z_{q}^{k \times \left(m + k \right)}$ and let $\Sm_{i, j}$ denote the $\left(i, j \right)th$ entry of $\Sm$.  Suppose we pick $k$ columns of $\Sm$ such that the $k \times k$ submatrix formed by these columns is full-rank; and note that we can always do this unless $\Sm$ has rank less than $k$ (in which case, we can just reduce $\Sm$ by one row and repeat the process).  Call this matrix $\Tm$.  Let $\tilde{\Tm}$ denote the \emph{adjugate} matrix of $\Tm$; in other words, $\tilde{\Tm} \cdot \Tm = det \left( \Tm \right) \cdot \mathbf{I}_{k}$ and $\tilde{\Tm} \in \Z^{k \times k}$.  Let $\tilde{\Sm} \in \Z_{k \times m} = \tilde{\Tm} \cdot \Sm$, and note that it has the following structure:
\[
\tilde{\Sm} = \left[ \begin{array}{cccccccc}
det \left( \Tm \right) & 0 & 0 &... & 0 &\tilde{\Sm}_{1, k+1} & ... & \tilde{\Sm}_{1, k + m} \\
0 & det \left( \Tm \right) & 0 & ... &  0 & \tilde{\Sm}_{2, k + 1} & ... & \tilde{\Sm}_{2, k + m} \\
0 & 0 &  det \left( \Tm \right) & ... &  0 & \tilde{\Sm}_{3, k + 1} & ... & \tilde{\Sm}_{3, k + m} \\
... & ... & ... & ... & ... & ... & ... & ... \\
0 & 0 & 0 & ... &  det \left( \Tm \right) & \tilde{\Sm}_{k, k + 1} & ... & \tilde{\Sm}_{k, k + m} 
\end{array} \right]
\]
Suppose we let $\Um \in \Z^{\left(m + k \right) \times m}$ be defined in the following way:  for each entry $\Um_{i, j}$ where $i \leq k$, we set $\Um_{i,j} = \tilde{\Sm}_{i, j + k}$.  For each entry $\Um_{i, j}$ for $i > k$, we set $\Um_{i, j } = - det \left( \Tm \right)$ if $i = j$ and $0$ otherwise.  Note that, pictorially, this gives us the following structure:
\[
\Um = \left[ \begin{array}{ccccc}
\tilde{\Sm}_{1, k+1} & \tilde{\Sm}_{1, k+2} & \tilde{\Sm}_{1, k+3} & ... & \tilde{\Sm}_{1, k + m} \\
\tilde{\Sm}_{2, k + 1} & \tilde{\Sm}_{2, k + 2} & \tilde{\Sm}_{2, k + 3} &   ... & \tilde{\Sm}_{2, k + m} \\
\tilde{\Sm}_{3, k + 1} & \tilde{\Sm}_{3, k + 2} & \tilde{\Sm}_{3, k + 3} & ... & \tilde{\Sm}_{3, k + m} \\
 ... & ... & ... & ... & ... \\
 \tilde{\Sm}_{k, k + 1} &  \tilde{\Sm}_{k, k + 2} &  \tilde{\Sm}_{k, k + 3} &  ... & \tilde{\Sm}_{k, k + m} \\
- det \left( \Tm \right) & 0 & 0 &... & 0 \\
0 & - det \left( \Tm \right) & 0 & ... &  0  \\
0 & 0 & - det \left( \Tm \right) & ... & 0 \\
... & ... & ... & ... & ...\\
0 & 0 & 0 & ... & - det \left( \Tm \right)  
\end{array} \right]
\]
We interrupt the reduction to make a couple of claims.

\medskip
\noindent \textbf{Claim:  $\Sm \cdot \Um = 0$}.  This follows from the fact that $\tilde{\Sm}$ is orthogonal to $\Um$ by definition and $\Sm = \tilde{T} \Sm$.  Since $\Tm$ is invertible (over both $\Z$ and $\Z_{q}$, $\Sm \cdot \Um = 0$ if and only if $\tilde{\Sm} \cdot \Um = 0$.

\medskip
\noindent \textbf{Claim:  no entry in $\Um$ has value larger than $\left(2B \right)^{k}$}.  First, note that the determinant of any $k \times k$ matrix with entries in $\left[-B, B \right]$ is at most $\left(2B \right)^{k}$, and the determinant of any $B$-bounded smaller matrices must be smaller than that.  Since the entries of the adjugate matrices are themselves determinats of submatrices of $\Tm$, they must also follow this bound. 

\medskip

Now we may continue with our reduction.  Suppose we sample $\ev' \leftarrow \mathcal{D}_{\mathbb{\Psi}_{\sigma f \left(n \right)}}^{m + k}$ and let $\Am' \in \Z_{q}^{m + k} \times n = \Um \cdot \Am$ and $\tv' = \Um \cdot \tv + \ev'$.  Additionally, suppose we sample some random matrix $\Wm \in \Z_{q}^{m \times m}$ and set $\Cm = \Um \cdot \Wm$.  We claim that the tuple $\left(\Sm, \Am', \Cm, \tv' \right)$ is an appropriately distributed $k$-LWE challenge.  We show this in a number of claims:  first, with some claims that apply to both the ``real'' and ``random'' cases, and then we argue these cases separately.

\medskip
\noindent \textbf{Claim:  $\Cm$ is a uniform basis of the kernel of $\Sm$ $\mod q$.}  This follows from the fact that $\Sm \cdot \Cm = 0$ over $\Z$ and the ranks of the matrices sum to $m + k$.

\medskip
\noindent \textbf{Claim:  $\Am' = \Um \cdot \Am \mod q$ is distributed uniformly at random subject to the constraint $\Sm \cdot \Am' = 0$}.  Since $\Am$ is uniformly random and $\Um$ is a basis for all vectors in the kernel of $\Sm$, we know that $\Am'$ is distributed appropriately $\mod q$.

\medskip

So, at this point we know that $\Sm$, $\Am'$, and $\Cm$ are distributed appropriately.  All that remains is to handle $\tv'$  We handle this separately for the ``real'' and ``random'' cases below.

\medskip
\noindent \textbf{The ``Real'' Case.}  Assume now that $\tv = \Am \cdot \rv + \ev$.  Then we have 
\[
\tv' = \Um \left( \Am \cdot \rv + \ev \right) + \ev' = \Am' \cdot \rv + \left(\Um \ev + \ev' \right).
\]

If $k$ is constant, then $\left(2B \right)^{k}$ is constant and $\Um \ev$ has no entries larger than $O \left( m^2 \sigma \right)$ with all but negligible probability since the probability that a discrete Gaussian with parameter $\sigma$ is larger than $m \sigma$ is negligible.  Since $f \left(n \right)$ grows faster than any polynomial, we know that $\Um \cdot \ev + \ev'$ is distributed statistically close to just sampling a discrete Gaussian with parameter $\sigma f \left(n \right)$ (i.e. how we sampled $\ev'$), so we know that $\tv'$ is sampled appropriately in this case.  

\medskip
\noindent \textbf{The ``Random'' Case.}  Now assume that $\tv'$ is distributed uniformly at random over $\Z_{q}^{m}$.  In this case, we know that $\Um \cdot \tv$ is distributed uniformly at random over the kernel of $\Sm$ $\mod q$ and thus can be written as $\Cm \cdot \rv'$ for a uniformly random $\rv'$ as desired.  Since $\tv' = \Um \cdot \tv + \ev'$, we therefore know that the output distribution is correct in the ``random'' case as well.

\medskip

Completing these two cases finishes the reduction and completes the proof.
\end{proof}

\subsection{Putting It All Together}
In order to show that any adversary that can solve the $k$-LWE instance implied by the~\cite{KLS22} scheme can solve worst-case lattice problems, we just need to put all of our lemmas together.  Below, we show a table containing all of our hybrid arguments with approximate factors:  we ignore constant factors in the dimensions and polynomial factors in the noise and modulus (when the modulus is exponentially large).  We assume $q > \left(q' \right)^{n}$ and that the noise distributions are Gaussians.  We note that the ``Noise'' entry for GapSVP is the approximation ratio defined by the problem, not the noise itself.
\[
\begin{array}{|c|c|c|c|c|c|c|}
\text{Assumption} & \text{Lattice Dim.} & \text{Samples} & \text{Modulus} & \text{Key Dist.} & \text{Noise} & \text{Proof Comment} \\ \hline
k-LWE & 1 & m & q & \text{Unif.} & \frac{q}{\eta} & \text{n/a} \\ \hline
LWE & 1 & m & q & \text{Unif.} & \frac{q}{\eta f \left(n \right)} & \text{Lemma~\ref{lem:constkLWE}} \\ \hline
LWE & 1 & m & q & \text{Noise} & \frac{q}{\eta f \left(n \right)} & \text{Lemma~\ref{lem:unifkey}} \\ \hline
LWE & 1 & m & \left(q' \right)^{n} & \text{Noise} & \frac{\left(q' \right)^{n}}{\eta f \left(n \right)} & \text{Lemma~\ref{lem:modreduction}} \\ \hline
LWE & n & m & q' & \text{Noise} & \frac{q'}{\eta f \left(n \right)} & \text{Lemma~\ref{lem:modexp}} \\ \hline 
LWE & n & m & q' & \text{Unif.} & \frac{q'}{\eta f \left(n \right)} & \text{Lemma~\ref{lem:noisekey}} \\ \hline
GapSVP & n & - & - &  - & \eta f \left(n \right) & \text{Theorem~\ref{thm:LWEHardness}} \\ \hline
\end{array}
\]

We can state this in a nice lemma.

\begin{lemma}
Let $n$, $m$, $q$, and $k$ be integers such that $k$ is a constant.   Let $f \left(n \right)$ be some function that grows superpolynomially in $n$.  Let $q$ be exponentially large and $m$ polynomially sized in some security parameter $n$, and let $q' = poly \left(n \right) f \left(n \right)$.  Let $\sigma$ be a discrete Gaussian parameter such that $\sigma \geq \frac{q}{\eta}$ and let $\mathcal{D}_{\Psi_{\sigma}}$ denote a discrete Gaussian with parameter $\sigma$.  Let $\mathcal{D}_{\Sm}$ be any distribution over $\Z^{m}$ that outputs vectors bounded by some polynomial in $n$.

An adversary that solves the $\left(k, 1, m, q, \mathcal{D}_{\Sm}, \mathcal{U} \left( \Z_{q} \right), \mathcal{D}_{\Psi_{\sigma}} \right)$-k-LWE problem with non-negligible advantage can be used to solve the $\left( \log_{q'} \left(q \right), \sigma f \left(n \right) \right)$-GapSVP problem with non-negligible advantage.
\end{lemma}

\begin{proof}
The proof follows from a simple hybrid argument outlined in the table above and by just applying the lemmas in sequence without much additional thought.  We explain each step below.  For the sake of clarity, we ignore small factors in the discussion below (and these don't matter anyway because we are only focused on an asymptotic result).
\begin{itemize}
\item We start by reducing to $k$-LWE in dimension $1$ from LWE in dimension $1$, with the only loss being a superpolynomial factor in the noise parameter.  This follows from the one lemma that did not follow from previous work, lemma~\label{lem:constkLWE}.
\item The modulus switching lemmas require a key with ``small'' entries, so we next use a reduction to LWE in dimension $1$ with uniform key from LWE in dimension $1$ with small key.  This is exactly what is stated in lemma~\ref{lem:unifkey}.
\item We now have a dimension $1$ LWE instance with a ``small'' key.  We use lemma~\ref{lem:modreduction} to reduce to this LWE instance from an LWE instance with a tailored modulus of the form $q'^{n}$ so that we can easily apply the core modulus switching lemma.
\item Given an LWE instance in dimension $1$ with modulus $q'^{n}$, we apply the core modulus-switching lemma--lemma~\ref{lem:modexp}--to reduce to this from an LWE instance in dimension $n$ with modulus $q'$.  This allows us to reach our desired ``standard'' dimension.  We note that the key in both of these instances is still ``small'' and of the noise distribution.
\item We next apply lemma~\ref{lem:noisekey} to reduce to LWE with a small key (drawn from the noise distribution, or a similar distribution) from LWE with a uniform key.  After this step we have essentially a ``standard'' LWE problem.
\item Finally, we apply the famous theorem of Brakerski \emph{et al.} to reduce to ``standard'' LWE from the Gap-SVP problem.
\end{itemize}

This completes the proof, which is essentially just a combination of a new result on $k$-LWE and an application of modulus switching.
\end{proof}

Note that $\log_{q'} \left(q \right)$ will be quite large (possibly even polynomial, depending on the choice of $q$), so the Gap-SVP problem reduced from here seems very likely to be hard.


\section{Instantiation Using Elliptic Curves}
\label{sec:ellipticcurve}
We next outline how our invariant construction might be instantiated with elliptic curve isogenies.  Our candidate construction(s) here are relatively high-level and need substantially more study on their security properties, particularly in light of some recent cryptanalysis on isogenies~\cite{cryptoeprint:2022/975,cryptoeprint:2022/1026}.

\paragraph{Elliptic Curve Isogenies.}  We briefly outline some properties of elliptic curves and isogenies that we will use in our candidate constructions.  For a full treatment, we highly recommend~\cite{de2017mathematics}.  \emph{Elliptic curves} are projective curves of genus one with a specified base point.  Elliptic curves over finite fields $k$ are often defined in \emph{Weierstrass form}, consisting of all points in the locus of the equation $y^{2} = x^{3} + ax + b$ and the point at infinity in $\mathbb{P}^{2} \left( \overline{k} \right)$, where we use $\overline{k}$ to denote the algebraic closure of $k$.  The \emph{$j$-invariant} of an elliptic curve $E$ in Weierstrass form is defined as 
\[
j \left( E \right) = 1728 \frac{4a^{3}}{4a^{3} + 27b^{2}} .
\]
Two elliptic curves $E_{1}$ and $E_{2}$ are isomorphic if and only if they have the same invariant.

An \emph{isogeny} $\psi: E \rightarrow E$ is a surjective group morphism between elliptic curves.  We note that isogenies only exist between elliptic curves with the same number of points.  We denote the number of points of a curve as $\# E \left( k \right)$, or sometimes just $\# E$.  Very roughly speaking, isogeny-based cryptography is built on the fact that for certain elliptic curves $E$ and related isogenies $\psi$, it is possible to efficiently compute $\psi \left(E \right)$, but given two elliptic curves $E_{1}$ and $E_{2}$, it is hard to find an isogeny $\psi$ (or set of isogenies applied in sequence) so that $\psi \left(E_{1} \right) = E_{2}$.

\paragraph{Background Ideas.}
Suppose we start with a hypothetical example.  Let $X$ be the set of elliptic curves (perhaps represented by their $j$-invariants) where the number of points on the curve has the form $\ell*q$ for a prime $q$ and some small $\ell$. This can be generalized; having several small factors could work as well, or also potentially restricting to super-singular curves. 

The invariant could be the number of points on the curve. The $\sigma_{y,i}$ are degree-$\ell$ isogenies, or if we are generalizing to multiple small factors, $\sigma_i$ will range over low-degree isogenies. Basically, we choose some arbitrary way of mapping $[r]$ to kernels of the $\ell$-torsion, and then $\sigma_i$ is applying the isogeny defined by that kernel. The inverse of an isogeny is just the dual isogeny, which can be efficiently computed.

The orbits $O$ then correspond to sets of elliptic curves that can be reached by sequences of degree-$\ell$ isogenies. Under some Ramanujan-Petersson conjecture~\cite{ChaLauGor08} (or alternatively GRH~\cite{JaoMilVen08}), action by small-degree isogenies gives an expander. This should give us $\lambda_2\leq 1-\delta$ as needed. We also see that isogenies preserve the number of points on the curve, so the invariant property is satisfied.

For security, the path finding is probably hard, as it is essentially the basis of isogeny cryptography~\cite{EPRINT:Couveignes06,EPRINT:RosSto06}. The plain knowledge of path is false: since we can compute elliptic curves with a given size, we have an invertible invariant. However, the knowledge of path for invertible invariants could possibly be true: perhaps the only way of obtaining two elliptic curves with the same number of points is to either:
\begin{itemize}
	\item Sample an elliptic curve $E_1$ and then follow a sequence of isogenies from it
	\item Sample an elliptic curve $E_1$, compute the number of points on the curve, and then construct a curve $E_2$ with that many points using the known algorithms for doing so (or potentially follow a sequence of isogenies from $E_2$).
\end{itemize}
If so, then the knowledge of path assumption holds. Lastly, the invariant inversion assumption could plausibly be true: given an elliptic curve, it seems likely that an adversary cannot construct random coins for the elliptic curve construction procedure that produce the given elliptic curve. 

\paragraph{Instantiating These Ideas.}
While it may seem straightforward to build a candidate quantum lightning scheme using the template above, it unfortunately is not so straightforward.  To start, generating superpositions over $X$, where $X$ is the set of all ellitpic curves with some (even polynomially likely) property seems difficult.  In fact, we do not even know how to generate a uniform superposition over all elliptic curves efficiently.  Recent work~\cite{mula2022random} explains the known approaches to sampling uniform \emph{supersingular} elliptic curves classically, and unfortunately none of it is amenable to sampling a uniform quantum superposition.  Many of the most common ways to sample elliptic curves use random walks, and generating superpositions this way makes it difficult to avoid a hard index erasure problem~\cite{ambainis2011symmetry}.  

There may also be a duality:  the easier the sets of elliptic curves are to sample, the more difficult it is to prove or argue security (since the number of isogenies that are computable may be less or less sophisticated).  We leave it open to future work to instantiate our framework (or something similar) from elliptic curve isogenies, and instead lay out a rough sketch of what such an instantiation might be like below.  To do this, we start by making a conjecture on the efficient samplability of certain elliptic curves.

\begin{conjecture}~\label{conj:iso}
There exists an efficient quantum algorithm $\mathcal{A}$ to sample a uniform superposition over some distribution $\mathcal{E}$ of elliptic curves over some finite field $\F_{p}$ with the following properties:
\begin{itemize}
\item Given two random elliptic curves $E_{1}, E_{2} \in \mathcal{E}$ such that $E_{1}$ and $E_{2}$ are isogenous, there is no efficient (quantum) algorithm to find an isogeny $\psi$ such that $\psi \left(E_{1} \right) = E_{2}$.  This is analogous to the \emph{path-finding assumption} holding.
\item Let $I \left(E_{0} \right)$ denote the number of points on an elliptic curve, and consider an algorithm $I^{-1}: \Z \times \left\{0, 1 \right\}^{\ell}$ that takes an integer $k$ and a random bit string $b$ and outputs an elliptic curve $E_{0}$ with $k$ points.  We require the following two properties (which imply the knowledge of path assumption for invertible invariants, and the inversion inverting assumption be true):
\begin{itemize}
\item Given any algorithm that outputs two elliptic curves $E_{1}$ and $E_{2}$ with $k$ points, there exists an extractor $E$ that can either compute an isogeny $\psi$ such that $\psi \left(E_{1} \right) = E_{2}$ or it can find randomness $t$ and an isogeny $\psi'$ such that either $\psi' \left( I^{-1} \left(k, t \right) \right) = E_{1}$ or $\psi' \left( I^{-1} \left(k, t \right) \right) = E_{2}$.  
\item No efficient adversary can do the following:  sample an (arbitrary) elliptic curve $E_{1}$ with $k$ points, and then, for a randomly selected elliptic curve $E_{2}$, compute randomness $t$ and an isogeny $\psi$ such that $\psi \left( I^{-1} \left(k, t \right) \right) = E_{2}$. 
\end{itemize}
\end{itemize}
\end{conjecture}

We note that our chosen invarant $I$ is the number of points on the elliptic curve.  This is necessary because this is the set of curves that are isogenous.  Note that it is possible, in general, to compute an elliptic curve with a certain number of points~\cite{BroSte04}, so we need to use the invertible invariant form of our quantum lightning construction.  We note that $\mathcal{A}$ could, in theory, restrict the set of curves to supersingular curves.

We emphasize that, while we cannot currently come up with algorithms to satisfy this conjecture, it does not seem like a fundamentally impossible problem to us.  In fact, it might be doable if we knew of an algorithm to sample a uniform superposition of (all) elliptic curves over $\F_{p}$.  If we had such a superposition, we could compute the number of points using Schoof's algorithm~\cite{schoof1995counting} or some related algorithm and store this in an adjacent register.  Then, we could compute a bit that indicates whether or not the number of points on the curve satisfies some property, and then measure this bit.  If $1$, we could continue (and thus have a superposition over curves where the number of points had some property), and if $0$, abort.  If successful, we now have a uniform superposition over elliptic curves where the number of points on the curve satisfies some property.\footnote{We would need this measurement to output $1$ with polynomial probability, which would heavily restrict our properties we could use here.}

We next provide a an instantiation of a quantum lightning scheme from elliptic curve isogenies assuming that Conjecture~\ref{conj:iso} holds.  A construction of quantum lightning from elliptic curve isogenies might look something like the following:

\paragraph{Minting.}  To mint an instance of quantum money/lightning, we would do the following:
\begin{itemize}
\item Sample a uniform superposition over all non-degenerate elliptic curves using the (conjectured) sampling algorithm $\mathcal{A}$.   
\item Measure the number of points on the curve (i.e. compute the invariant $I$).  This becomes the serial number of the note.
\end{itemize}

\paragraph{Verification.}  Our verification procedure would be very similar to as described in the walkable invariant construction.
\begin{itemize}
\item To verify, we would apply isogenies that induce a random walk on the isogeny graph.  We simply need to apply enough isogenies so that the graph ``mixes''.
\item Then we simply apply the check as described in the walkable invariant construction to ensure that we still have a uniform superposition over the orbits.  
\end{itemize}

We do not know if it is even possible to instantiate such a scheme, as it would likely require new ideas in isogeny-based cryptography.  However, we think it is an enticing direction for future research.  We note that, given Conjecture~\ref{conj:iso}, security immediately follows from the security of our invariant money construction.  So, if it is possible to come up with algorithms that satisfy the conjecture, we can build secure quantum lightning from elliptic curve isogenies.


\section{Functional Encryption-inspired Instantiation}
\label{sec:fe_instantiation}

We present another candidate instantiation of invariant quantum money. We call it a ``functional encryption-inspired'' instantiation because many components are functional encryption-like, but the security properties we look for are very different than what is typically required in functional encryption definitions.

At a high level, we need the following components:
\begin{itemize}
    \item A secret-key functional encryption scheme $\fe$ for general functions with the following additional properties: 
	\begin{itemize}
		\item An \emph{invertible} rerandomization algorithm $\rerand$ that allows for ciphertext rerandomization
    		\item Obliviously sampleable ciphertexts.
	\end{itemize}
    \item A (collision resistant) invariant function $H$ or a family of such functions $\mathcal{H}$.
\end{itemize}

We will specify the properties we need for the above building blocks in detail.

\subsection{Re-randomizable Functional Encryption}
\label{sec:fe_defs}

We first have the basic algorithms for an $\fe$ scheme:

\paragraph{Basic Functional Encryption.}
A secret-key functional encryption scheme $\fe$ consists of the following (basic) algorithms:

\begin{itemize}
    \item $\fe.\setup(1^\lambda) \to (\pp,\msk)$:
     a polynomial time algorithm that takes the security parameter as input and outputs public parameters $\pp$ and a master secret key $\msk$.
    
\item $\fe.\keygen(\msk, f) \to \sk{}_f$:
    a polynomial time algorithm that takes as input the master secret key $\msk$ and a
function description $f$ and outputs a corresponding secret key $\sk{}_f$.
    
    \item $\fe.\enc(\msk, m, r) \to \ct$:
    a polynomial time algorithm that takes the master secret key $\msk$, a message $m$ and randomness $r$,
outputs a ciphertext $\ct$.
    
    \item $\fe.\dec(\sk{}_f, \ct) \to m$:
    a polynomial time algorithm that takes a secret key $\sk{}_f$ and ciphertext encrypting
message $m$ and outputs a result $y$.
\end{itemize}

\medskip

We note that $\setup$ and $\keygen$ may additionally take randomness, but we omit that in our description for simplicity.  To construct quantum money,
we need the following additional properties:

\paragraph{Re-randomization}

We additionally need the following algorithm:

\begin{itemize}
    \item $\fe.\rerand(\pp, \ct_{m,r}, r_\delta) \to \ct_{m,r'}$: takes in public parameters $\pp$, a ciphertext $\ct_{m,r} = \fe.\enc(\msk, m, r)$, a string $r_\delta$; outputs a new ciphertext $\ct_{m,r'}$.
    
\end{itemize}

We also require the $\rerand$ algorithm to be \emph{invertible}: there exists an efficient $\rerand^{-1}$ such that if $\ct_{m,r'} \gets \rerand(\pp, \ct_{m,r}, r_\delta)$, we can compute $\ct_{m,r} \gets \rerand^{-1}(\pp, \ct_{m,r'}, r_\delta)$.

\paragraph{Randomness Recoverability}
In order for the re-randomization algorithm(and its inverse) to be useful in the context of our scheme, we also need an efficient procedure to recover the randomness from the ciphertext:
\begin{itemize}
    \item $\recoverrand(\msk, \ct_{m,r}) \to r$: a deterministic algorithm takes in the the master secret key and a ciphertext $\ct_{m,r}$; outputs the randomness $r$.
\end{itemize}
Note that in our scheme, it suffices to let $\recoverrand$ output both the message $m$ and randomness $r$. Therefore, we can just let $\recoverrand(\msk, \cdot)$ and the decryption using the master secret key be the same algorithm.

\paragraph{Obliviously Sampleable Ciphertexts.}

There exists a bijective function $G:\{0,1\}^{n+\ell} \to \mathcal{C}$ where $\mathcal{C}$ is the space of all possible ciphertexts, $n$ is the length of message and $\ell$ is the length of the randomness used in encryption. We also give out the inverse function $G^{-1}: \mathcal{C} \to \{0,1\}^{n+\ell}$.

Note that this function $G$ is independent of the encryption function $\fe.\enc(\msk,\cdot)$, and, informally speaking, $G$ should not be ``useful'' to any adversary attempting to attack the scheme.   One simple example of $G$ is the identity function, when the ciphertexts are all possible strings in $\{0,1\}^{n+m}$. However, we define $G$ more generally since it may not be the case that the ciphertext space is dense.

\begin{remark}
An alternative requirement to obliviously sampleable ciphertexts is to require that the encryptions of random messages are statiscally close to uniform random strings.  We note that this is essentially equivalent to $G$ being the identity function.
\end{remark}

\paragraph{Correctness.}
A functional encryption scheme is correct for a family of functions $\mathcal{F}$ if for all deterministic $f \in \mathcal{F}$, all messages $m \in \mathcal{M}$, all randomness $r \in \mathcal{R}$:

    \begin{align*}
         & \Pr\left[ \dec(\sk{}_f, \ct) = f(m) \middle| \begin{array}{cc} 
              & (\msk,\pp) \gets \setup(1^\lambda), \\
              & \sk{}_f \gets \keygen(\msk, f) \\
              & \ct \gets \enc(\msk, m, r)
         \end{array}\right] \geq 1-\negl(\lambda)  
     \end{align*} 

Note that the above correctness should also hold for re-randomized ciphertexts: that is, for all $f$, all messages $m \in \mathcal{M}$, all randomness $r, r_\delta \in \mathcal{R}$:

    \begin{align*}
         & \Pr\left[ \dec(\sk{}_f, \ct') = f(m) \middle| \begin{array}{cc} 
              & (\msk, \pp) \gets \setup(1^\lambda, k), \\
              & \sk{}_f \gets \keygen(\msk, f) \\
              & \ct \gets \enc(\msk, m, r) \\
            &  \ct' \gets \rerand(\pp, \ct, r_\delta)
         \end{array}\right] \geq 1-\negl(\lambda)  
     \end{align*}

Regarding encryption security, we do not require the usual indistinguishability based definition. We instead present the following "Hardness of path finding" definition, analogous to \ref{assump:pathfinding}:

\begin{definition}[(Single-key) Ciphertext  Path-Finding Security Game]
\label{def:fe_pathfinding}

The above $\fe$ scheme is secure if for all functions $f$ and all admissible non-uniform
QPT $\adv_1,\adv_2$ with quantum advice $\{\langle{\psi_\aux}_\lambda\}\rangle_{\lambda \in \N}$, there exists a
 negligible function $\negl(\cdot)$  such that for all $\lambda \in \N$:

 \begin{align*}
         & \Pr\left[ p = \adv_2(1^\lambda, \pp,\ct_1, |\st\rangle)
         \,\middle| \begin{array}{cc} 
              & (\msk,\pp) \gets \setup(1^\lambda)\\
             &  (\ct_{m_0, r_0} |\st\rangle) \gets \adv_1(1^\lambda, \pp, \sk{}_f \gets \keygen(\msk, f)) \\
              & \ct_{m_0,r_1} \gets \enc(\msk, m_0, r_1), r_1 \gets \mathcal{R}, r_1 \neq r_0        \end{array}\right] \leq \negl(\lambda)  
     \end{align*} 

The above $p$ needs to be a path between $\ct_{m_0, r_0}$ and  $\ct_{m_0, r_1}$. That is, a sequence of inputs of the form $(\ct, r_\delta)$ to the rerandomization algorithm $\fe.\rerand(\pp, \cdots)$ that starts from $\ct_{m_0, r_0}$ and ends with
 $\ct_{m_0, r_1}$.

$\adversary$ is admissible if and only if it makes a single key query to the oracle $\keygen(\msk,\cdot)$. 


Note that in the functional key query phase of the above security game, the adversary may attempt send a superposition query to the challenger; but the challenger can measure the register that stores the circuit $f$'s description so that it does not get to query on a superposition of circuits. 

\begin{remark}
Informally speaking, the above game allows the adversary to choose a ciphertext, and then the challenger chooses a ciphertext which is an encryption of the same message as the adversarially chosen ciphertext.  We note that this is a stronger notion of security than if the challenger sampled two random ciphertexts of the same (randomly chosen) message, and that this strengthened notion of security is essential for our path-finding assumption to hold.
\end{remark}

\end{definition}

\paragraph{Knowledge of Path Assumption for Rerandomizable FE}
We also make analogous assumptions to \ref{assump:kop1} on the knowledge of re-randomization path between two random encryptions of the same message:

\begin{assumption}
\label{assump:kop_fe}
Let $\adv$ be a quantum polynomial time adversary $\adv$ that is unitary (in the above sense where there are no measurements except the output registers), given $\pp \gets \setup(1^\lambda), \sk{}_H \gets \keygen(\msk, H)$.

Let $E$ be a quantum polynomial time extractor that is given $\adv$'s final output as well as its final state. Let $(\ct_{1},\ct_{2})$ be the $\adv$'s output, and $p$ be the output of $E$. 
	
Let $B$ be the event that (1) $H(\ct_{1})=H(\ct_{2})$, but (2) $p$ is not a path between $\ct_{1}$ and $\ct_{2}$. In other words, $B$ is the event that $\adv$ outputs two elements with the same invariant but $E$ fails to find a path between them.

The \emph{knowledge of path assumption for rerandomizable FE} is that, for any quantum polynomial time unitary $\adv$, there exists a quantum polynomial time $E$ such that $\Pr[B]$ is negligible.
\end{assumption}

Other variants of the assumption \ref{assump:kop2},\ref{assump:invertinvert} can be made analogously.

\begin{remark}
We can in fact modify the above requirement of $H(\ct_{1})=H(\ct_{2})$ into "$\ct_1$ and $\ct_1$ are encryptions of the same message $m$". As we will see later, the case that they are not of the same message can be ruled out by collision resistance of $H$.
\end{remark}

\subsection{Quantum Money from Re-randomizable FE}

Our candidate quantum money construction from re-randomizable FE follows from the framework of walkable invariants in \ref{sec:qm_scheme_invariant}.  We will describe its instantiation with a rerandomizable FE $\fe$ with the properties described previously, as well as a collision resistant hash function.

\paragraph{Setup.} the setup algorithm runs the $\fe$ setup and samples $H \gets \mathcal{H}$, computes $\sk{}_H \gets \fe.\keygen(\msk, H)$ of collision resistant hash function $H$. Publishes $\fe.\pp, \sk{}_H$.

\paragraph{Minting.}
First, prepare a uniform superposition over all strings of length $n+\ell$, where $n$ is the message length and $\ell$ is the randomness length $\sum_{x \in \{0,1\}^{n+\ell}} | x \rangle$.
\begin{itemize}
    
    \item  Compute the oblivious sampling function G in an output register to obtain $\sum_{x \in \{0,1\}^{n+\ell}} | x \rangle | G(x) \rangle$.
    
    \item Use the inverse of the oblivious sampling function, $G^{-1}$ to remove the input register $\sum_{x \in \{0,1\}^{n+\ell}} | x + G^{-1}(G(x)) \rangle | G(x) \rangle = \sum_{x \in \{0,1\}^{n+\ell}}  | G(x) \rangle$. Now equivalently, we obtain a uniform superposition of all possible ciphertexts $\sum_{m,r}  | \ct_{m,r} \rangle$.
    
    \item Compute $\fe.\dec(\sk{}_H, \cdot)$ coherently on the state above, measure the output register to obtain serial number $y$ and money state:
    \[|P_y\rangle:=\frac{1}{\sqrt{|P_y|}} \sum_{m: H(m) = y;  r}| \ct_{m,r}\rangle\]
    where $P_y$ is the set of pre-images of $y$.
\end{itemize}

\paragraph{Verification. }
The verification procedure is the same as the one described in \ref{sec:qm_scheme_invariant}.  Here, the permutation $\sigma$  is the rerandomization operation $\fe.\rerand(\pp, \cdot, r_\delta)$, specified by randomness $r_\delta$ ($\sigma^{-1}$ corresponds to $\invertre$, accordingly). The "orbit" $O_m$ for message $m$ corresponds to all possible encryptions of $m$.

\medskip

We note that correctness and security should follow immediately from the invariant money scheme if the path-finding assumptions and knowledge of path assumptions hold.

\subsection{Candidate Construction for Re-randomizable FE}
\label{sec:candidate_FE}

We briefly sketch a candidate construction for the FE scheme with the above properties we need.   We leave the construction as a sketch because we do not know how to fully instantiate it, but think a full, concrete instantiation from reasonably trusted assumptions is excellent future work.

\begin{itemize}
    \item The encryption procedure is
a permutation $P(\msk_\eval, \cdot)$ with puncturable secret key $\msk_\eval$.
To show the path-finding security defined in \ref{def:fe_pathfinding}, we (informally) characterize the security of the puncturable permutation we use as follows:
\begin{itemize}
    \item Let us denote the secret key as $\msk_\eval$.
The adversary $\adv$ is allowed to make queries on both forward evaluations $P(\msk_\eval, \cdot)$ as well as inversion evaluations $P^{-1}(\msk_\eval, \cdot)$; it is also given a "suffix" rerandomization program described as the above rerandomization program in the $\fe$ scheme. $\adv$ can then submit a challenge "prefix" $m$. The challenger samples a random $r$ and appends it to $m$\footnote{More specifically, $r$ is sampled through rejection sampling so that it does not collide with any suffix in previously queried $m||r'$ and $P(\msk_\eval, m||r')$ with the challenge prefix $m$. $\adv$ is not allowed to query on the challenge value $\eval_{m,r}$ after the challenge phase. These queries can all be seen as classical for the same reason in definition \ref{def:fe_pathfinding}.},  computes $\eval_{m,r} = P(\msk_\eval, m||r)$ and puncture the key at value $\eval_{m,r}$. It gives $\eval_{m,r}$ and the punctured key $\msk_\eval^*$ to $\adv$. $\adv$ finally outputs a value $v$ and wins if and only if $v  = P^{-1}(\msk_\eval, \eval_{m,r})$. 

\item Note that the usual indistinguishability based security notion(i.e. pseudorandomness of the evaluation at punctured points) does not work here, since $\adv$ knows the evaluation starts with prefix $m$; we therefore 
rely on a search-type security. 
\end{itemize}

\item To encrypt, we apply
$P(\msk_\eval, \cdot)$ on the concatenation of message $m$ and randomness $r$. We consider the master secret key $\msk$ for FE to contain both the (forward) evaluation key of the permutation and the key for inversion.

\item The functional decryption key for a function $f$ is an obfuscated program that hardcodes $\msk_\eval$ and function $f$: on input ciphertext $P(\msk_\eval, m||r)$, it decrypts the ciphertext to obtain $m$ using $\msk_\eval$; then outputs $f(m)$. The functional decryption procedure is hence simply running the obfuscation program on the ciphertext.

\item Randomness recoverability follows from invertibility of the permutation.
The rerandomization algorithm is also an obfuscated program that hardcodes $\msk_\eval$ and takes input ciphertext $P(\msk_\eval, m||r)$ and $r_\delta$: it inverts the input ciphertext to obtain both $m$ and $r$; computes $r' = r \oplus r_\delta$; finally outputs the re-encryption $P(\msk_\eval, m||r')$. This re-randomizarion procedure is clearly invertible.

\item Oblivious sampleability follows from the fact that encryptions of random messages in the above scheme are uniform random strings.  This follows from the fact that $P$ is a permutation.

\item For our CRHF $H$,  we would like a post-quantum CRHF candidate, for example the SIS hash function\footnote{Note that using a collapsing hash function as the invariant in an invariant money scheme does not lead to an attack: intuitively, it only "collapses" the superposition over different orbits, but the superposition of all elements in the same orbit remains hard to clone, which our verification essentially checks.}.

\end{itemize}



Most of the above programs are constructible assuming indistinguishability obfuscation\cite{garg2016candidate,dottling2021universal,boneh2017constrained_invert,boneh2017constraining}. The permutation we need is trickier to handle: one possible notion we can take use is prefix-constrained PRP discussed in \cite{boneh2017constrained_invert} and its construction remains an open problem. Besides, we need the evaluation key to be puncturable. \cite{boneh2017constrained_invert} also pointed out that puncturable PRP is impossible for the usual security notion of pseudorandomness at punctured points. But since we do not require our permutation to have such strong indistinguishability based security, the impossibility does not apply here.


 

\paragraph{Ciphertext Path-Finding security}
Now we base the path-finding security (see definition \ref{def:fe_pathfinding}) on the permutation above. We first consider a weaker security called selective security: the adversary has to commit to the challenge messages $m_0,m_1$ at the beginning of the security game, before seeing the public parameters $pp$. 


\begin{itemize}
    \item After the adversary $\adv_1$ submits $\ct_{m_0,r_0}$ to the reduction and the reduction submits it to an inversion oracle of the permutation challenger to obtain $(m_0, r_0)$.  The reduction then submits $m_0$ as its challenge prefix to the permutation challenger.

 \item The challenger samples $r_1$, append it to $m_1$. It computes a value $P(\msk_\eval,m_0||r_1)$ and puncture the key at this value. Let us call this value $\ct_{m_0, r_1}$ and the punctured key $\msk^* = \puncture(\msk_\eval, \ct_{m_0, r_1})$. 

\item The reduction receives $\ct_{m_0, r_1}$ and the punctured key $\msk^*$. It can therefore prepare the following program's obfuscation as the functional decryption key on the adversary's query $f$:

\begin{itemize}
    \item Input: $\ct$
    \item Harcoded: $\msk^*, \ct_{m_0, r_1}, y = f(m_0)$
    \item If $\ct = \ct_{m_0, r_1}$: output $y$.
    \item Else: compute  $(m||r) \gets P(\msk^*, \ct)$; output $f(m)$.
\end{itemize}

The reduction also sends $\ct_{m_0, r_1}$ as the challenge ciphertext.

\item Suppose the adversary is able to produce a path between $\ct_{m_0, r_0}$ and $\ct_{m_0, r_1}$, then the reduction can use the path to compute $r_1$ and thus knows $m_0||r_1$, which is the evaluation of $P^{-1}(\msk^*, \cdot)$ at the punctured point $\ct_{m_0, r_1}$, while presumably it should not be able 
to compute this value.

\end{itemize}

In the quantum money scheme, we need to give out the functional key before the adversary hands in the forged money states(i.e. the challenge ciphertext) and thus we need adaptive security instead of selective security. This can be achieved through complexity leveraging (the reduction needs to guess the correct $m_0$; guessing $r_0$ is not necessary) with subexponential security assumption, as it is conventionally dealt with for FE \cite{boneh2004ibe}.




\paragraph{Quantum Lightning Security.}
We discuss the security of the above quantum lightning scheme based on functional encryption.  

First, the verification correctness is satisfied since the re-randomization procedure is an invertible permutation on the ciphertexts of the same message indexed by the randomness, and therefore \ref{sec:verification_main} applies.

For security: suppose there is an adversary that produces two valid money states with the same serial number and suppose measuring both of the money states in computational basis, there are possible two events for the money states it produced:
\begin{itemize}
    \item Event 1: both states produce two ciphertexts $\ct_{m,r_1}, \ct_{m,r_2}$, which are encryptions of the same message, with all but negligible probability.
    
    \item Event 2: with non-negligible probability, two money states give encryptions of different messages, $\ct_{m_1,r_1}, \ct_{m_2, r_2}$, where $H(m_1) = H(m_2)$.
\end{itemize}
First, we can observe  the probability that Event 2 happens is negligible: otherwise such an adversary will help break the collision resistance of $H$ (this case can also be covered by the knowledge-of-path assumption \ref{assump:kop_fe} if we do not require the two outputs $(\ct_1,\ct_2)$ by $\adv$ to be of the same message. But we can rule it out completely here by collision resistance).

To rule out Event 1, we first take use of the hardness of ciphertext path-finding security \ref{def:fe_pathfinding} that we have shown.
Then combining with the knowledge-of-path assumptions defined in \ref{assump:kop_fe}, we would be able to argue that Event 1 happens with negligible probability, just as shown in section \ref{sec:security_from_kop_proof}, thus ruling out any quantum lightning adversary. The proof would be largely identical so we omit it here.

Ideally we would want to rule out Event 1 without the knowledge-of-path assumptions.
 However, such a proof is likely to be beyond the power of existing functional encryption or even indistinguishability obfuscation techniques.

\section{Instantiation from Classical Oracles and Group Actions}
\label{sec:groupaction_instantiation}

We present yet another instantiation of invariant quantum money.  In this case, we present a \emph{classical} oracle-based scheme, whereas some previous lightning schemes are based on quantum oracles. 

This scheme provides an alternative view on both of our isogenies over elliptic curve instantiation and our functional encryption instantiation: group actions can be viewed as an abstraction for certain isogeny-based cryptography, and the oracles we will use are abstractions of the algorithms in the functional encryption scheme.

We show that any adversary that can break this scheme can be used to solve the discrete logarithm problem over (quantum-accessible) generic group actions, assuming the knowledge of path property.




\subsection{Preliminaries: Cryptographic Group Actions}
\label{prelim:groupaction}

First and foremost we  provide some preliminaries for group actions.

We define cryptographic group actions following Alamati \emph{et al.}~\cite{AC:ADMP20}, which are based on those of Brassard and Yung~\cite{C:BraYun90} and Couveignes~\cite{EPRINT:Couveignes06}.  Our presentation is borrowed from~\cite{cryptoeprint:2022/1135}.  

\begin{definition}{\textbf{\textnormal{(Group Action)}}}
	A group $G$ is said to \emph{act on} a set $X$ if there is a map $\star  : G \times X \to X$ that satisfies the following two properties:
	
	\begin{enumerate}
		\item {Identity:} If $e$ is the identity of $G$, then $\forall x\in X$, we have $e \star  x = x$.
		\item {Compatibility:} For any $g,h\in G$ and any $x\in X$, we have $(g h) \star  x=g \star  (h \star  x)$.
	\end{enumerate}
\end{definition}
\noindent We may use the abbreviated notation $(G,X, \star )$ to denote a group action.  We extensively consider group actions that are \emph{regular}:

\begin{definition}
A group action $\left(G, X, \star \right)$ is said to be \emph{regular} if, for every $x_{1}, x_{2} \in X$, there exists a \emph{unique} $g \in G$ such that $x_{2} = g \star x_{1}$.
\end{definition}

We emphasize that most results in group action-based cryptography have focused on regular actions. As emphasized by~\cite{AC:ADMP20}, if a group action is regular, then for any $x\in X$, the map $f_x:g\mapsto g \star  x$ defines a bijection between $G$ and $X$; in particular, if $G$ (or $X$) is finite, then we must have $|G|= |X|$. 

In this paper, unless we specify otherwise, we will work with \emph{effective} group actions (EGAs).  An effective group action $\left(G, X, \star \right)$ is, informally speaking, a group action where all of the (well-defined) group operations and group action operations are efficiently computable, there are efficient ways to sample random group elements, and set elements have unique representation. 

In this work we will also use the \emph{group action discrete logarithm} problem.

\begin{definition}{\textbf{\textnormal{(Group Action Discrete Logarithm)}}}~\label{def:GADL}
	Given a group action $(G,X, \star )$ and distributions $(\mathcal{D}_X,\mathcal{D}_G)$, the group action discrete logarithm problem is defined as follows:  sample $g \leftarrow \mathcal{D}_{G}$ and $x \leftarrow \mathcal{D}_{X}$, compute $y = g \star x$, and create the tuple $T = \left(x, y \right)$.  We say that an adversary solves the group action discrete log problem if, given $T$ and a description of the group action and sampling algorithms, the adversary outputs $g$. 
\end{definition}

\subsubsection{A Generic Group Action Framework}~
\label{sec:GGA}
In this section, we present the generic group action framework from \cite{cryptoeprint:2022/1135}. Their framework is based on the generic group framework of Shoup~\cite{EC:Shoup97}. The following is taken mostly verbatim from \cite{cryptoeprint:2022/1135}.

Let $G$ be a group of order $n$, let $X$ be a set that is representable by bit strings of length $m$, and let $\left(G, X, \star \right)$ be a group action.  We define additional sets $S_{G}$ and $S_{X}$ such that they have cardinality of at least $n$ and $2^m$, respectively.  We define \emph{encoding functions} of $\sigma_{G}$ and $\sigma_{X}$ on $S_{G}$ and $S_{X}$, respectively, to be injective maps of the form $\sigma_{G} : G \rightarrow S_{G}$ and $\sigma_{X}:  X \rightarrow S_{X}$.

A generic algorithm $\mathcal{A}$ for $\left(G, X, \star \right)$ on $\left(S_{G}, S_{X} \right)$ is a probabilistic algorithm that behaves in the following way.  It takes as input two \emph{encoding lists} $\left( \sigma_{G} \left(g_{1} \right), ... , \sigma_{G} \left(g_{k} \right) \right)$ and $\left(\sigma_{X} \left( x_{1} \right), ... , \sigma_{X} \left(x_{k'} \right) \right)$ where each $g_{i} \in G$ and $x_{i} \in X$ and where $\sigma_{G}$ and $\sigma_{X}$ are encoding functions of $G$ on $S_{G}$ and $X$ on $S_{X}$, respectively.  As the algorithm executes, it may consult two oracles, $\mathcal{O}_{G}$ and $\mathcal{O}_{X}$.

The oracle $\mathcal{O}_{G}$ takes as input two strings $y, z$ representing group elements and a sign ``+'' or ``--'', computes $\sigma_{G} \left(\sigma_G^{-1}(y) \pm \sigma_G^{-1}(z) \right)$. The oracle $\mathcal{O}_{X}$ takes as input a string $y$ representing a group element and string $z$ representing a set element, and computes $\sigma_{X} \left( \sigma_G^{-1}(y) \star \sigma_X^{-1}(z) \right)$. As is typical in the literature, we can force all queries to be on either the initial encoding lists or the results of previous queries by making the string length $m$ very long. We typically measure the running time of the algorithm by the number of oracle queries.

\medskip

It is also possible extend the generic group action model to the quantum setting, where we allow \emph{quantum} queries to the oracles. We model quantum queries in the usual way: $\mathcal{O}_{G}\sum_{y,z,\pm,w}\alpha_{y,z,\pm,w}|y,z,\pm,w\rangle=\sum_{y,z,\pm,w}\alpha_{y,z,\pm,w}|y,z,\pm,w\oplus\mathcal{O}_{G}(y,z,\pm)\rangle$ and $\mathcal{O}_X\sum_{y,z,w}\alpha_{y,z,w}|y,z,w\rangle=\sum_{y,z,w}\alpha_{y,z,w}|y,z,w\oplus \mathcal{O}_X(y,z)\rangle$.

\subsubsection{Post-Quantum Instatiation of Group Actions}


Some isogeny based group actions are CSIDH \cite{castryck2018csidh},
CSI-FiSh \cite{beullens2019csi} as well as its derivatives/applications \cite{de2020threshold}. 
We refer the readers to \cite{alamati2020cryptographic} for a detailed discussion on the classification of various isogeny protocols into group action
definitions.

Some other group actions as post-quantum candidates are not isogeny-based, for example
\cite{ji2019general}.


\begin{remark}
A few very recent works 
\cite{cryptoeprint:2022/975,cryptoeprint:2022/1026,robert2022breaking}
break SIDH by showing how to solve
the discrete log problem. However, the attack exploits certain extra points that are made
public in SIDH, and these points are exactly one of the reasons that SIDH is \emph{not} a group action.
In particular, the attack does not seem to apply to CSI-FISH or CSIDH, the main instantiations
of group actions.
\end{remark}




\subsection{The Oracles}
We begin by defining some oracles that we will use.

 Let $n_{0}$, $n_{1},  n_{2}$ be positive integers. 
Let $G$ be a prime-order cyclic group such that $|G| >> n_{2}$, and suppose we have a generic \emph{regular} group action defined by $\left(G, X, \star \right)$.  We also require $n_{0} >> n_{1} + |G|$.

We work in the Generic Group Action framework defined in the previous section ~\ref{sec:GGA}. We will sometimes assume that the oracles $\cO_G, \cO_X$ are given implicitly. 


To implement our invariant function. we first need a collision-resistant function(or a random oracle):
\[
H: \{0,1\}^{n_1} \rightarrow \left\{0, 1\right\}^{n_{2}}
\]

With this in mind, we can define other functions as well.  We define a set of four functions as follows.  These will be our ``setup'' and ``utility'' oracles. We need $|S| = 2^{n_0}$; $P,Q$ are bijective funtions (we will generally model these as random oracles subject to certain constraints. ).
\[
Q:  \left\{0, 1\right\}^{n_{0}}  \rightarrow S\]
\[
Q^{-1} :  S \rightarrow \left\{0, 1\right\}^{n_{0}}
\]
\[
P:  \left\{0, 1\right\}^{n_{1}} \times X  \rightarrow S
\]
\[
P^{-1} :  S  \rightarrow \left\{0, 1\right\}^{n_{1}} \times X
\]
We include a few extra restrictions.  We first require that, for all tuples of bit strings $\xv \in \left\{ 0, 1 \right\}^{m}$ and elements $g \in G$, we have $Q^{-1} \left(Q \left( \xv, g \right) \right)$ and $P^{-1} \left(P \left( \xv , g\right) \right)$.  Moreover, we require that the \emph{output sets} of $P$ and $Q$ be \emph{identical}, although the mappings themselves are random subject to this constraint (and thus, almost certainly different).

We next define our invariant oracle as follows.  Let $S$ be the set of bit strings in the output space of $P$ and $Q$.
\[
I:  S \rightarrow \left\{0, 1\right\}^{n_{2}}
\]
Note that we are assuming that $I$ takes as input an element $e_{\xv,x} = P(\xv, x)$ of $S$ and returns the invariant $H(\xv)$ computed in an output register. Note that we can instantiate oracle $I$ as follows: first apply $P^{-1}$ on $e_{\xv,x}$ to get back to $\xv||x$, then apply the function $H$ on $\xv$. Since $I$ can only be accessed as an oracle, functionality of $P^{-1}$ is not given out.


Finally, we must define our ``random walking'' function $R$:
\[
R:  S \times G \rightarrow S
\]

We define $R$'s functionality as: $\forall x \in X, \xv \in \{0,1\}^{n_1}, g \in G: R \left( P \left( \xv, x \right), g \right) = P \left( \xv, g \star x \right)$.  $R$ can be implemented using $P^{-1}$ on element $e_{\xv,x} = P(\xv,x)$ to recover $(\xv,x)$; then it calls the group action oracle $\cO_G$ to apply $g \star x$; finally it outputs the result by applying $P$ on $(\xv, g \star x)$. Naturally, to apply a reverse random walk, one would use $g^{-1}$.


\subsection{Quantum Money from Oracles and Group Actions}

As before, the quantum money(lightning) scheme in this section follows from the framework of walkable invariants in \ref{sec:qm_scheme_invariant}.  We will define the scheme and then prove that it instantiates a secure walkable invariant.

\paragraph{Setup} the setup algorithm generates all of the oracles defined in the previous section.  It then publishes $Q$, $Q^{-1}$, $H$, $I$ and $R$, as well as a description of all the parameters and the relevant oracles for the generic group action $\left(G, X, \star \right)$.
Note that $P, P^{-1}$ are kept secret.

\paragraph{Minting}
First, prepare a uniform superposition over all strings of length $n_{0}$,  giving us the following:  $\sum_{s \in \{0,1\}^{n_{0}}} | s \rangle$.  
\begin{itemize}
    
    \item  Compute the oblivious sampling function $Q$ in an output register to obtain $\sum_{s \in \{0,1\}^{n_{0}}} | s \rangle | Q(s) \rangle$.
    
    \item Use the inverse of the oblivious sampling function, $Q^{-1}$ to remove the input registers $\sum_{s \in \{0,1\}^{n_{0}}} | s \rangle$.  This allows us to obtain a uniform superposition of all possible values  $\sum_{s \in \{0,1\}^{n_{0}}}| Q(s) \rangle$. Equivalently, we have obtained $\sum_{e \in S}| e \rangle = \sum_{\xv \in \{0,1\}^{n_1},  x \in X}| P \left(\xv , x \right) \rangle$.
    
    \item Compute the invariant function $I$ coherently on the state above, and then measure the output register to obtain serial number $y$ and money state:
\[|M_y\rangle:=\frac{1}{\sqrt{|M_y|}} \sum_{\xv : H(\xv ) = y;  x}| P \left(\xv , x \right) \rangle\]
    where $M_y$ is the set of pre-images of $y$.
\end{itemize}


\paragraph{Verification}
The verification procedure is straightforward and the same as the one in \ref{sec:qm_scheme_invariant}.  Note that this is even simpler than defined above because we required $G$ to be a prime-order cyclic group.

Here, the permutation $\sigma$ is the rerandomization oracle $R$ specified by randomness $g \in G$. The "orbit" $O_{\xv}$ corresponds to all possible evaluations of $P$ on a fixed bit string $\xv$.

\subsection{Quantum Lightning Security}


We also have a reasonably reliable hard problem to reduce our path-finding hardness to, namely, the group action discrete log problem.

To prove security, we will first show that the Path-Finding hardness holds in the GGA framework.  

\begin{definition}[Path-Finding Game for Group Actions]
\label{assump:pathfinding_group} Consider an adversary $\adv$ playing the following game: 
	\begin{itemize}
	    \item Give the adversary access to $Q,Q^{-1},R, I$ (and the oracles in GGA framework).
		\item The adversary outputs an $e_{\xv,x} \in S$.
		\item The challenger then computes a random $e_{\xv,z} \in O_{\xv}$, where the  "orbit" $O_{\xv}$ corresponds to the set of all possible evaluations of $P$ on a fixed bit string $\xv$ . 
		\item The adversary wins if it can output a path $p$ between $e_{\xv,x}$ to $e_{\xv,z}$, where a path is a sequence of \emph{classical} queries to oracle $R$. 
	\end{itemize}

\end{definition}

We next show that for all quantum polynomial-time adversaries $\adv$, the probability $\adv$ wins in the above game is negligible.

\begin{lemma}
\label{lem:pathfinding_groupaction}
Any adversary that breaks the Path-Finding game in \ref{assump:pathfinding_group} can be used to solve discrete logarithm on $\left(G, X, \star \right)$ (definition \ref{def:GADL}).
\end{lemma}

\begin{proof}
This is almost immediate due to the construction of our scheme. Suppose there exists an adversary that outputs a path, i.e. a sequence of classical oracle queries made to $R$ to get from $e_{\xv, x}$ to $e_{\xv, z}$, then one can output a sequence of group elements $g_1, \cdots, g_k$ such that $z = (g_k \cdots g_1) \star x$.
\end{proof}


We give a variant for the Knowledge of Path assumption \ref{assump:kop1} in the generic group action framework:

\begin{assumption}\label{assump:kop1_group}Let $\adv$ be a quantum polynomial time adversary $\adv$ that is unitary (in the above sense where there are no measurements except the output registers), given oracle access to $Q,Q^{-1},R, I$.

Let $E$ be a quantum polynomial time extractor that is given $\adv$'s final output as well as its final state. Let $(e_x,e_z)$ be the $\adv$'s output, and $p$ be the output of $E$. 
	
Let $B$ be the event that (1) $I(e_x)=I(e_z)$, but (2) $p$ is not a path between $x$ and $z$. In other words, $B$ is the event that $\adv$ outputs two elements with the same invariant but $E$ fails to find a path between them.

The \emph{knowledge of path assumption for group actions} is that, for any quantum polynomial time unitary $\adv$, there exists a quantum polynomial time $E$ such that $\Pr[B]$ is negligible.
\end{assumption}

The above event $B$ can again be divided into two cases: (1) $\adv$'s outputs $e_x, e_z$ are in different "orbits", i.e. $e_x = P(\xv_1,x)$ and $e_z = P(\xv_2,z)$ for some $\xv_1 \neq \xv_2, H(\xv_1) = H(\xv_2)$; (2)
$e_x, e_z$ are in the same orbit. The first case is ruled out by the collision resistance of $H$. Thus we can  modify the above assumption into requiring that $e_x. e_z$ are in the same orbit.

We can also analogously transform \ref{assump:kop2}, \ref{assump:invertinvert} to the group action version.

Therefore, using the path-finding hardness for group actions shown above \ref{lem:pathfinding_groupaction} together with the above assumptions, the security proof will be almost identical to section \ref{sec:security_from_kop_proof}, ruling out any quantum lightning adversary.

.


\section{Instantiation Using Knots}
\label{sec:knot_instantiation}
The inspiration for our invariant scheme was the construction of quantum money from knots in~\cite{ITCS:FGHLS12}.  In this section, we explain how this construction can be modelled in our framework and what assumptions on knots need to be true in order for our security proof to apply. While it's possible that their verification procedure may not satisfy correctness(the Markov chain does not mix) and the security property is challenging to investigate, we believe an alternative view would provide insight into proving/breaking their scheme.

\paragraph{On Knots.}  Recall that a knot is an embedding of the circle into three-dimensional Euclidean space, or more informally, as the authors of ~\cite{ITCS:FGHLS12} nicely put it, a loop of string in three dimensions.  Informally speaking, two knots are considered to be \emph{equivalent} if they can be morphed into each other without ``cutting the string''.  We say that a function is a \emph{knot invariant} if it has the same value on all equivalent knots.

Knots (of a certain maximum size) can be represented by \emph{planar grid diagrams}, which are $d \times d$ grids containing exactly $d$ Xs and $d$Os such that each row and column of the grid have exactly one X and one O, and there is no space in the grid with both an X and an O.  Alternatively, knots can be represented by a tuple of two disjoint permutations of size $d$.  We note that knots with smaller numbers of ``elements'' (i.e. with less than $d$ Xs and Os) can be gracefully represented on $d \times d$ grids by leaving certain columns and rows blank.

\emph{Reidemeister moves}~\cite{reidemeister1927elementare} are three types of local moves that can be applied to a planar grid diagram with the properties that any two knots that are reachable from each other by Reidemeister moves are equivalent and that any two equivalent knots must be reachable by the application of a finite number of Reidemeister moves (although the best bound on the minimum number of moves is enormous~\cite{coward2014upper}).

It will be useful to think of Reidemeister moves as a set $R$ acting on the set of planar grid diagrams of size $d$ $S_{d}$.  While each move $r \in R$ has an inverse, it is unfortunately impossible to model this interaction as a group action because the application of Reidemeister moves may not be associative.  

We defer a full description of knots, Reidemeister moves, and associated concepts to~\cite{ITCS:FGHLS12}; we will not need the full formalism to convey our ideas here.

\subsection{A General Description of the~\cite{ITCS:FGHLS12} Scheme}
We next provide a general description of the quantum money from knots scheme in~\cite{ITCS:FGHLS12}.  We omit some of the details of the scheme (some of which are important for its security) so we can present the scheme in its most general form.

\paragraph{Minting.}  To mint a note, start by constructing a \emph{specific} superposition over grid diagrams of size $d$, which we refer to as $| S_{d} \rangle$ over $S_{d}$.  In particular, a uniform distribution over all possible knots represented by $d \times d$ planar grid diagrams is \emph{not} chosen for security reasons.  Instead, knots are weighted in the superposition based on the number of ``elements'' present in the grid diagram representation according to a Gaussian distribution centered on $\frac{d}{2}$. 

After generating $| S_{d} \rangle$, compute a knot invariant $A$ on the superposition $| S_{d} \rangle$, store it in an adjacent register, and then measure it, getting some value $y$.  The serial number is the value of this invariant, and the money state becomes $\frac{1}{\sqrt{|\sum_{A \left( S_{d} \right) = y} |}} \sum_{A \left( S_{d} \right) = y} | S_{d} \rangle$.  In~\cite{ITCS:FGHLS12}, the authors use the Alexander polynomial as the knot invariant of choice.

\paragraph{Verification.}  To verify, a quantum verification procedure based on a classical Markov chain is applied.  Essentially, many randomly chosen Reidemeister moves are applied to the money state superposition with the restriction that Reidemeister moves that would expand the planar grid diagram beyond $d$ dimensions are ignored.  This (somewhat) simulates a (not necessarily uniform) random walk on the graph of equivalent knots.

In other words, the verification procedure is almost identical to our process described in section~\ref{sec:invariant_main}.  The set of all Reidemeister moves on $d \times d$ grid diagrams form something very close to a set of permutations $\sigma_{i}:  S_{d} \rightarrow S_{d}$; it turns out to be possible to construct permutations on grid diagrams from Reidemeister moves such that the permutation set enables all possible Reidemeister moves.\footnote{To be precise, we just need to pair Reidemeister moves that increase the number of elements in the grid diagram with a corresponding move that ``undoes'' the transformation, and choose which move we apply based on the number of elements in the grid diagram.} 

\subsection{Correctness Properties}
In order to fit in our quantum money framework, we require two essential properties of a scheme:  efficient generation of superpositions and mixing walks.

\paragraph{Efficient Generation of Superpositions.}  The authors of~\cite{ITCS:FGHLS12} show clearly how to generate the appropriate superpositions.  It is a slightly long presentation so we defer to their paper.

\paragraph{Mixing Walks.}  In the context of our framework, ``orbits'' are knots and ``permutations'' are Reidemeister moves.  Our framework requires uniform superpositions over orbits, while the construction in~\cite{ITCS:FGHLS12} uses non-uniform superpositions.  However, if the mixing process preserves the starting superposition as the authors of~\cite{ITCS:FGHLS12} conjecture, then their knot scheme could be fit into our framework with relatively minor modifications to the framework.

However, a formal proof that mixing occurs would be quite difficult and likely involve solving several longstanding open problems in knot theory.  A uniform mixing process for knots could only work in polynomial time if the number of Reidmeister moves between all equivalent knots of a certain size were bounded.  Unfortunately, the only known bound for the number of Reidemeister moves between arbitrary knots invoves a ``tower of exponentials'' function~\cite{coward2014upper}.  While it is known that the number of Reidemeister moves between unknot\footnote{The unknot is the simplest possible knot:  just a circle.} representations of a certain size is polynomially bounded~\cite{lackenby2015polynomial}, there are no known results for more general knots.

It may be the case that the restriction to $d \times d$ grids may make proving mixing easier, but we could not effectively utilize this fact.

We emphasize that the lack of a good mixing algorithm may not constitute an attack:  it could still be the case that it is hard for an adversary to find a superposition that is unchanged by the existing mixing algorithm.  But these sorts of statements and proofs are outside of our framework.

\subsection{Security Properties}
If we want to build secure quantum lightning from knots using our framework, we would need to show that two assumptions are true:  the hardness of path-finding assumption and the knowledge of path assumption.

\paragraph{Hardness of Path-finding.}  It has long been conjectured that the \emph{knot equivalence problem}--in other words, distinguishing whether or not two knots are equivalent or not--is hard~\cite{dynnikov2003recognition}.  The hardness of path assumption in this scheme would correspond to actually finding the Reidemeister move set necessary to transform one knot into another equivalent knot, and any efficient algorithm for this would immediately imply an efficient algorithm for the knot recognition problem.  It is known that determining if two knots are equivalent or not is decidable~\cite{lackenby2016elementary}, but it is not even known if the problem is in NP.  Presumably if there were polynomial-length numbers of moves between equivalent knots of certain sizes, the problem would be in NP, but this is not known.

On the other hand, Lackenby~\cite{lackenbyannounce} has recently announced a quasipolynomial algorithm for recognizing the unknot.  If this could be generalized to other knot equivalences (which may or may not be possible), it would spell trouble for basing cryptographic primitives on the hardness of knot equivalence.

\paragraph{Knowledge of Path Assumption.}  It is more difficult to assess the knowledge of path assumption over knots.  Intuitively, this asks whether or not it is possible to create two equivalent knots without knowing a path between them (in particular, the path must only involve knots that fit on $d \times d$ planar grids).  We certainly do not see an easy way to do this, but to our knowledge no one has ever studied this problem.

\end{document}